\PassOptionsToPackage{dvipsnames,svgnames}{xcolor}		% for xcolor conflict
\PassOptionsToPackage{inline,shortlabels}{enumitem}		% for enumitem conflict
\PassOptionsToPackage{numbers,sort&compress}{natbib}		% for natbib conflict
\documentclass[final]{colt2021}

\usepackage{amsmath}		% for AMS macros
\usepackage{amssymb}		% for AMS symbols
\usepackage{amsfonts}		% for AMS fonts
\usepackage{amsthm}		    % for theorems

\usepackage{mathtools}		% for advanced math
\mathtoolsset{%
}

\usepackage[utf8]{inputenc}		% for source encoding
\usepackage[T1]{fontenc}		% for font encoding

\usepackage[%		% for math font selection
cal=cm,
]
{mathalfa}

\usepackage{dsfont}		% for blackboard bold font
\usepackage[scale=.96]{inconsolata}

\usepackage{comment}

\usepackage{times}
\usepackage{acronym}		% for acronyms
\newcommand{\acli}[1]{\emph{\acl{#1}}}		% for italicized acro
\newcommand{\acdef}[1]{\emph{\acl{#1}} \textup{(\acs{#1})}\acused{#1}}		% for acro def
\usepackage[labelfont={bf,small},labelsep=colon,font=small, singlelinecheck=false]{caption}	% for caption control
\colorlet{MyBlue}{MediumBlue}
\colorlet{MyGreen}{DarkGreen!85!Black}
\usepackage{wasysym}		% for extra symbols

\usepackage{subcaption}		% for subfigures
\usepackage{wrapfig}    % for wraping text around figures
\usepackage{tikz}		% for figures
\usetikzlibrary{calc,patterns}

\usepackage{array}		% for flexible tables
\usepackage{booktabs}		% for professional tables
\usepackage[inline,shortlabels]{enumitem}		% for list control (conflicts with COLT)
\usepackage{multirow}
\setenumerate{itemsep=\smallskipamount,topsep=\medskipamount}		% for list controls

\usepackage{microtype}		% for microtypography

\usepackage{xspace}		% for flexible spaces
\usepackage{xparse}     % for two optional arguments

\bibpunct[, ]{[}{]}{,}{n}{,}{,}

\usepackage{hyperref}
\hypersetup{
colorlinks=true,
linktocpage=true,
pdfstartview=FitH,
breaklinks=true,
pdfpagemode=UseNone,
pageanchor=true,
pdfpagemode=UseOutlines,
plainpages=false,
bookmarksnumbered,
bookmarksopen=false,
bookmarksopenlevel=1,
hypertexnames=true,
pdfhighlight=/O,
urlcolor=Maroon,linkcolor=MyBlue!85!black,citecolor=MyBlue!85!black,	% for on-screen
pdftitle={},
pdfauthor={},
pdfsubject={},
pdfkeywords={},
pdfcreator={pdfLaTeX},
pdfproducer={LaTeX with hyperref}
}
\usepackage[sort&compress,capitalize,nameinlink]{cleveref}		% for cleveref formatting
\crefname{assumption}{Assumption}{Assumptions}
\crefname{assumptionloc}{Assumption}{Assumptions}
\usepackage{thmtools}
\usepackage{thm-restate}
\usepackage{thm-autoref}

\theoremstyle{plain}
\newtheorem{theorem}{Theorem}		% for theorems
\newtheorem{lemma}[theorem]{Lemma}		% for lemmas
\newtheorem*{corollary*}{Corollary}		% for corollaries (unnumbered)

\theoremstyle{definition}
\newtheorem{definition}[theorem]{Definition}		% for definitions
\newtheorem{assumption}{Assumption}		% for assumptions
\newtheorem{example}{Example}		% for examples

\newtheorem*{definition*}{Definition}		% for definitions (unnumbered)
\newtheorem*{assumption*}{Assumptions}		% for assumptions (unnumbered)

\makeatletter		% for custom tags
\newcommand{\asmtag}[1]{% \asmtag{<tag>}
  \let\oldtheassumption\theassumption% Store \theassumption
  \renewcommand{\theassumption}{#1}% Redefine it to a fixed value
  \g@addto@macro\endassumption{% At \end{assumption}, ...
    \addtocounter{assumption}{-1}% ...restore assumption counter value and...
    \global\let\theassumption\oldtheassumption}% ...restore \theassumption
  }
\makeatother

\newtheorem*{remark*}{Remark}		% for remarks (unnumbered)

\newcounter{proofpart}

\usepackage[showdeletions]{color-edits}		% for editing macros / use
\addauthor[Yu-Guan]{YGH}{OrangeRed}

\addauthor[Pan]{PM}{Purple}

\addauthor[Kimon]{KA}{RoyalBlue}

\newacro{NI}{Nikaido-Isoda}
\newacro{VI}{variational inequality}
\newacroplural{VI}[VIs]{variational inequalities}
\newacro{SP}{saddle-point}

\newacro{GAN}{generative adversarial network}
\newacro{LHS}{left-hand side}
\newacro{RHS}{right-hand side}
\newacro{iid}[i.i.d.]{independent and identically distributed}
\newacro{usc}[u.s.c.]{upper semi-continuous}
\newacro{lsc}[l.s.c.]{lower semi-continuous}
\newacro{NE}{Nash equilibrium}
\newacroplural{NE}[NE]{Nash equilibria}
\newacro{CCE}{coarse correlated equilibrium}
\newacroplural{CCE}[CCE]{coarse correlated equilibria}

\newacro{VS}{variationally stable}
\newacro{SVS}{strictly variationally stable}

\newacro{RVU}{Regret bounded by Variations in Utilities}

\newacro{DA}{dual averaging}
\newacro{MD}{mirror descent}
\newacro{PEG}{past extra-gradient}
\newacro{MWU}{multiplicative weights update}
\newacro{OMWU}{optimistic multiplicative weights update}
\newacro{OMP}{optimistic mirror-prox}

\newacro{FTRL}{``follow the regularized leader''}
\newacro{OMD}{online mirror descent}
\newacro{OptMD}{optimistic mirror descent}
\newacro{OptDA}{optimistic dual averaging}
\newacro{OptFTRL}{optimistic follow the regularized leader}
\newacro{DS-OptMD}{dual stabilized optimistic mirror descent}

\newacro{OptDE}{optimistic dual extrapolation}
\newacro{OMPPS}{optimistic mirror-prox with primal stabilization}
\newacro{OMPDS}{optimistic mirror-prox with dual stabilization}

\usepackage{macros}

\title[Adaptive Learning in Continuous Games]{Adaptive Learning in Continuous Games:\\ Optimal Regret Bounds and Convergence to Nash Equilibrium}

\coltauthor{%
 \Name{Yu-Guan Hsieh} \Email{yu-guan.hsieh@univ-grenoble-alpes.fr}\\
 \addr Univ. Grenoble Alpes, Inria, LJK, 38000 Grenoble, France
 \AND
 \Name{Kimon Antonakopoulos} \Email{kimon.antonakopoulos@inria.fr}\\
 \addr Univ. Grenoble Alpes, CNRS, Inria, Grenoble INP, LIG, 38000 Grenoble, France
 \AND
 \Name{Panayotis Mertikopoulos} \Email{panayotis.mertikopoulos@imag.fr}\\
 \addr Univ. Grenoble Alpes, CNRS, Inria, Grenoble INP, LIG, 38000 Grenoble, France, \&
 Criteo AI Lab
}

\begin{document}
\maketitle

\begin{abstract}
%----------------------------------------------------------------------
%%% ABSTRACT
%----------------------------------------------------------------------
% !TEX root = ../Main.tex

In game-theoretic learning, several agents are simultaneously following their individual interests, so the environment is non-stationary from each player's perspective.
In this context, the performance of a learning algorithm is often measured by its regret.
%In this regard, the performance of a learning algorithm is often measured by regret
%, which evaluates how much the player could have gained by switching to any other fixed strategy.
However, no-regret algorithms are not created equal in terms of game-theoretic guarantees:
depending on how they are tuned, some of them may drive the system to an equilibrium, while others could produce cyclic, chaotic, or otherwise divergent trajectories.
To account for this, we propose a range of no-regret policies based on optimistic mirror descent, with the following desirable properties:
\begin{enumerate*}
[\itshape i\hspace*{.5pt}\upshape)]
\item
they do not require \emph{any} prior tuning or knowledge of the game;
\item
they all achieve $\bigoh(\sqrt{T})$ regret against arbitrary, adversarial opponents;
and
\item
they converge to the best response against convergent opponents.
Also, if employed by all players, then
\item
they guarantee $\bigoh(1)$ \emph{social} regret;
while
\item
the induced sequence of play converges to \acl{NE} with $\bigoh(1)$ \emph{individual} regret in all variationally stable games (a class of games that includes all monotone and convex-concave zero-sum games).
\end{enumerate*}
\end{abstract}

%\begin{keywords}%
%  \YGHcomment{TODO}
%\end{keywords}

%----------------------------------------------------------------------
%%% INTRODUCTION
%----------------------------------------------------------------------
\section{Introduction}
%----------------------------------------------------------------------
%%% INTRODUCTION
%----------------------------------------------------------------------
% !TEX root = ../Main.tex

A fundamental problem at the interface of game theory and online learning concerns the exact interplay between static and dynamic solution concepts.
On the one hand,
%from a top-down, \emph{deductive} viewpoint,
if all players know the game and are assumed to be rational, the most relevant solution concept is that of a \acli{NE}:
this represents a stationary state from which no player has an incentive to deviate.
On the other hand,
%from a bottom-up, \emph{inductive} viewpoint,
this knowledge is often unavailable, so players must adapt to each other's actions in a dynamic manner;
in this case, the standard figure of merit is an agent's \emph{regret}, \ie the cumulative difference in performance between an agent's trajectory of play and the best action in hindsight.
Optimistically, one would expect that the two approaches should yield compatible answers \textendash\ and, indeed, one direction is clear:
\acl{NE} never incurs any regret.
Our paper deals with the converse question, namely:
\emph{Does no-regret lead to \acl{NE}?}

This question has attracted considerable interest in the literature and the answer can be particularly nuanced.
To provide some context, it is well known that the empirical frequency of no-regret play in \emph{finite} games converges to the set of \acfp{CCE} \textendash\ also known as the game's \emph{Hannan set} \cite{Han57,HMC00}.
This is sometimes interpreted as a ``universal equilibrium convergence'' result, but one needs to keep in mind that
\begin{enumerate*}
[\itshape a\upshape)]
\item
the type of convergence involved is \emph{not} the actual, day-to-day play but the empirical frequency of play;
and
\item
the game's \acp{CCE} set may contain elements that fail even the most basic rationalizability axioms.
\end{enumerate*} 
In particular, \citet{VZ13} constructed a simple two-player game (a variant of rock-paper-scissors with a feeble twin) that admits \acp{CCE} supported \emph{exclusively} on strictly dominated strategies.

This interplay becomes even more involved because the behavior of a no-regret learning algorithm could switch from convergent to non-convergent by a slight variation of its hyperparameters or a small perturbation of the game.
As a simple example, optimistic gradient methods are known to converge to \acl{NE} in smooth convex-concave games, provided that they are tuned appropriately.
However, if the algorithm's step-size is out-of-tune even by a little bit, the trajectory of play could diverge and the players' mean behavior could converge to an irrelevant off-equilibrium profile (we provide a concrete example of this behavior in \cref{sec:OptMD}).
Equally pernicious examples can be found in symmetric $2\times2$ congestion games:
even though such games have a very simple equilibrium structure (a unique, evolutionarily stable equilibrium), running a no-regret learning algorithm \textendash\ like the popular \acl{MWU} scheme \textendash\ may lead to chaos  \cite{PPP17,CFMP20,CFMP20-NIPS}.

%----------------------------------------------------------------------
%%% Contribs
%----------------------------------------------------------------------
\para{Our contributions and related work}

In view of all this, the equilibrium convergence properties of no-regret learning depend crucially on the algorithm's tuning \textendash\ and the parameters required for this tuning could be beyond the players' reach.
With this in mind, we propose a range of no-regret policies with the following desirable properties:
\begin{enumerate}
[leftmargin=\parindent,itemsep=0pt]
\item
They do not require \emph{any} prior tuning or knowledge of the game's parameters:
each player updates their individual step-size with purely local, individual gradient information.
\item
They guarantee an order-optimal $\bigoh(\sqrt{\nRuns})$ regret bound against adversarial play, and they further enjoy \emph{constant} social regret when all players employ one of these algorithms.
\item
In any continuous game with smooth, convex losses, the sequence of chosen actions of any player converges to best response against convergent opponents.
\item
If all players follow one of these algorithms, the induced trajectory of play converges to \acl{NE} and the individual regret of each player is bounded as $\bigoh(1)$ in all variationally stable games \textendash\ a large class of games that contains as special cases all convex-concave zero-sum games and monotone / diagonally convex games.
\end{enumerate}
To the best of our knowledge, the proposed methods \textendash\ \acdef{OptDA} and \acdef{DS-OptMD} \textendash\ are the first in the literature that concurrently enjoy even a subset of these properties in games with continuous action spaces.
To achieve this, they rely on two principal ingredients:
\begin{enumerate*}
[\itshape a\upshape)]
\item
a regularization mechanism as in the popular \ac{FTRL} class of policies \cite{SS11,BCB12};
and
\item
a player-specific adaptive step-size rule inspired by \cite{RS13-NIPS}.
\end{enumerate*}
In this regard, they resemble the policy employed by \citet{SALS15} who established comparable individual/social regret guarantees for \emph{finite} games.
Our results extend the analysis of \citet{SALS15} to games with \emph{continuous} \textendash\ and possibly \emph{unbounded} \textendash\ action spaces, and, as a pleasing after-effect, they also shave off all logarithmic factors.

%\begin{enumerate}
%    \item Adapt to adversarial opponent and self-play: \cite{KHSC18,RS13-NIPS,SALS15}.
%    \item Last-iterate convergence of mirror-prox type algorithms: \cite{DP19,MLZF+19,WLZL21}. Convergence to Nash with adaptive stepsize \cite{LZMJ20}.
%    \item Learning criteria (older papers, mostly assume that the game is known and Nash can be computed in advance): convergence to best response against stationary opponent and converge to Nash against self play \cite{CS07,PS04}, with additionally no-regret property but with infinitesimal update \cite{BP04}, no-regret and convergence only in two-player two-action zero-sum games \cite{Bow04}, the NNR criterion is of interest.
%    \item Universal
%\end{enumerate}

Concerning the convergence behavior of \ac{OptMD},
it is known that the sequence of realized actions converges to a \acl{NE} in all variationally stable games,
provided that every player runs the algorithm with a sufficiently large regularization parameter, common across all players \citep{MLZF+19,HIMM20}.%
\footnote{Strictly speaking, \cite{MLZF+19} analyzes the Mirror-Prox algorithm, but the same arguments apply to \ac{OptMD}.
On the other hand, several other papers have focused on obtaining convergence rate of \ac{OptMD} in more specific settings \cite{HIMM19,LS19,WLZL21}.}
This result cannot be attained by ``vanilla'' first-order methods that do not include an extra-gradient mechanism, but it also comes with several important caveats.
First, running \ac{OptMD} with a constant step-size robs the algorithm of any fallback guarantees:
a player's individual regret may grow linearly if the other players switch to an adversarial behavior (\eg as part of a ``grim trigger'' strategy).
Second, the method's convergence is contingent on the players' using a fine-tuned regularization parameter, depending on the smoothness modulus of their payoff functions.
This constant cannot be estimated without prior, \emph{global} knowledge of the game's primitives, and if a player misestimates it, the algorithm's convergence breaks down completely (see \cref{fig:constant-bilinear-diverge} in \cref{sec:OptMD}).

In terms of trajectory convergence of adaptive methods, the closest antecedents of our results are the recent papers by \citet{LZMJ20} and \citet{ABM19,ABM21}, where the authors propose an adaptive step-size rule for cocoercive games and \aclp{VI} respectively.
However, in both cases, the method's step-size requires \emph{global} gradient information, and therefore does not apply to a fully distributed game-theoretic setting.

%----------------------------------------------------------------------
%%% SETUP
%----------------------------------------------------------------------
\section{Online learning in games}
\label{sec:setup}
%----------------------------------------------------------------------
%%% SETUP
%----------------------------------------------------------------------
% !TEX root = ../Main.tex

In this section, we present the necessary background material on normal form games with continuous action spaces and the corresponding learning framework.

%----------------------------------------------------------------------
%% Games
%----------------------------------------------------------------------
\subsection{Games with continuous action spaces}
\label{subsec:game}

\para{Definitions and examples}

Throughout the paper, we focus on normal form games played by a finite set of players $\players\defeq\{1,\ldots,\nPlayers\}$.
Each player $\play\in\players$ is associated with a closed convex action set $\vw[\points]\subseteq\R^{\vdim_{\play}}$ and a loss function $\vw[\loss]\from\points\to\R$, where $\points\defeq\prod_{\play=1}^{\nPlayers}\vw[\points]$ denotes the game's joint action space.
For the sake of clarity, a joint action of multiple players will be typeset in bold;
in particular, the joint action profile of all players will be written as $\jaction=(\vw[\action], \vw[\jaction][\playexcept])=\allplayers{\vw[\action]}$, where $\vw[\action]$ and $\vw[\jaction][\playexcept]$ respectively denote the action of player $\play$ and the joint action of all players \emph{except} player $\play$.

Our blanket assumption concerning the players' loss functions is the following:

\begin{assumption}[Individual convexity + Smoothness]
\label{asm:convexity+smoothness}
For each $\play\in\players$, $\vw[\loss]$ is continuous in $\jaction$ and convex in $\vw[\action]$ \textendash\  that is, $\vw[\loss](\cdot,\vw[\jaction][\playexcept])$ is convex for all $\vw[\jaction][\playexcept]\in\prod_{\playalt\neq\play}\vw[\points][\playalt]$.
Furthermore, the subdifferential $\subd_{\play}\vw[\loss]$ of $\vw[\loss]$ relative to $\vw[\action]$ admits a Lipschitz continuous selection $\vw[\vecfield]$ on $\points$.
\end{assumption}

In the sequel, we will refer to any game that satisfies \cref{asm:convexity+smoothness} as a (continuous) convex game.
For the sake of concreteness, we briefly discuss below two examples of such games.

\begin{example}[Mixed extensions of finite games]
\label{ex:finite}
In a \emph{finite game}, each player $\play\in\players$ has a finite set $\vw[\pures]$ of \emph{pure strategies} and no assumptions are made on the loss function $\vw[\loss]\from\prod_{\play=1}^{\nPlayers}\vw[\pures]\to\R$.
A \emph{mixed strategy} for player $\play$ is a probability distribution $\vw[\action]$ over their pure strategies, so the player plays $\indg$ with probability $\vw[\action_{\indg}]$ (\ie the $\indg$-th coordinate of $\vw[\action]$).%
\footnote{In a slight abuse of notation, a subscript may denote either time or a coordinate, but this should be clear from the context.}
In this case,
$\vw[\points] = \simplex(\vw[\pures])$ is the set of mixed strategies,
the expected loss at a mixed profile is given by $\vw[\loss](\jaction)=\ex_{\bb{s}\sim\jaction}{\vw[\loss](\bb{s})}$,
and
the player's feedback is the observation of the expected loss $\ex_{\vw[\bb{s}][\playexcept]\sim\vw[\jaction][\playexcept]}[\vw[\loss](\indg,\vw[\bb{s}][\playexcept])]$ for all $\indg\in\vw[\pures]$.
Our blanket assumption holds trivially since the mixed losses are multilinear.
\end{example}

\begin{example}[Kelly auctions]
\label{ex:auction}
Consider an auction of $K$ splittable resources among $\nPlayers$ bidders (players).
For the $\indg$-th resource, let $q_{\indg}$ and $c_{\indg}$ denote respectively its available quantity and the entry barrier for bidding on it;
for the $\play$-th bidder, let $\vw[b]$ and $\vw[\textsl{g}]$ denote respectively the bidder's budget and marginal gain from obtaining a unit of resources. 
During play, each bidder submits a bid $\vw[\action_{\indg}]$ for each resource $\indg$ with the constraint $\sum_{\indg=1}^{K}\vw[\action_{\indg}]\le\vw[b]$.
Resources are then allocated to bidders proportionally to their bids, so the $\play$-th player gets $\vw[\rho_{\indg}]=q_{\indg}\vw[\action_{\indg}]/(c_{\indg}+\sumplayer\vw[\action_{\indg}])$ units of the $\indg$-th resource.
The utility of player $\play\in\players$ is given by $\vw[\pay](\jaction)=\sum_{\indg=1}^K(\vw[\textsl{g}]\vw[\rho_{\indg}]-\vw[\action_{\indg}])$, %where $\vw[\textsl{g}]\ge0$ is the marginal gain for player $\play$ for a unit of resource,
and the loss function is $\vw[\loss]=-\vw[\pay]$.
\end{example}

\para{\acl{NE}}

In terms of solution concepts, the most widely used notion is that of a \acl{NE}, \ie a strategy profile from which no player has incentive to deviate unilaterally.
Formally, a point $\bb{\sol}\in\points$ is a \acl{NE} if for all $\play\in\players$ and all $\vwsp[\point]\in\vw[\points]$,
$\vw[\loss](\oneandother[\sol][\sol[\jaction]])\le\vw[\loss](\oneandother[\point][\sol[\jaction]])$.
For posterity, we will write $\sols$ for the set of \aclp{NE} of the game;
by a famous theorem of \citet{Deb52}, $\sols$ is always nonempty if $\points$ is compact.

%----------------------------------------------------------------------
%% Regret
%----------------------------------------------------------------------
\subsection{Regret minimization}

In the multi-agent learning model that we consider, players interact with each other repeatedly via a continuous convex game.
In more detail, during each round $\run$ of the process, each player $\play$ selects an action $\vwt[\action]$ from their action set $\vw[\points]$ and suffers a loss $\vw[\loss](\vt[\jaction])$, where $\vt[\jaction]=\allplayers{\vwt[\action]}$ is the joint action profile of all players.
At the end of each round, the players receive as feedback a subgradient vector 
\begin{equation}
    \label{eq:feedback}
    \vwt[\gvec]=\vw[\vecfield](\vt[\jaction])\in\subd_{\play}\vw[\loss](\vwt[\action],\vwt[\jaction][\playexcept]),
\end{equation}
and the process repeats.
We will also write $\jvecfield=\allplayers{\vw[\vecfield]}$ for the joint feedback operator.

In this low-information setting, the players have no knowledge about the rules of the game, and can only improve their performance by ``learning through play''.
It is therefore unrealistic to assume that players can pre-compute their component of an equilibrium profile;
however, it is plausible to expect that rational players would always seek to minimize their accumulated losses.
This criterion can be quantified via each player's \emph{individual regret}, \ie the difference between the player's cumulative loss and the best they could have achieved by playing a given action from a compact comparator set $\vw[\arpoints]\subseteq\vw[\points]$.

Following \citet{SS11}, we define the regret relative to a set of competing actions as
\begin{equation}
    \notag
    \vwt[\reg][\play][\nRuns](\vw[\arpoints]) = \max_{\vw[\arpoint]\in\vw[\arpoints]}\sum_{\run=1}^{\nRuns} 
    \left(\vw[\loss](\vwt[\action],\vwt[\jaction][\playexcept])
    - \vw[\loss](\vw[\arpoint],\vwt[\jaction][\playexcept])\right).
\end{equation}
Likewise, for $\arpoints\defeq\prod_{\play=1}^{\nPlayers}\vw[\arpoints]\subseteq\points$, we define the \emph{social regret} by aggregating over all players, viz,
\begin{equation}
    \notag
    \vt[\reg][\nRuns](\arpoints)
    =\sumplayer\vwt[\reg][\play][\nRuns](\vw[\arpoints])
    =\max_{\bb{\arpoint}\in\arpoints}
    \sumplayer\sum_{\run=1}^{\nRuns} 
    \left(\vw[\loss](\vwt[\action],\vwt[\jaction][\playexcept])
    - \vw[\loss](\vw[\arpoint],\vwt[\jaction][\playexcept])\right).
\end{equation}
In this context, a sequence of play $\vwt[\jaction]$ of player $\play$ incurs \emph{no individual regret} if $\vwt[\reg][\play][\nRuns](\vw[\arpoints])=\smalloh(\nRuns)$ for every (compact) set of alternative strategies;
correspondingly, $\vt[\jaction]$ incurs \emph{no social regret} if $\vt[\reg][\nRuns](\arpoints) = o(\nRuns)$.

In certain classes of games, the growth rate of the social regret can be related to the empirical mean of the players' social welfare \citep{SALS15}.
However, beyond this ``aggregate'' criterion, having no regret does not translate into any tangible guarantees for the quality of ``day-to-day'' play \cite{VZ13}.
On that account, we will measure the optimality of $\vwt[\point]$ at a given stage $\run$ by the gap function
\begin{equation}
    \notag
    \vwsp[\gap_{\vw[\arpoints]}](\vt[\jaction])
    = \vw[\loss](\vwt[\action],\vwt[\jaction][\playexcept])
    -\min_{\vw[\arpoint]\in\vw[\arpoints]}
    \vw[\loss](\vw[\arpoint],\vwt[\jaction][\playexcept]),
\end{equation}
\ie the best that the player could have gained by switching to any other strategy in $\vw[\arpoints]$ at round $\run$.
When $\vw[\arpoints]=\vw[\points]$ and $\vwsp[\gap_{\vw[\points]}](\vt[\jaction])\le\eps$ for every $\play\in\players$, we recover the definition of an $\eps$-equilibrium.

%----------------------------------------------------------------------
%%% OPTMD
%----------------------------------------------------------------------
\section{Optimistic mirror descent and its failures}
\label{sec:OptMD}
%----------------------------------------------------------------------
%%% OPTMD
%----------------------------------------------------------------------
% !TEX root = ../Main.tex

%----------------------------------------------------------------------
%% OPTMD
%----------------------------------------------------------------------
\para{The \acs{OptMD} template}

Our focal point in the sequel will be the \acdef{OptMD} class of algorithms, which, under different assumptions, has been shown to enjoy optimal regret minimization guarantees \cite{CYLM+12,RS13-COLT,RS13-NIPS}.
To define it, assume that each player $\play\in\players$ is equipped with a \emph{regularizer} $\vw[\hreg]\from\vw[\points]\to\R$,
\ie
a continuous, strongly convex function whose subdifferential $\subd\vw[\hreg]$ admits a continuous selection $\grad\vw[\hreg]$.
Then, given a sequence of feedback signals $\seqinf[\vwt[\gvec]]$ (defined in \eqref{eq:feedback} with the notation $\vwt[\gvec][\play][0]=0$),
the $\play$-th player
plays an action $\vwt[\action] = \vwtinter$ via the update rule
\begin{equation}
    \tag{OptMD}
    \label{eq:OptMD}
    \vwt[\state] = \argmin_{\point\in\vw[\points]}\thinspace
    \product{\vwtlast[\gvec]}{\point} + \vwtlast[\regpar]\vw[\breg](\point, \vwtlast[\state]),~~~
    \vwtinter = \argmin_{\point\in\vw[\points]}\thinspace
    \product{\vwtlast[\gvec]}{\point} + \vwt[\regpar]\vw[\breg](\point, \vwt[\state]),~~~
\end{equation}
%\vspace{4pt}
where
\begin{equation}
    \notag
    \vw[\breg](\arpoint,\point)
    = \vw[\hreg](\arpoint) - \vw[\hreg](\point) - \product{\grad\vw[\hreg](\point)}{\arpoint-\point}
    \qquad
    \arpoint\in\vw[\points],
    \point\in\dom\subd\vw[\hreg],
\end{equation}
denotes the \emph{Bregman divergence} of $\vw[\hreg]$ and $\vwt[\regpar]$ is a player-specific regularization parameter (more details on this below).
We also stress that, although \eqref{eq:OptMD} produces two iterates per step, only \emph{one} is actually played and directly contributes to the received feedback \textendash\ namely, $\vwt[\gvec]=\vw[\vecfield](\vwtinter[\state],\vwt[\jaction][\playexcept])$.

% \begin{wrapfigure}{br}{0.4\textwidth}
%     \centering
%     \includegraphics[height=0.18\textwidth]{Figures/bilinear-diverge-last.pdf}
%     \includegraphics[width=0.18\textwidth]{Figures/bilinear-diverge-ergodic.pdf}
%     \caption{The trajectories of the last (left) and the averaged (right) iterate when running \ac{PEG} for $\min_{\minvar\in[-6,10]}\max_{\maxvar\in[-6,10]}\minvar\maxvar$ with constant stepsize $\step=0.7>1/\sqrt{3}$. Neither of the two converges to the unique Nash equilibrium $(0,0)$.}
%     \label{fig:constant-bilinear-diverge}
% \end{wrapfigure}

\setlength{\columnsep}{15pt}
\begin{wrapfigure}{tr}{0.35\textwidth}
    \vspace{-2ex}
    \centering
    \includegraphics[width=0.35\textwidth]{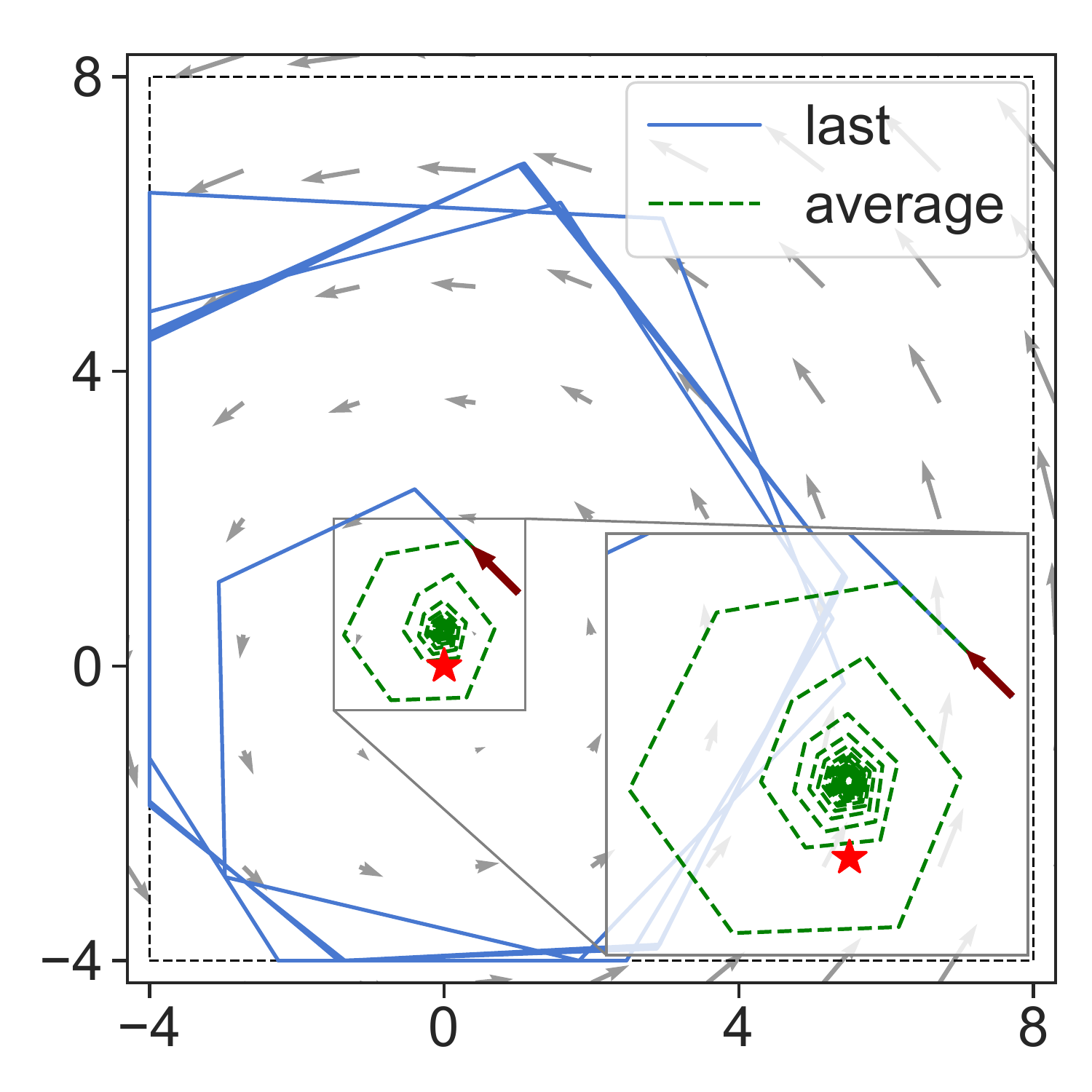}
    \vspace{-1.2em}
    \caption{The trajectories of play and its time-average when running \acs{PEG} for $\min_{\minvar\in[-4,8]}\max_{\maxvar\in[-4,8]}\minvar\maxvar$ with constant stepsize $\step=0.7>1/\sqrt{3}$. Neither of the two converges to the unique Nash equilibrium at $(0,0)$.}
    \label{fig:bilinear_compare}
    \label{fig:constant-bilinear-diverge}
    \vspace{-2em}
\end{wrapfigure}

Two of the most widely used instances of \eqref{eq:OptMD} are the \acdef{PEG} and \acdef{OMWU} algorithms, obtained respectively by
the quadratic regularizer $\vw[\hreg](\point) = \norm{\point}_{2}^{2}/2$
and
the negentropy function $\vw[\hreg](\point) = \sum_{\indg=1}^{\vdim_{\play}}\point_{\indg}\log\point_{\indg}$.
For a detailed discussion, see \cite{RS13-NIPS,MS16,GBVV+19,HIMM19,DP19} and references therein.

%----------------------------------------------------------------------
%% Failures
%----------------------------------------------------------------------
\para{Failures of \ac{OptMD}}

As we mentioned in the introduction, the convergence of \eqref{eq:OptMD} is only guaranteed as long as the players' regularization parameter $\vwt[\regpar]$ has been suitably fine-tuned \textendash\ specifically, as long as it is sufficiently large relative to the smoothness modulus of the players' loss functions.
However, this tuning is contingent on a degree of coordination and global knowledge of the game that is often impractical:
if $\vwt[\regpar]$ is not chosen properly, \eqref{eq:OptMD} may \textendash\ and, in fact, \emph{does} \textendash\ fail to converge.

We illustrate this failure in the simple min-max game
%$\min_{\minvar\in\R}\max_{\maxvar\in\R}\minvar\maxvar$, \ie
$\vw[\loss][1](\minvar,\maxvar) = \minvar\maxvar = -\vw[\loss][2](\minvar,\maxvar)$.
In this case, if both players run the \ac{PEG} instance of \eqref{eq:OptMD} with $\regpar > \sqrt{3}$, the sequence of play converges to the game's unique \acl{NE}.
However, if the players misestimate the critical value $\sqrt{3}$ and choose $\regpar < \sqrt{3}$, the method no longer converges to equilibrium, in either the ``ergodic'' or ``trajectory/last-iterate'' sense (for a proof, see \eg \cite{ZY20}).
Moreover, as we show in \cref{fig:constant-bilinear-diverge}, this ``off-equilibrium'' behavior persists even if we restrict the players' actions to a compact set:
in fact, not only does the method fail to converge to equilibrium, its average actually converges to an irrelevant action profile (an artifact of the trajectory's divergence).
This makes such failures particularly spurious and difficult to detect:
even though the algorithm stabilizes, the players' regret continues to accrue at a linear rate.

A simple remedy to the above would be to run \eqref{eq:OptMD} with an increasing regularization schedule, \eg of the form $\vwt[\regpar] \propto \sqrt{\run}$.
In some cases, this could indeed stabilize the algorithm and ensure convergence, but at a much slower rate \textendash\ in terms of both regret minimization and convergence speed.
An alternative would be to employ an adaptive schedule in the spirit of \cite{RS13-NIPS} (see \cref{sec:adaptive} for details), but even this is not enough:
as was shown by \citet{OP18}, when the ``Bregman diameter'' $\breg_{\points} \defeq [2\sup_{\arpoint,\point} \breg(\arpoint,\point)]^{1/2}$ of $\points$ is unbounded, mirror-based methods with an increasing regularization parameter may \textendash\ and often \emph{do} \textendash\ lead to \emph{superlinear} regret.%
\footnote{The precise result of \cite{OP18} concerns mirror descent;
however, it is straightforward to adapt their argument to show that, for example, the \ac{PEG} variant of \eqref{eq:OptMD} run on $\points=\R$ against the sequence $\vwt[\gvec] = (-1)^{\floor{(2\run-1)/\nRuns}}$ imposes $\Omega(\nRuns^{3/2})$ regret for both $\sqrt{\run}$ and adaptive regularization schedules.}
%\PM{On your note for proving this result formally:
%really, as you prefer, I would simly suggest not putting it in the main but in the supplement.}
This ``finite Bregman diameter'' condition rules out both MWU on the simplex and gradient descent in unbounded domains, and it is the first requirement that we relax in the next section.

\section{Optimistic averaging, adaptation, and stabilization}
\label{sec:adaptive}
%----------------------------------------------------------------------
%%% ADAPTIVE
%----------------------------------------------------------------------
% !TEX root = ../Main.tex

%----------------------------------------------------------------------
%% OptDA
%----------------------------------------------------------------------
\subsection{Optimistic dual averaging}
\label{sec:OptDA}

Viewed abstractly, the failures described above are due to the following aspect of \eqref{eq:OptMD}:
\begin{center}
%\vspace{-\smallskipamount}
\emph{With an increasing schedule for $\vt[\regpar]$, new information enters \eqref{eq:OptMD} with a decreasing weight.}
\end{center}
%\vspace{-\smallskipamount}
From a learning viewpoint, this behavior is undesirable because it gives more weight to earlier, uninformed updates, and less weight to more recent, more relevant ones (so, mutatis mutandis, an adversary could push the algorithm very far from an optimal point in the starting iterations of a given window of play).
To account for this disparity, we build on an idea originally due to \citet{Nes07}, and introduce the \acdef{OptDA} method as:
\begin{equation}
    \tag{OptDA}
    \label{eq:OptDA}
    \begin{aligned}
    \vwt[\state] &= \argmin_{\point\in\vw[\points]} \sum_{\runalt=1}^{\run-1}\product{\vwt[\gvec][\play][\runalt]}{\point} + \vwt[\regpar]\vw[\hreg](\point),\\
    \vwtinter &= \argmin_{\point\in\vw[\points]}\thinspace
    \product{\vwtlast[\gvec]}{\point} + \vwt[\regpar]\vw[\breg](\point, \vwt[\state]).
\end{aligned}
\end{equation}
In contrast to \eqref{eq:OptMD}, the base state $\vt[\state]$ of \eqref{eq:OptDA} is produced by aggregating all feedback received with the \emph{same weight} (the first line in the algorithm);
subsequently, each player selects an action $\vwt[\action]=\vwtinter[\state]$ after taking a ``conservatively optimistic'' step forward (this one with a decreasing step-size, for reasons of stability).
As we will show, this different aggregation architecture plays a crucial role in overcoming the ``finite Bregman diameter'' limitation of \eqref{eq:OptMD}.

From a design perspective, the core elements of \eqref{eq:OptDA} are
\begin{enumerate*}
[\itshape a\upshape)]
\item
the choice of ``learning rate'' parameters $\vwt[\regpar]$ (which now acts both as a regularization weight and as an inverse step-size);
and
\item
the choice of regularizer $\vw[\hreg]$, which defines the ``mirror map'' $\vw[\mirror]\from\dvec\to \argmax_{\point\in\vw[\points]} \product{\dvec}{\point} - \vw[\hreg](\point)$ that determines the update of the base state $\vwt[\state]$ of \eqref{eq:OptDA}.
\end{enumerate*}
We discuss both elements in detail in the remainder of this section.

\begin{remark*}[Optimistic FTRL]
Another closely related algorithm is the optimistic variant of \acli{FTRL} (\acs{OptFTRL})\acused{OptFTRL}
%\acdef{OptFTRL}
\cite{ALLW18,MY16,SALS15}, whose updates follow the recursion
\begin{equation}
    \tag{OptFTRL}
    \label{eq:OptFTRL}
    \vwtinter[\state] = \argmin_{\point\in\vw[\points]}
    \left\langle
    \sum_{\runalt=1}^{\run-1}\vwt[\gvec][\play][\runalt]+\vwtlast[\gvec],
    \point\right\rangle
    + \vwt[\regpar]\vw[\hreg](\point).
\end{equation}
Compared to \eqref{eq:OptDA}, \eqref{eq:OptFTRL} aggregates all the relevant feedback, including $\vwtlast[\gvec]$ directly in the dual space.
In this way, there is no need to define $\vwt[\state]$, which acts as an auxiliary state to produce the actual iterate $\vwtinter[\state]$ in both \eqref{eq:OptMD} and \eqref{eq:OptDA}.
Nonetheless, while the regret bounds presented in \cref{sec:regret} can also be obtained for adaptive variants of \eqref{eq:OptFTRL}, the fact that all updates are performed in the dual space prevents us from proving the last-iterate convergence results of \cref{sec:last}.
\end{remark*}

%----------------------------------------------------------------------
%% Adaptation
%----------------------------------------------------------------------
\subsection{Learning rate adaptation}

Since running the algorithm with a $\sqrt{\run}$ learning rate schedule is, in general, too pessimistic, we will consider an adaptive policy in the spirit of \citet{RS13-NIPS},
namely
\begin{equation}
    \label{eq:adaptive-reg}
    \txs
    \tag{Adapt}
    \vwt[\regpar] = \sqrt{\vw[\tau]+\sum_{\runalt=1}^{\run-1}\vwt[\increment]}
    ~~~~\text{where}~~~
    \vwt[\increment] = \vwpdual[\norm{\vwt[\gvec]-\vwtlast[\gvec]}^2].
\end{equation}
In the above, $\vw[\tau] > 0$ is a player-specific constant that can be chosen freely by each player,
and $\vwpdual[\norm{\cdot}]\from\dvec\to\max_{\vwp[\norm{\point}]\le1}{\product{\dvec}{\point}}$ is the dual norm of $\vwp[\norm{\cdot}]$, itself a norm on $\R^{\vdim_{\play}}$.
Intuitively, in the favorable case (\eg when the environment is stationary), the increments $\vwt[\increment]$ will eventually vanish, so the policy \eqref{eq:adaptive-reg} will be a proxy for the ``constant step-size'' case.
By contrast, in a non-favorable\,/\,adversarial setting, we have $\vwt[\increment] = \Theta(1)$ and $\vwt[\regpar]$ grows as $\Theta(\sqrt{\run})$, which makes the algorithm robust.

We should also note here that \eqref{eq:adaptive-reg} involves \emph{exclusively} player-specific quantities, and its computation only makes use of information that is available to each player \emph{locally}.
This is not always the case for other adaptive learning rates considered in the game-theoretic literature, \eg as in \cite{LZMJ20,ABM19,ABM21}.
Even though this ``local information'' desideratum is very natural, very few algorithms with this property have been analyzed in the game theory literature.

%----------------------------------------------------------------------
%% Stabilization
%----------------------------------------------------------------------
\subsection{Reciprocity and stabilization}

%\para{Fenchel coupling and reciprocity}
In the aggregation step of \eqref{eq:OptDA}, the mirror map $\vw[\mirror]$ maps a dual vector back to the primal space to obtain $\vwt[\state]$.
%These dual vectors play a fundamental role in the analysis of the algorithm.
For this reason, to analyze the players' sequence of play, we will make use of the Fenchel coupling, a ``primal-dual'' distance of measure first introduced in \cite{MS16,MerSta18,BM17}.
To define it, let $\dual{(\vw[\hreg])}$ be the Fenchel conjugate of $\vw[\hreg]$, \ie $\dual{(\vw[\hreg])}(\dvec)=\max_{\point\in\vw[\points]}\product{\dvec}{\point}-\vw[\hreg](\point)$.
The Fenchel coupling induced by $\vw[\hreg]$ between a primal point $\arpoint\in\vw[\points]$ and a dual vector $\dvec\in\R^{\vdim_{\play}}$ is then defined as 
\begin{equation}
    \notag
    \vw[\fench](\arpoint,\dvec) = \vw[\hreg](\arpoint) + \dual{(\vw[\hreg])}(\dvec) - \product{\dvec}{\arpoint}.
\end{equation}
One key property of the Fenchel coupling is that $\vw[\fench](\arpoint,\dvec)\ge(1/2)\vwp[\norm{\vw[\mirror](\dvec)-\arpoint}^2]$ for some norm $\vwp[\norm{\cdot}]$ on $\vw[\points]$.
Therefore, it can be used to measure the convergence of a sequence.
In particular, $\vw[\mirror](\vwt[\dstate])\to\vw[\arpoint]$ whenever $\vw[\fench](\vw[\arpoint],\vwt[\dstate])\to0$.
For several results concerning the trajectory convergence of the algorithm, it will also be convenient to assume the converse, that is

\begin{assumption}[Fenchel reciprocity \cite{MZ19}]
\label{asm:Fenchel-recp}
For any $\play\in\players$, $\vw[\arpoint]\in\vw[\points]$, and $\seqinf[\vwt[\dstate]]$ a sequence of dual vectors such that $\vw[\mirror](\vwt[\dstate])\to\vw[\arpoint]$, we have $\vw[\fench](\vw[\arpoint],\vwt[\dstate])\to0$.
\end{assumption}

Given the similarity between the Fenchel coupling and the Bregman divergence (which we discuss in detail in \cref{app:mirror}),
Fenchel reciprocity may be regarded as a primal-dual analogue of the more widely used Bregman reciprocity condition \cite{CT93,Kiw97}.

\asmtag{\ref*{asm:Fenchel-recp}$'$}
\begin{assumption}[Bregman reciprocity]
\label{asm:Bregman-recp}
For any $\play\in\players$, $\vw[\arpoint]\in\vw[\points]$, and $\seqinf[\vwt[\state]]$ a sequence of primal points such that $\vwt[\state]\to\vw[\arpoint]$, it holds $\vw[\breg](\vw[\arpoint],\vwt[\state])\to0$.
\end{assumption}

It can be verified that Bregman reciprocity is indeed implied by Fenchel reciprocity, but the opposite is generally not true.
For example, when $\vw[\hreg]$ is the quadratic regularizer, Bregman reciprocity always holds while Fenchel reciprocity is only guaranteed when $\vw[\points]$ is a polytope.

In this regard, it is desirable to devise an algorithm with the same regret guarantees as \ac{OptDA} while only requiring the less stringent Bregman reciprocity condition to ensure the convergence of the trajectory.
This motivates us to introduce \acdef{DS-OptMD}, in which player $\play$ recursively computes their realized action $\vwt[\action]=\vwtinter[\state]$ by
\begin{equation}
    \tag{DS-OptMD}
    \label{eq:DS-OptMD}
    \begin{aligned}
    \vwt[\state] &= \argmin_{\point\in\vw[\points]}\thinspace
    \product{\vwtlast[\gvec]}{\point} + \vwtlast[\regpar]\vw[\breg](\point, \vwtlast[\state])
    + (\vwt[\regpar]-\vwtlast[\regpar])\vw[\breg](\point, \vwt[\state][\play][\start]),
    \\
    \vwtinter &= \argmin_{\point\in\vw[\points]}\thinspace
    \product{\vwtlast[\gvec]}{\point} + \vwt[\regpar]\vw[\breg](\point, \vwt[\state]).
    \end{aligned}
\end{equation}
The stabilization step (\ie the anchoring term that appears in the first line of the update) is inspired by \citet{FHPF20},
and it has been shown to help the algorithm achieve no regret even when the Bregman diameter is unbounded.
Moreover, by standard arguments \cite{KM17,Mer19,JKM19,FHPF20}, we can show that when the mirror map is interior-valued, \ie $\im\vw[\mirror]=\relint\vw[\points]$ (here $\relint\vw[\points]$ denotes the relative interior of $\vw[\points]$), the update of \eqref{eq:DS-OptMD} coincides with that of \eqref{eq:OptDA}.%
\footnote{Precisely, this requires to set $\vt[\state][\start]=\argmin_{\point\in\vw[\points]}\vw[\hreg](\points)$ in \eqref{eq:DS-OptMD}.}
One important example which falls into this situation is the (stabilized) \ac{OMWU} algorithm \cite{DP19}, whose update can be written in a coordinate-wise way as follows 
\begin{equation}
    \label{eq:OMWU-dynamic}
    \tag{OMWU}
    \vwct[\action] = \vwctinter[\state]
    = \frac{\exp(-(\sum_{\runalt=1}^{\run-1}\vctnp[\gvec][\indg][\runalt]+\vctnp[\gvec][\indg][\run-1])/\vwt[\regpar])}
    {\sum_{\indgalt=1}^{\vdim_{\play}}\exp(-(\sum_{\runalt=1}^{\run-1}\vctnp[\gvec][\indgalt][\runalt]+\vctnp[\gvec][\indgalt][\run-1])/\vwt[\regpar])}.
    %\,~~~\text{where}~~~\vwt[\step]=\frac{1}{\vwt[\regpar]}.
\end{equation}

%----------------------------------------------------------------------
%% Template
%----------------------------------------------------------------------
\subsection{A template descent inequality}
\label{sec:template-descent}

For the results presented in this work, we provide an umbrella analysis for \ac{OptDA} and \ac{DS-OptMD} by means of the following energy inequality.

\begin{lemma}
\label{lem:template-descent}
Suppose that player $\play$ runs \eqref{eq:OptDA} or \eqref{eq:DS-OptMD}.
Then, for any $\vw[\arpoint]\in\vw[\points]$, we have
\begin{equation}
    \label{eq:template-descent-main}
    \begin{aligned}
    \vwtupdate[\regpar]\vwtupdate[\estseq](\vw[\arpoint])
    &\le
    \vwt[\regpar]\vwt[\estseq](\vw[\arpoint])
    - \product{\vwt[\gvec]}{\vwtinter-\vw[\arpoint]}
    + (\vwtupdate[\regpar]-\vwt[\regpar])\vw[\armeasure](\vw[\arpoint])\\
    &~+ \product{\vwt[\gvec]-\vwtlast[\gvec]}{\vwtinter-\vwtupdate[\state]}
    - \vwt[\regpar]\vw[\breg](\vwtupdate,\vwtinter) - \vwt[\regpar]\vw[\breg](\vwtinter,\vwt[\state]),
    \end{aligned}
\end{equation}
where:
\begin{enumerate}
[\upshape(\itshape i\hspace*{.5pt}\upshape)]
\item
$\vwt[\estseq](\vw[\arpoint]) = \vw[\fench](\vw[\arpoint],\vwt[\dstate])$,
$\vw[\armeasure](\vw[\arpoint]) = \vw[\hreg](\vw[\arpoint]) - \min\vw[\hreg]$
for \eqref{eq:OptDA}.
%and
\item
$\vwt[\estseq](\vw[\arpoint]) = \vw[\breg](\vw[\arpoint],\vwt[\state])$,
$\vw[\armeasure](\vw[\arpoint]) = \vw[\breg](\vw[\arpoint],\vwt[\state][\play][\start])$
for \eqref{eq:DS-OptMD}.
\end{enumerate}
\end{lemma}

The proof of \cref{lem:template-descent} combines several techniques used in the analysis of regularized online learning algorithms and is deferred to \cref{app:mirror}.
As a direct consequence of \cref{lem:template-descent}, we have
\begin{equation}
    \label{eq:template-regret-main}
    \sum_{\run=1}^{\nRuns}\product{\vwt[\gvec]}{\vwtinter[\state]-\vw[\arpoint]}
    \le \vwt[\regpar][\play][\nRuns+1]\vw[\armeasure](\vw[\arpoint])
    + \sum_{\run=1}^{\nRuns}
    \frac{\vwpdual[\norm{\vwt[\gvec]-\vwtlast[\gvec]}^2]}{\vwt[\regpar]}
    -\sum_{\run=2}^{\nRuns}\frac{\vwtlast[\regpar]}{8} \vwp[\norm{\vwtinter[\state]-\vwtpast[\state]}^2].
\end{equation}
This is very similar to the \acdef{RVU} property introduced by \citet{SALS15}, but it now applies to an algorithm with possibly non-constant learning rate (and, of course, to continuous action spaces).
By invoking the individual convexity assumption, \eqref{eq:template-regret-main} gives an implicit upper bound on the individual regret of each player.
Moreover, \eqref{eq:template-descent-main} relates the distance measure of round $\run$ to that of round $\run+1$. Therefore, we can also leverage \cref{lem:template-descent} to prove the convergence of the learning dynamics.
In \cref{app:template} we explain in detail how this template inequality can be used to derive exactly the same guarantees for other learning algorithms as long as they satisfy a version of \eqref{eq:template-descent-main}.

%----------------------------------------------------------------------
%%% REGRET
%----------------------------------------------------------------------
\section{Optimal regret bounds}
\label{sec:regret}
%----------------------------------------------------------------------
%%% REGRET
%----------------------------------------------------------------------
% !TEX root = ../Main.tex

In this section, we derive a series of min-max optimal regret bounds, both when the opponents are adversarial and when all the players interact according to prescribed algorithms.
The proofs of our results leverage the template inequality \eqref{eq:template-regret-main} and are deferred to \cref{app:regret}.

%----------------------------------------------------------------------
%% Regret
%----------------------------------------------------------------------
\subsection{Regret guarantees: individual and social}

Our first result provides a worst-case guarantee for \emph{any} sequence of play realized by the opponents.

\begin{restatable}{theorem}{AdvReg}
\label{thm:adversarial-regret}
Suppose that \cref{asm:convexity+smoothness} holds, and a player $\play\in\players$ adopts \eqref{eq:OptDA} or \eqref{eq:DS-OptMD} with the adaptive learning rate \eqref{eq:adaptive-reg}.
If $\vw[\arpoints] \subseteq \vw[\points]$ is bounded and $\gbound = \sup_{\run} \norm{\vwt[\gvec]}$,
the regret incurred by the player is bounded as
$\vwt[\reg][\play][\nRuns](\vw[\arpoints])=\bigoh(\gbound\sqrt{\nRuns}+\gbound^2)$.
\end{restatable}

\cref{thm:adversarial-regret} is a direct consequence of \eqref{eq:template-regret-main} and the definition of the adaptive learning rate.
It addresses what is traditionally referred to as the adversarial scenario, since we do not make any assumptions on how the opponents' actions are selected;
in particular, they may choose the actions so as to maximize the player's cumulative loss.
Even in this case, \cref{thm:adversarial-regret} shows that the two adaptive algorithms that we consider would achieve no regret provided that the sequence of feedback is bounded (this is for example the case when $\points$ is compact).
%In other words, the two adaptive algorithms that we consider are both order-optimal in terms of regret minimization guarantees.

%We remark that in all the examples motioned in \cref{subsec:repeated-game}, the action space $\points$ is compact and thus the boundedness of the feedback is verified automatically.

We now proceed to show that, if all players adhere to one of the adaptive policies discussed so far, the social regret is at most constant.

\begin{theorem}
\label{thm:social-regret-bounded}
Suppose that \cref{asm:convexity+smoothness} holds and all players $\play\in\players$ use \eqref{eq:OptDA} or \eqref{eq:DS-OptMD} with the adaptive learning rate \eqref{eq:adaptive-reg}.
Then, for every bounded comparator set $\arpoints\subseteq\points$, the players' social regret is bounded as $\vt[\reg][\nRuns](\arpoints)=\bigoh(1)$.
\end{theorem}

The closest result in the literature is that of \cite{SALS15}, which proves a constant regret bound for \emph{finite} games for all algorithms that satisfy the \ac{RVU} property.
\cref{thm:social-regret-bounded} improves upon this result in two fundamental aspects:
First, \cref{thm:social-regret-bounded} applies to \emph{any} continuous game with smooth and convex losses, % and a convex structure,
not just mixed extensions of finite games.
Second,
the proposed policies do not require any prior knowledge about the game's parameters (such as the relevant Lipschitz constants and the like).

An additional appealing property of our analysis is that, to the best of our knowledge, this is the first guarantee that shaves off the logarithmic in $\nRuns$ factors in this specific setting for a method that is robust to adversarial opponents (\ie \cref{thm:adversarial-regret}).
This relies on a careful analysis of \eqref{eq:template-regret-main} with the specific learning rate \eqref{eq:adaptive-reg}.
We note additionally that, in \cref{thm:social-regret-bounded}, the players \emph{do not} need to use the same regularizer or even the same template algorithm:
As a matter of fact, the only requirement for this result to hold is that the players' sequence of play satisfies a version of the inequality \eqref{eq:template-regret-main}.

%----------------------------------------------------------------------
%% Variational stability
%----------------------------------------------------------------------
\subsection{Individual regret under variational stability}

We close this section by zooming in on a class of convex games known as \emph{variationally stable}:

\begin{definition}
A continuous convex game is \emph{variationally stable} if
the set $\sols$ of Nash equilibria of the game is nonempty and
\begin{equation}
    \label{eq:VS}
    \product{\jvecfield(\jaction)}{\jaction-\bb{\sol}}
    =\sumplayer\product{\vw[\vecfield](\jaction)}{\vwsp[\action]-\vwsp[\sol]}
    \ge0
    \quad
    \text{for all $\jaction\in\points$, $\bb{\sol}\in\sols$.}
\end{equation}
The game is \emph{strictly variationally stable} if \eqref{eq:VS} holds as a strict inequality whenever $\jaction\notin\sols$.%
%\footnote{In full generality, we may have different scaling for different players in \eqref{eq:VS}. That is, the term $\sumplayer\product{\vw[\vecfield](\jaction)}{\vwsp[\action]-\vwsp[\sol]}$ should be replaced by $\sumplayer\vw[\pi]\product{\vw[\vecfield](\jaction)}{\vwsp[\action]-\vwsp[\sol]}$ for some $\bb{\pi}=\allplayers{\vw[\pi]}\in\R^{\nPlayers}_+$.
%\eqref{eq:VS} can then be achieved by scaling each loss function properly, and thus our results readily apply to this more general case.}%
\end{definition}

A notable family of games that verify the variational stability condition is monotone games (\ie $\jvecfield$ is monotone), which includes convex-concave zero-sum games, zero-sum polymatrix games, Cournot oligopolies, and Kelly auctions (\cref{ex:auction}) as several examples.
The last two examples satisfy in fact a more stringent diagonal strict concavity condition (\citet{Rosen65}), \ie the vector field $\jvecfield$ is strictly monotone, which implies the strict variational stability of the game.
%Similarly, strict variational stability is implied by the strict monotonicity of $\jvecfield$, also known as  for payoff maximization games. It is satisfied by \YGHcomment{...}.

\begin{remark*}
In the literature, the term ``variationally stable'' frequently signifies what we refer to as ``strictly variationlly stable'';
this is for example the case in \cite{MZ19}, where the concept was first introduced.
\end{remark*}

Under this stability condition, we derive a constant regret bound on the individual regrets of the players when they play against each other using a prescribed algorithm.

\begin{theorem}
\label{thm:inidividual-regret-bound}
Suppose that \cref{asm:convexity+smoothness} holds and all players $\play\in\players$ use \eqref{eq:OptDA} or \eqref{eq:DS-OptMD} with the adaptive learning rate \eqref{eq:adaptive-reg}.
%Let \cref{asm:convexity+smoothness} holds and that all the players adopt either adaptive \ac{OptDA} or adaptive \ac{DS-OptMD}. 
If the game is variationally stable, then, for every bounded comparator set $\vw[\arpoints]\subseteq\vw[\points]$, the individual regret of player $\play\in\players$ is bounded as $\vwt[\reg][\play][\nRuns](\vw[\arpoints])=\bigoh(1)$.
\end{theorem}

\cref{thm:inidividual-regret-bound} extends a range of results previously proved for \emph{finite} two-player, zero-sum games for various learning algorithms \citep{DDK11,KHSC18,RS13-NIPS}.
It also inherits the appealing attribute of the social regret bound of \cref{thm:social-regret-bounded} \textendash\ namely, that all logarithmic factors have been shaved off.

The main difficulty in the proof of \cref{thm:inidividual-regret-bound} is to show that the sequence of gradient increments $\seqinf[\vwt[\increment]]$ is actually summable for all $\play\in\players$.
Equivalently, this implies that each player's learning rate $\vwt[\regpar]$ converges to a finite constant that is automatically adapted to the smoothness landscape of the game.
To achieve this, we follow a proof strategy that is similar in spirit to the approach of \cite{ABM21} for solving variational inequalities;
however, our setting is considerably more complicated because each player's learning rate is different.
%\PM{Maybe put a couple of words here about how we solve this issue?}

%We adapt their technique to our player-specific learning rate which is more relevant to the learning-in-game setting.
%It is also related to the $\bigoh(1/\nRuns)$ convergence rate

%----------------------------------------------------------------------
%%% LAST ITERATE
%----------------------------------------------------------------------
\section{Convergence of the day-to-day trajectory of play}
\label{sec:last}
%----------------------------------------------------------------------
%%% LAST ITERATE
%----------------------------------------------------------------------
% !TEX root = ../Main.tex

So far, our results have focused on ``average'' measures of performance, namely the players' individual and social regret.
Even though the derived bounds are sharp, as we discussed in \cref{sec:setup}, they cannot be used to draw meaningful conclusions for the players' \emph{actual} sequence of play.
Our analysis in this section shows that, in fact, the proposed learning methods actually stabilize to a best response or a \acl{NE} in a number of relevant cases.
The proof details are deferred to \cref{app:convergence}.

%----------------------------------------------------------------------
%% Variational stability
%----------------------------------------------------------------------
\subsection{Convergence to best response against convergent opponents}

A fundamental consistency property for online learning in games is that any player should end up ``best responding'' to the action profile of all other players if their actions stabilize (or are stationary).
Formally, a player $\play\in\players$ is said to ``\emph{converge to best response}'' if, whenever the action profile $\vwt[\jaction][\playexcept]$ of all other players converges to some limit profile $\vw[\limp{\jaction}][\playexcept]\in\prod_{\playalt\neq\play}\vw[\points][\playalt]$, the sequence of actions $\vwt[\action][\play] \in \vw[\points][\play]$ of the focal player $\play\in\players$ converges itself to $\BR(\vw[\limp{\jaction}][\playexcept]) \defeq \argmin_{\vw[\action]\in\vw[\points]}\vw[\loss](\vw[\action],\vw[\limp{\jaction}][\playexcept])$.
We establish this key requirement below.

\begin{theorem}
\label{thm:cvg-best-response}
Suppose that \cref{asm:convexity+smoothness,asm:Fenchel-recp} \textpar{resp. \ref{asm:Bregman-recp}} hold, and a player $\play\in\players$ employs \eqref{eq:OptDA} \textpar{resp. \eqref{eq:DS-OptMD}} with the adaptive learning rate \eqref{eq:adaptive-reg}.
If $\vw[\points]$ is compact, the trajectory of chosen actions of the player in question converges to best response.
\end{theorem}

\begin{proof}[Idea of proof]
The fact that the opponents are only convergent rather than stationary makes the proof much more challenging and requires a non-standard ``trapping'' argument.%
\footnote{In fact, the compactness assumption in \cref{thm:cvg-best-response} can be dropped if the opponents are stationary.}
Specifically, we show that when the sequence $\vwt[\state]$ gets close to a best response (\ie when $\min_{\vw[\sol]\in\BR(\vw[\limp{\jaction}][\playexcept])}\vwt[\estseq](\vw[\sol])\le\sbradius$ for some $\sbradius>0$),
all subsequent iterates must remain in this neighborhood provided that $\run$ is sufficiently large.
Subsequently, we also show that the sequence $\seqinf[\vwt[\state]]$ visits any neighborhood of $\BR(\vw[\limp{\jaction}][\playexcept])$ infinitely many times.
Therefore, for every neighborhood of $\BR(\vw[\limp{\jaction}][\playexcept])$, the iterates eventually get trapped into that neighborhood, and we conclude by showing $\vwp[\norm{\vwtinter[\state]-\vwt[\state]}]$ converges to zero.
\end{proof}

As a direct consequence of \cref{thm:cvg-best-response}, we deduce that $\lim_{\run\to+\infty}\vwsp[\gap_{\vw[\points]}](\vt[\jaction])=0$ whenever the opponents' actions converge. Therefore, the action of the player becomes quasi-optimal as time goes by, in the sense that they would not earn much more by switching to any other strategy in each round.

%\begin{remark*}
%The compactness assumption in \cref{thm:cvg-best-response} can be dropped if the opponents are stationary.
%\end{remark*}

\subsection{Main result: Convergence to Nash equilibrium}

Moving forward, we proceed to establish a series of results concerning the convergence of the players' trajectory of play to \acl{NE} when all players employ an adaptive learning algorithm.
%\PM{The theorem is a bit tricky to parse, but I don't have the energy to try and simplify, sorry\dots}
%following theorem establishes the convergence to a Nash equilibrium when all the players use a prescribed algorithm in a (strictly) variationally stable game.

\begin{theorem}
\label{thm:converge-to-Nash}
Suppose that \cref{asm:convexity+smoothness,asm:Fenchel-recp} \textpar{resp. \ref{asm:Bregman-recp}} hold and all players $\play\in\players$ use either \eqref{eq:OptDA} or \eqref{eq:DS-OptMD} \textpar{resp.~only \eqref{eq:DS-OptMD}} with the adaptive learning rate \eqref{eq:adaptive-reg}.
%Let \cref{asm:convexity+smoothness,asm:Fenchel-recp} (resp. \ref{asm:Bregman-recp}) hold and that all the players adopt either adaptive \ac{OptDA} or adaptive \ac{DS-OptMD} (resp. only adaptive \ac{DS-OptMD}).
Then the induced trajectory of play converges to a \acl{NE} provided that either of the following conditions is satisfied
\begin{enumerate}[\itshape a\upshape),itemsep=0pt,topsep=1.2pt]
\item
The game is strictly variationally stable.
\item
The game is variationally stable and $\vw[\hreg]$ is subdifferentiable on all $\vw[\points]$, \ie $\dom\subd\vw[\hreg]=\vw[\points]$.
\vspace{-0.4em}
%$\forall \play\in\players, \vw[\points]\subset\dom\subd\vw[\hreg]$;
\end{enumerate}
%This also implies that for every $\play\in\players$ and every compact set $\vw[\arpoints]\in\vw[\points]$, we have $\limsup_{\run\to+\infty}\vwsp[\gap_{\vw[\arpoints]}](\vt[\jaction])\le0$
\end{theorem}
\begin{proof}[Idea of proof]
The proof of the two cases follow the same schema.
We first establish that every cluster point of $\seqinf[\vt[\jstate]]$ is a Nash equilibrium.
This utilizes the fact that $\vwt[\regpar]$ converges to a finite constant as shown in the proof of \cref{thm:inidividual-regret-bound}.
Then, to prove the sequence actually converges, we leverage the reciprocity conditions discussed in \cref{sec:OptDA} together with a quasi-Fejér property \cite{Com01} that we establish for the induced sequence of play relative to a suitable divergence metric.
\end{proof}
%\vspace{-0.55em}

The convergence to a Nash equilibrium $\bb{\sol}$ implies that for every $\play\in\players$ and every compact set $\vw[\arpoints]\in\vw[\points]$, $\lim_{\run\to+\infty}\vwsp[\gap_{\vw[\arpoints]}](\vt[\jaction])=\vwsp[\gap_{\vw[\arpoints]}](\bb{\sol})\le0$.
Thus, in the long run, the players are individually satisfied with their own choices of each play compared to any other action they could have pick from a comparator set.
To the best of our knowledge, this is the first equilibrim convergence result for online learning in variationally stable games with a player-specific, adaptive learning rate.
The closest antecedent to our result is the recent work of \cite{LZMJ20} where the authors prove convergence to \acl{NE} in unconstrained cocoercive games,%
\footnote{The class of cocoercive games is defined by the property $\product{\jvecfield(\jaction) - \jvecfield(\jactionalt)}{\jaction-\jactionalt} \geq (1/\beta) \dnorm{\jvecfield(\jaction) - \jvecfield(\jactionalt)}^{2}$.}
%\YGH{I do not really understand the second half of the footnote. I do not think there is a direct relation between strict monotonicity and cocoerciveness.}
%\PM{Removed.}
with an adaptive step-size that is the same across player (and which therefore requires access to global information to be computed).
In this regard, \cref{thm:converge-to-Nash} extends a wide range of earlier equilibrium convergence results for \emph{strictly} stable games that were obtained with a constant or diminishing \textendash\ but not \emph{adaptive} \textendash\ step-size.
%\cref{thm:converge-to-Nash} suggests that running algorithms with adaptive learning rate does not only guarantee 
%Prior to our work, the convergence to Nash equilibrium was mostly obtained with either constant or diminishing (square summable but not summable) learning rate.
%One exception is, \cite{LZMJ20} being the only exception.
%Moreover, as far as we are aware, this also appears to be the first last-iterate convergence result of any close variant of \ac{OptDA} under the variational stability condition.
%(while for the mirror-descent variant with constant stepsize, a closely related result can be found in)

Despite the generality of \cref{thm:converge-to-Nash}, it fails to cover the case where the players are running localized, adaptive versions of \ac{OMWU} in a game that is \acl{VS} but not \emph{strictly} so.
The most representative example of this special case is finite two-player zero-sum games with a mixed equilibrium;
we address this case below.
%We partially address this issue by providing a dedicated theorem for the special case of finite two-player zero-sum game, which we believe to be the most representative of this situation.

\begin{restatable}{theorem}{AdaptiveOMWU}
\label{thm:converge-OMWU}
Suppose that the players of a two-player, finite zero-sum game follow \eqref{eq:OMWU-dynamic} with the adaptive learning rate \eqref{eq:adaptive-reg}.
Then the induced sequence of play converges to a \acl{NE}.
\end{restatable}

The closest results in the literature are \cite{DP19} and, most recently, \cite{WLZL21}.
\Cref{thm:converge-OMWU} sharpens these results in two key aspects:
\begin{enumerate*}[\itshape i\upshape)]
    \item the players' learning rate is not contingent on the knowledge of game-specific constants;
    and
    \item we do not assume the existence of a \emph{unique} \acl{NE}.
\end{enumerate*}
%However, with an adaptive learning rate, we are not able to provide any convergence rate of the sequence.

Finally, following the proof of \cref{thm:converge-to-Nash}, we establish below an interesting dichotomy for general convex games with compact action sets (see also \cref{app:dichotomy} for a non-compact version).

\begin{theorem}
\label{thm:dichotomy-general-sum}
Suppose that \cref{asm:convexity+smoothness} holds and all players $\play\in\players$ use \eqref{eq:OptDA} or \eqref{eq:DS-OptMD} with the adaptive learning rate \eqref{eq:adaptive-reg}.
%Let \cref{asm:convexity+smoothness} holds and that all the players adopt either adaptive \ac{OptDA} or adaptive \ac{DS-OptMD}.
Assume additionally that $\vw[\points]\subset\dom\subd\vw[\hreg]$ for every $\play\in\players$ and $\points$ is compact.
Then one of the following holds:
\begin{enumerate}
[\upshape(\itshape a\upshape),topsep=1.2pt,itemsep=1.2pt,leftmargin=*]
\item
The sequence of realized actions converges to the set of Nash equilibria.
Furthermore, for every $\play\in\players$, it holds $\vwt[\reg][\play][\nRuns](\vw[\points])=\bigoh(1)$ and $\limsup_{\run\to+\infty}\vwsp[\gap_{\vw[\points]}](\vt[\jaction])\le0$. 

\item
The social regret tends to minus infinity when $\toinf$, \ie $\lim_{\toinf}\vt[\reg][\nRuns](\points)=-\infty$.
\end{enumerate}
\end{theorem}

\Cref{thm:dichotomy-general-sum} shows that, if the player's sequence of actions fails to converge, the social regret goes to $-\infty$;
%there is at least one player whose individual regret goes to $-\infty$;
in particular, there is at least one player who benefits more from the actions employed by all other players compared to the regret incurred by all the dissatisfied players put together.
For this player in question, the individual regret goes to $-\infty$ and the player actually benefits from not converging to a fixed action.
%\YGH{Maybe we can discuss this (and the last sentence of the section) live.}
%\PM{Sure.}
%and thus it actually benefits from not converging.
We are not aware of any similar result in the literature.

%This dichotomy is contingent on whether $\sum_{\run=1}^{+\infty}(\norm{\inter[\jstate]-\vt[\jstate]}^2+\norm{\vt[\jstate]-\past[\jstate]}^2)$ converges or not.
%In the experiments, we observe that the individual regret of \emph{every} player goes to $-\infty$ when the algorithm diverges and the players benefit more from following the dynamics than staying at an equilibrium profile (\cref{fig:illustration}, third column).
%Nonetheless, there is no reason to believe that this should always be the case, and understanding the dynamics of the algorithm even in the case of no-convergence remains an important and challenging direction for future research.
%Nonetheless, even in this situation, we cannot exclude the existence of a Nash equlibrium $\bb{\sol}$ such that, if every player plays constantly their respective component of $\bb{\sol}$, they will individually suffer less than following the dynamics induced by the considered methods.
%Proving or disproving this conjecture is an interesting \textendash\ and fruitful \textendash\ direction for future research.

%\vspace{-2.2em}
%----------------------------------------------------------------------
%%% NUMERICS
%----------------------------------------------------------------------
% !TEX root = ../Main.tex

\begin{figure}[t]
    \centering
    \begin{subfigure}[b]{\linewidth}
    \includegraphics[width=0.325\linewidth]{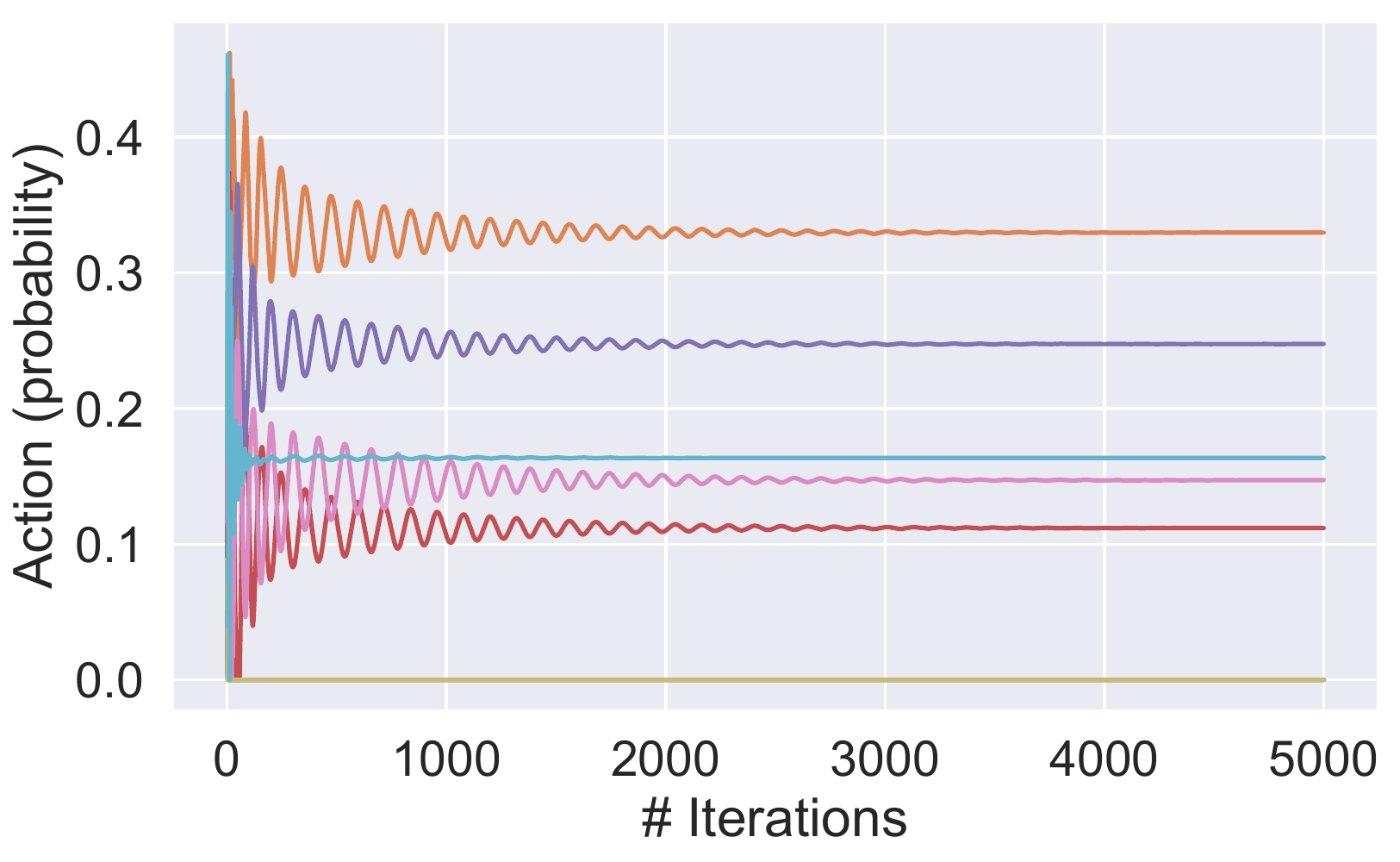}
    \includegraphics[width=0.325\linewidth]{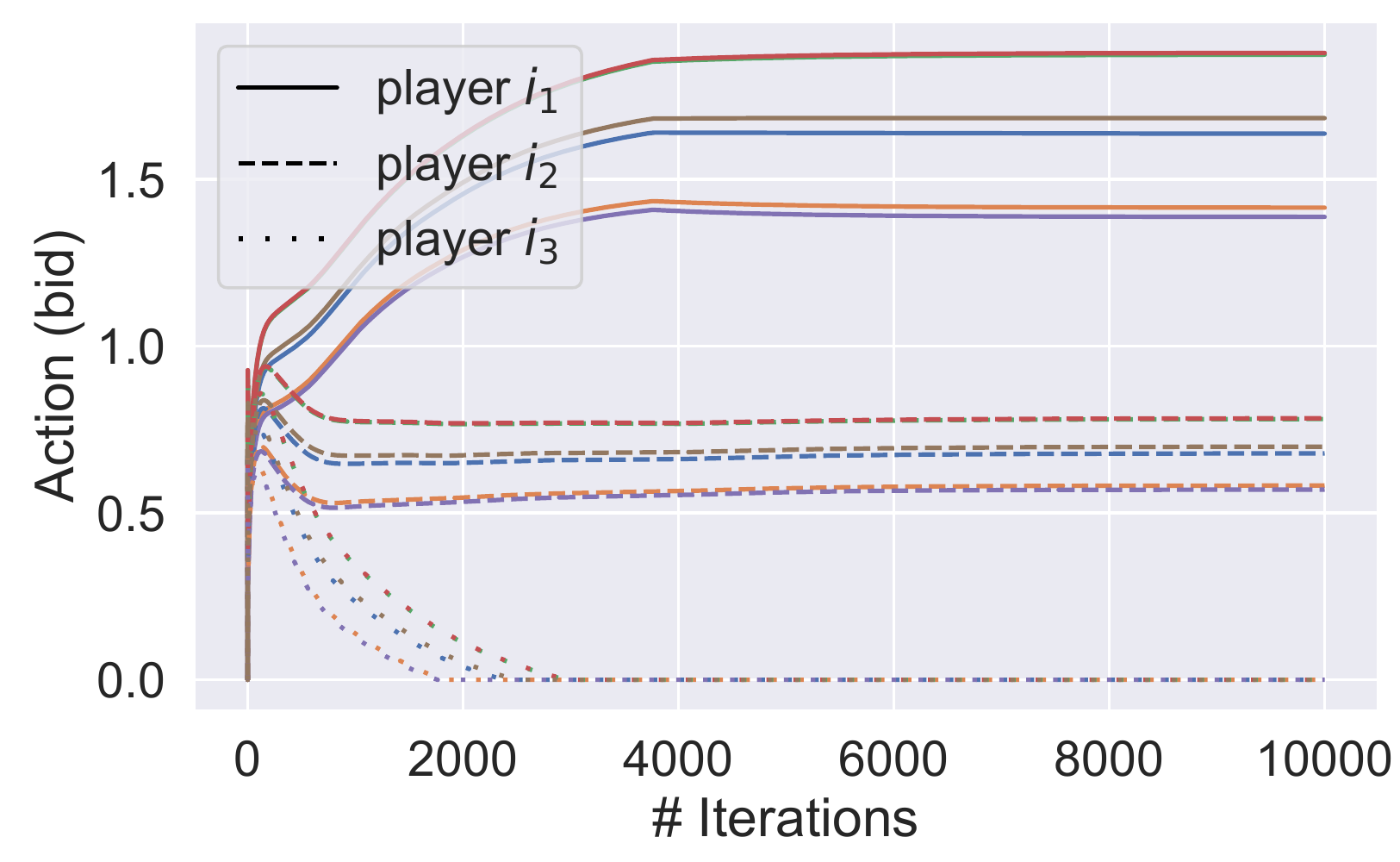}
    \includegraphics[width=0.325\linewidth]{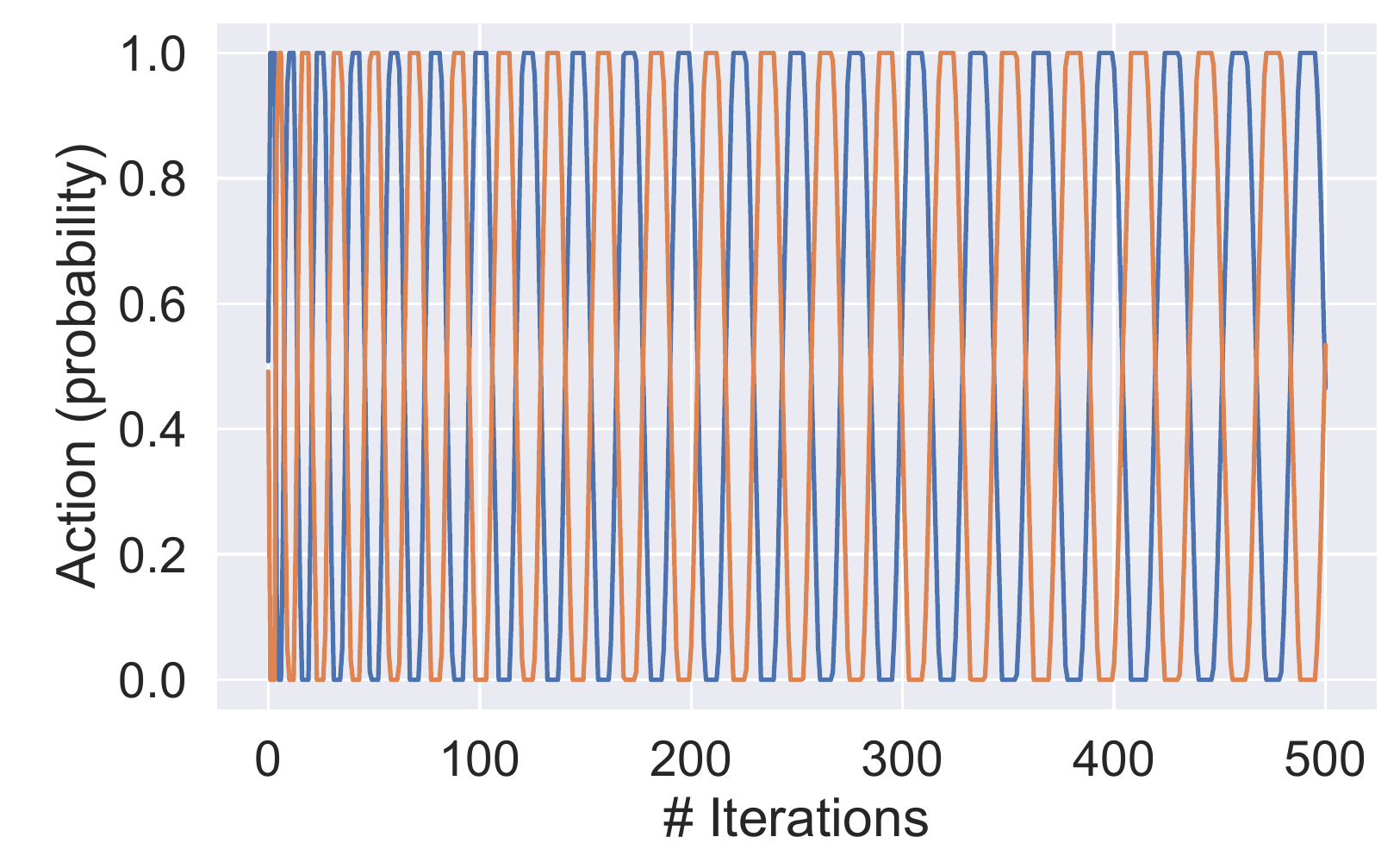}
    \end{subfigure}
    \begin{subfigure}[b]{\linewidth}
    \includegraphics[width=0.325\linewidth]{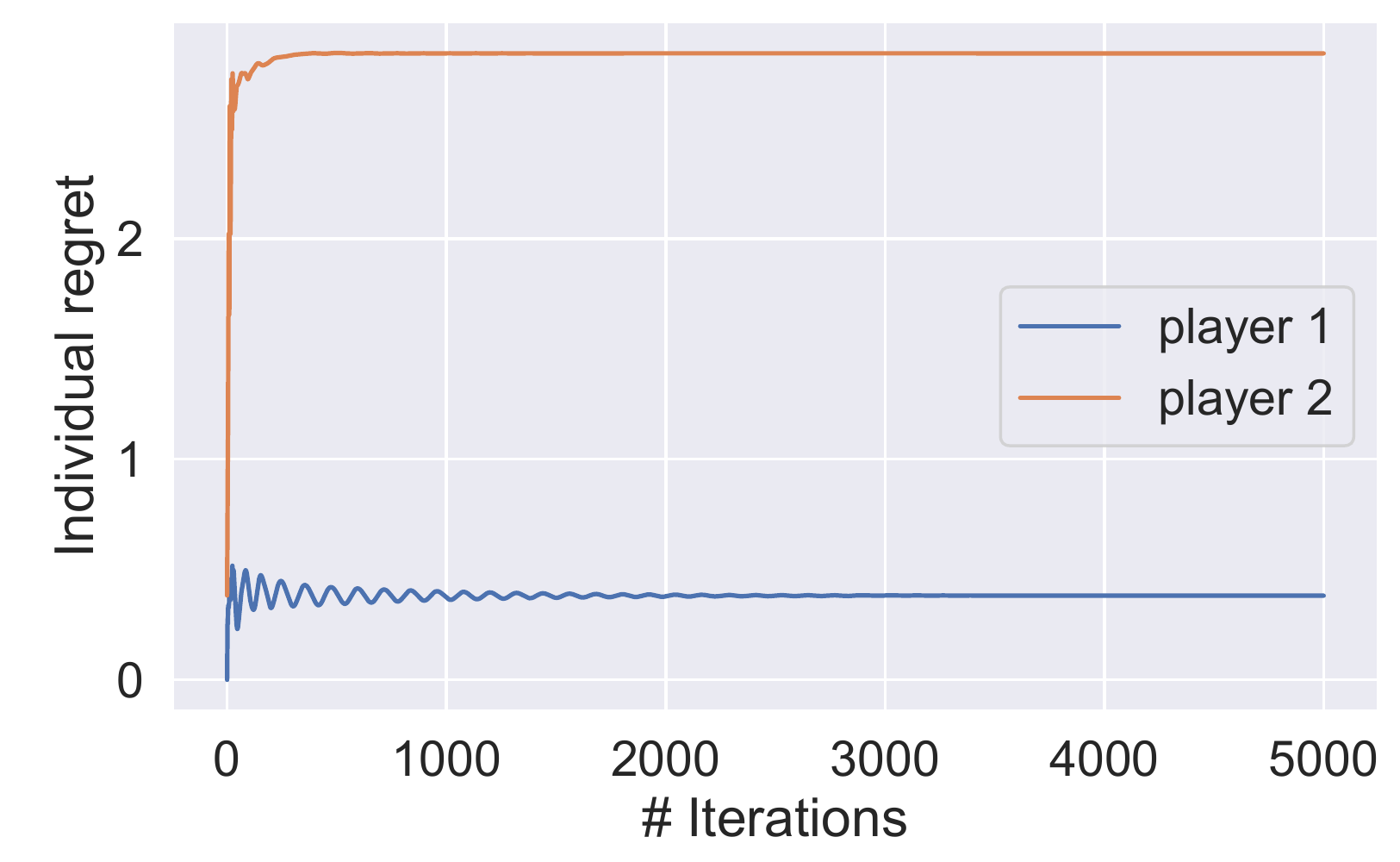}
    \includegraphics[width=0.325\linewidth]{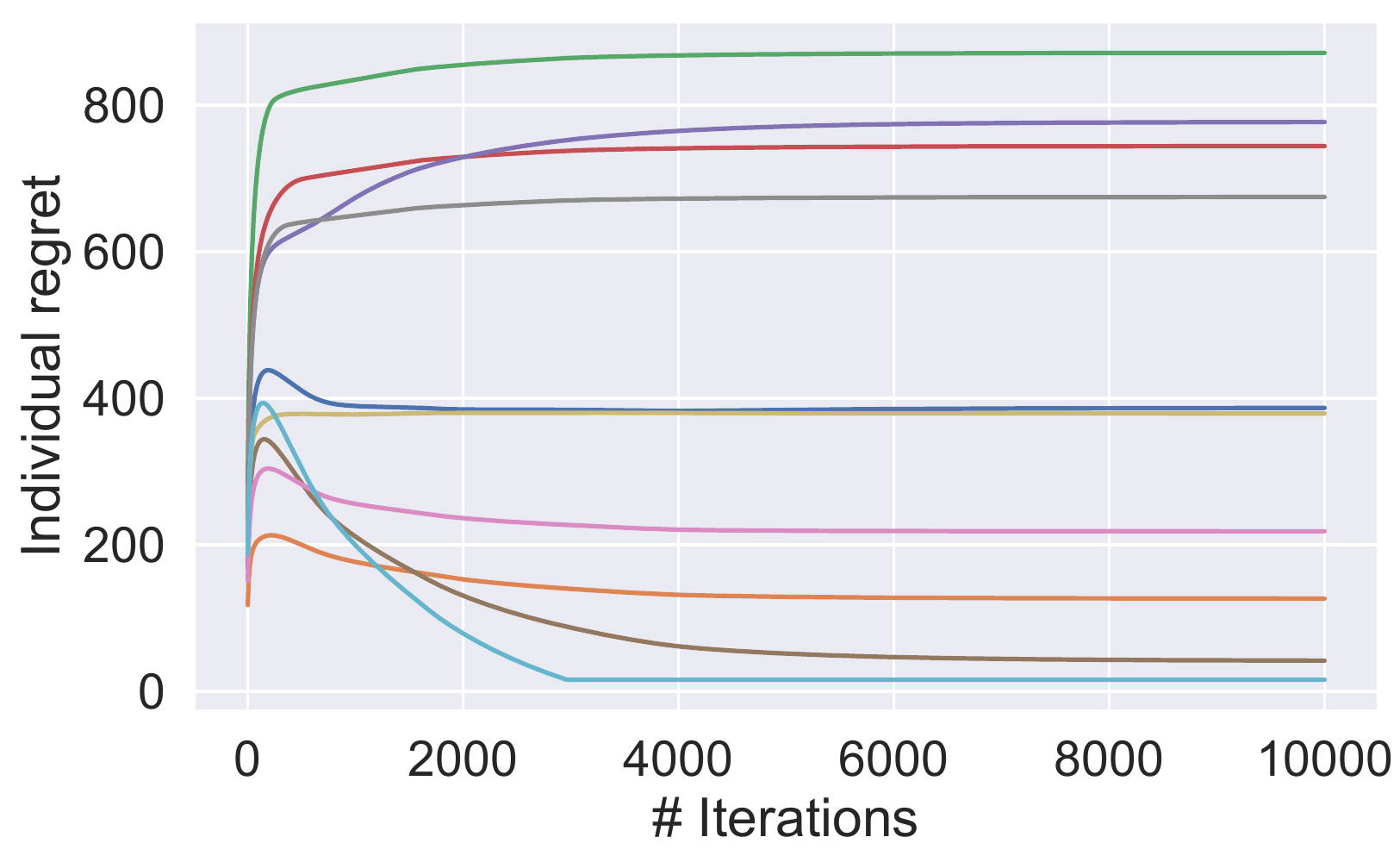}
    \includegraphics[width=0.325\linewidth]{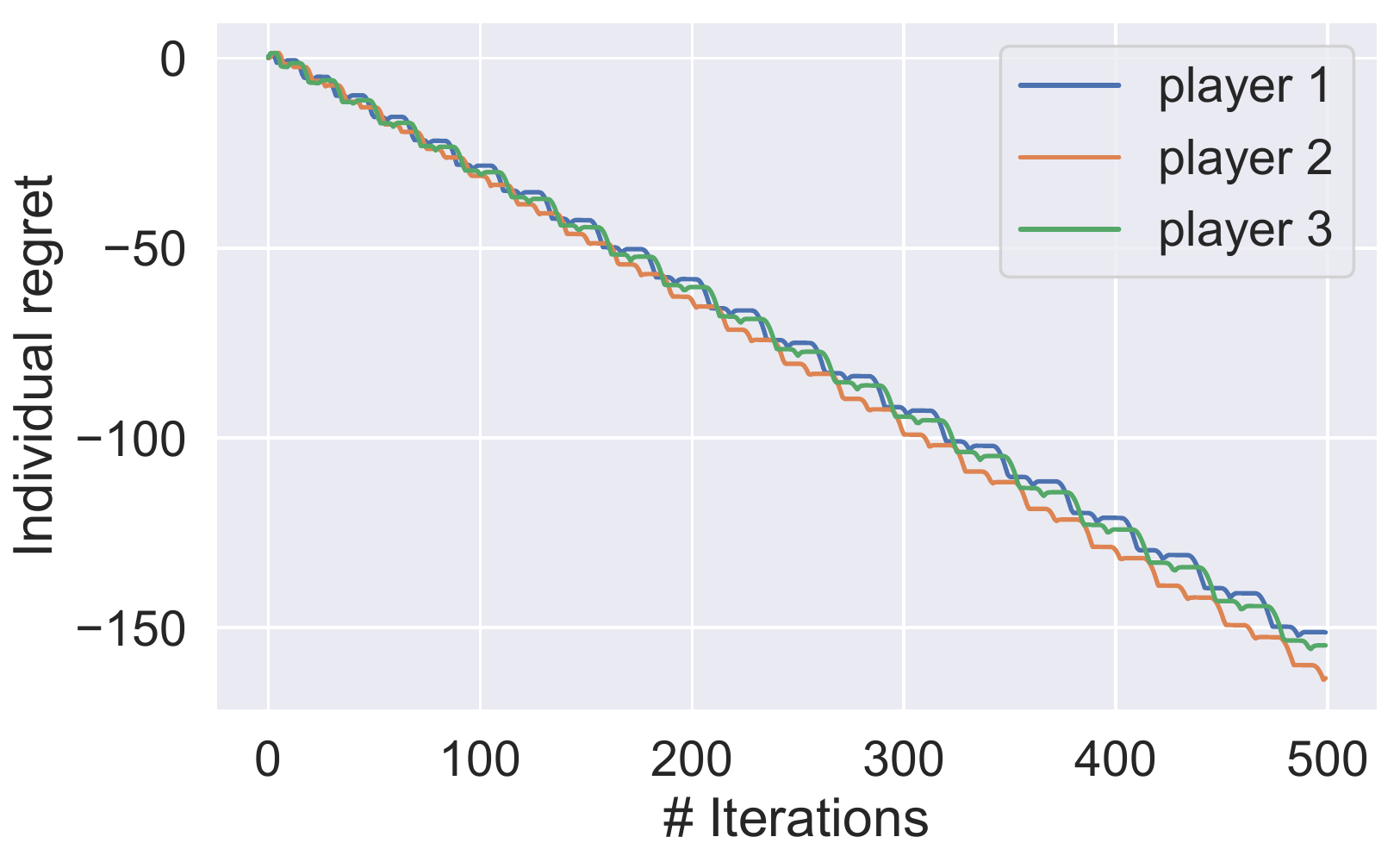}
    \end{subfigure}
    \captionsetup{labelsep=space}
    \caption{\textbf{- [Illustrative experiments]:} The realized actions (top, each line representing a coordinate of $\vwt[\action]$) and the individual regret (bottom) of a subset of players in a finite two-player zero-sum game (left), a resource allocation auction (middle), and a three-player matching-pennies game \cite{Jordan93} (right).
    All the players use either adaptive \ac{OptDA} or adaptive \ac{DS-OptMD} as their learning strategies.
    We observe convergence of the realized actions and the regrets in the first two examples.}
    %Full experimental details are provided in \cref{app:exps}}%\\[1pt]
    %We observe convergence of the realized actions and the regrets in both variationally stable games. In the third example, all the players oscillate between the two pure strategies, and have their individual regrets tend to minus infinity, which is consistent with our dichotomy result. (\cref{thm:dichotomy-general-sum}).}
    \label{fig:illustration}
    %\vspace{-2ex}
    \vspace{-1em}
\end{figure}

%----------------------------------------------------------------------
%%% NUMERICS
%----------------------------------------------------------------------
\section{Illustrative experiments}
\label{sec:numerics}
%In this appendix we provide all the necessary details for the replication of our experiments.
In this section we experimentally illustrate our theoretical results through \cref{ex:finite,ex:auction}.
%we consider two variationally stable games and a general-sum finite game, as detailed below
%To produce \cref{fig:illustration}, we considered three different setups as described below
Precisely, we investigate the following three different setups.

\begin{itemize}[itemsep=0cm, topsep=2pt]
    \item A finite zero-sum two-player game with $10\times10$ cost matrix whose elements are drawn uniformly at random from $[-1,+1]$: We let the two players play \ac{DS-OptMD} respectively with negative entropy and Euclidean regularizers.\footnote{The convergence of this particular situation can be proved following the proof of \cref{thm:converge-OMWU}.}
    \item A resource allocation auction (\cref{ex:auction}) with $6$ resources and $20$ bidders: 
    We fix $c_{\indg}=1$, draw $q_{\indg}$ and $\vw[\textsl{g}]$ uniformly at random from $[4,6]$, and draw $\vw[b]$ uniformly at random from $[5,10]$.
    Each player runs either \ac{OptDA} or \ac{DS-OptMD} and $\vw[\hreg](\point)=\norm{\point}_2^2/2$ for all $\play\in\players$.
    \item A three-player-matching-pennies game introduced in \cite{Jordan93}: Each player has two pure strategies. Player $1$ wants to match the pure strategy of player $2$; player $2$ wants to match the pure strategy of player $3$; and player $3$ wants to match the opposite of the pure strategy of player $1$. Each player receives a loss of $-1$ if they match as desired, and $1$ otherwise.
    It is straightforward to see that the unique equilibrium is achieved when all the players uniformly randomize.
    In this game, we let the three players run \ac{DS-OptMD} with Euclidean regularizer.
\end{itemize}

As for the learning rates, we fix $\vw[\tau]=1$ and use the Euclidean norm $\vwp[\norm{\cdot}]=\norm{\cdot}_2$ throughout.
The results are plotted in \cref{fig:illustration}.
The first two games that we consider are variationally stable, and as predicted by our analysis, we observe the convergence of the iterates and the boundedness of the individual regrets.
For the three-player-matching-pennies game, all the players oscillate between the two pure strategies, and have their individual regrets tend to minus infinity. This is consistent with our dichotomy result \cref{thm:dichotomy-general-sum}.
%In \cref{fig:illustration} we plot the realized actions and the individual regrets for a subset of players. We observe the convergence of both quantities, and in particular, the individual regrets are bounded as predicted by our analysis.
%In the rightmost column of \cref{fig:illustration}, we provide an example in which the algorithms under study do not lead to convergence. We let the three players of a three-player-matching-pennies game \cite{Jordan93} run \ac{DS-OptMD} with Euclidean regularizer. All the players oscillate between the two pure strategies, and have their individual regrets tend to minus infinity. This is consistent with our dichotomy result (\cref{thm:dichotomy-general-sum}).

%----------------------------------------------------------------------
%%% Conclusion
%----------------------------------------------------------------------
\section{Concluding remarks}
\label{sec:conclusion}
In this work, we have presented a family of adaptive algorithms for online learning in continuous games that solely utilizes local information received by each player.
We showed that these algorithms achieve optimal regret bounds under various conditions, and more importantly, lead to Nash equilibrium when employed by all the players in a variationally stable game.

Many interesting questions remain to be answered.
For example, it is known that optimistic algorithms can achieve individual regret much smaller than $\bigoh(\sqrt{\nRuns})$ in general-sum finite games when used by all players \cite{SALS15,CP20}.
Is this feature shared by our algorithm?
\cref{thm:dichotomy-general-sum} and preliminary experiments suggest that this could be the case.
Nonetheless, even this property does not imply that the algorithm effectively generates a `good' sequence of play.
%that is beneficial for the players, since 
In fact, in some cases, the players may benefit more from staying at a Nash equilibrium rather than following a trajectory that lead to $-\infty$ individual regret, and we believe that understanding the dynamics of the algorithm even in the case of no-convergence is an important and challenging direction for future research.

%----------------------------------------------------------------------
%%% ACKNOWLEDGMENTS
%----------------------------------------------------------------------
\section*{Acknowledgments}
%----------------------------------------------------------------------
%%% THANKS
%----------------------------------------------------------------------
% !TEX root = ../Main.tex
%
%
This research was partially supported by the COST Action CA16228 ``European Network for Game Theory'' (GAMENET),
and
the French National Research Agency (ANR) in the framework of
the ``Investissements d'avenir'' program (ANR-15-IDEX-02),
the LabEx PERSYVAL (ANR-11-LABX-0025-01),
MIAI@Grenoble Alpes (ANR-19-P3IA-0003),
and the grants ORACLESS (ANR-16-CE33-0004) and ALIAS (ANR-19-CE48-0018-01).

%----------------------------------------------------------------------
%%% REFERENCES
%----------------------------------------------------------------------
\bibliography{references,Bibliography}

\newpage
\setlength\parindent{0pt}
\appendix
%\section{Experimental details}
%\label{app:exps}
%\input{Appendix/app-exps}

\section{Proof of \cref{lem:template-descent}}
\label{app:mirror}
In this appendix, we present several basic properties of the Bregman divergence, the mirror map, and the Fenchel coupling, before proceeding to prove \cref{lem:template-descent}.
For ease of notation, the player index will be dropped in the notation.
In particular, we will write $\points$ and $\hreg$ respectively for the player's action space and the associated regularizer, and we assume that $\hreg$ is $1$-stronlgy convex relative to an ambient norm $\norm{\cdot}$.
%, $\breg$, $\mirror$, and $\fench$ respectively for the regularizer and its induced Bregman divergence, mirror map, and Fenchel coupling.
%With a slight abuse of notation, we will also write $\points$ for the player's action space.

\subsection{Bregman divergence, mirror map, and Fenchel coupling}

We first recall the definition of the Bregman divergence and the Fenchel coupling,
\begin{equation}
    \notag
    \begin{aligned}
    \breg(\arpoint,\point) &= \hreg(\arpoint) - \hreg(\point) - \product{\grad\hreg(\point)}{\arpoint-\point},
    \\
    \fench(\arpoint,\dvec) &= \hreg(\arpoint) + \dual{\hreg}(\dvec) - \product{\dvec}{\arpoint},
    \end{aligned}
\end{equation}
where $\dual{\hreg}\from\dvec\to\max_{\point\in\points}\product{\dvec}{\point}-\hreg(\point)$ is the Fenchel conjugate of $\hreg$.
We also recall that the mirror map induced by $\hreg$ is defined as
\begin{equation}
    \notag
    \mirror(\dvec) = \argmin_{\point\in\points}\product{-\dvec}{\point} + \hreg(\point).
\end{equation}
The auxiliary results that we are going to present below concerning these three quantities are not new (see \eg \cite{JNT11,NJLS09,MZ19} and references therein); however, the set of hypotheses used to obtain them varies widely in the literature, so we still provide the proofs for the sake of completeness.

To begin, our first lemma concerns the optimality condition of the mirror map.

\begin{lemma}
\label{lem:mirror-optimality}
Let $\hreg$ be a regularizer on $\points$. Then, for all $\point\in\dom\subd\hreg$ and all $\dvec\in\dspace$, we have
\begin{equation}
    \notag
    \point=\mirror(\dvec) \iff \dvec \in \subd\hreg(\point).
\end{equation}
Moreover, if $\point=\mirror(\dvec)$, it holds for all $\arpoint\in\points$ that
\begin{equation}
    \notag
    \product{\grad\hreg(\point)}{\point-\arpoint}\le\product{\dvec}{\point-\arpoint}.
\end{equation}
\end{lemma}
\begin{proof}
For the first claim, we have by the definition of the mirror map 
%note that $\point$ solves \eqref{eq:dconj} 
%$\dvec - \subd\hreg(\point) \ni 0$,
$\point=\mirror(\dvec)$ if and only if $0\in\subd\hreg(\point)-\dvec$, \ie $\dvec\in\partial\hreg(\point)$.
For the second claim, it suffices to show it holds for all $\base\in\relint\points$ (by continuity).
To do so, we can define
\begin{equation}
\notag
    \phi(t)
	= \hreg(\point + t(\base-\point))
	- \bracks{\hreg(\point) +  \braket{\dvec}{\point + t(\base-\point)}}.
\end{equation}
Since $\hreg$ is strongly convex and $\dvec\in\partial\hreg(\point)$ by the previous claim, it follows that $\phi(t)\geq0$ with equality if and only if $t=0$.
Moreover, as $\relint\points\subset\dom\partial\hreg$, $\nabla\hreg(\point + t(\base-\point))$ is well-defined and $\psi(t) = \braket{\nabla\hreg(\point + t(\base-\point)) - \dvec}{\base - \point}$ is a continuous selection of subgradients of $\phi$.
Given that $\phi$ and $\psi$ are both continuous on $[0,1]$, it follows that $\phi$ is continuously differentiable and $\phi' = \psi$ on $[0,1]$.
%Thus, with $\phi$ convex and
Thus, with $\phi(t) \geq 0 = \phi(0)$ for all $t\in[0,1]$, we conclude that $\phi'(0) = \braket{\nabla\hreg(\point) - \dvec}{\base - \point} \geq 0$, from which our claim follows.
\end{proof}

We continue with the ``three-point identity'' \cite{CT93} which will be used to derive the recurrent relationship between the divergence measures of different steps.

\begin{lemma}
\label{lem:3point}
Let $\hreg$ be a regularizer on $\points$. Then, for all $\arpoint\in\points$ and all $\point,\alt{\point}\in\dom\subd\hreg$, we have
\begin{equation}
    \label{eq:breg-3point}
    \product{\grad\hreg(\alt{\point})-\grad\hreg(\point)}{\point-\arpoint}
    =\breg(\arpoint,\alt{\point})-\breg(\arpoint,\point)-\breg(\point,\alt{\point}).
\end{equation}
Similarly, writing $\point=\mirror(\dvec)$, for all $\arpoint\in\points$ and all $\dvec,\alt{\dvec}\in\dspace$, we have
\begin{equation}
    \label{eq:fench-3point}
    \product{\alt{\dvec}-\dvec}{\point-\arpoint}
    =\fench(\arpoint,\alt{\dvec})-\fench(\arpoint,\dvec)-\fench(\point,\alt{\dvec}).
\end{equation}
\end{lemma}
\begin{proof}
We start with the Bregman version.
By definition,
\begin{equation}
\notag
\begin{aligned}
\breg(\base,\pointalt)
	&= \hreg(\base) - \hreg(\pointalt) - \braket{\nabla\hreg(\pointalt)}{\base - \pointalt}
	\\
\breg(\base,\point)\hphantom{'}
	&= \hreg(\base) - \hreg(\point) - \braket{\nabla\hreg(\point)}{\base - \point}
	\\
\breg(\point,\pointalt)
	&= \hreg(\point) - \hreg(\pointalt) - \braket{\nabla\hreg(\pointalt)}{\point - \pointalt}.
\end{aligned}
\end{equation}
The result then follows by adding the two last lines and subtracting the first.
On the other hand, in order to show the Fenchel coupling version we write
%have the following: By definition, we get:
\begin{equation}
\notag
\begin{aligned}
\fench(\base,\alt{\dvec})
	&= \hreg(\base) +\hreg^{\ast}(\alt{\dvec}) - \braket{\alt{\dvec}}{\base}
	\\
\fench(\base,\dvec)\hphantom{'}
	&= \hreg(\base) + \hreg^{\ast}(\dvec) - \braket{\dvec}{\base}.
	\\
\end{aligned}
\end{equation}
Then, by subtracting the above we obtain
\begin{align}
\fench(\base,\alt{\dvec})-\fench(\base,\dvec)
	&= \hreg(\base)
		+ \hreg^{\ast}(\alt{\dvec})
		- \braket{\alt{\dvec}}{\base}
		- \hreg(\base)
		- \hreg^{\ast}(\dvec)
		+ \braket{\dvec}{\base}
	\notag\\
	&=\hreg^{\ast}(\alt{\dvec})
		- \hreg^{\ast}(\dvec)
		- \braket{\alt{\dvec} - \dvec}{\base}
	\notag\\
	&=\hreg^{\ast}(\alt{\dvec})
		- \braket{\dvec}{\mirror(\dvec)}
		+ \hreg(\mirror(\dvec))
		- \braket{\alt{\dvec} - \dvec}{\base}
	\notag\\
	&=\hreg^{\ast}(\alt{\dvec})
		- \braket{\dvec}{\point}
		+ \hreg(\point)
		- \braket{\alt{\dvec} - \dvec}{\base}
	\notag\\
	&= \hreg^{\ast}(\alt{\dvec})
		+ \braket{\alt{\dvec} - \dvec}{\point}
		- \braket{\alt{\dvec}}{\point}
		+ \hreg(\point)
		- \braket{\alt{\dvec} - \dvec}{\base}
	\notag\\
	&= \fench(\point,\alt{\dvec})
		+ \braket{\alt{\dvec} - \dvec}{\point - \base}
	\notag
\end{align}
and our proof is complete.
\end{proof}

Since $\point=\mirror(\grad\hreg(\point))$ and $\fench(\arpoint,\grad\hreg(\point))=\breg(\arpoint,\point)$, the identity \eqref{eq:breg-3point} is indeed a special case of \eqref{eq:fench-3point}.
In the general case, the Fenchel coupling and the Bregman divergence can be related by the following lemma.

\begin{lemma}
\label{lem:fench-breg}
Let $\hreg$ be a regularizer on $\arpoints$.
%that is $1$-strongly convex relative to $\norm{\cdot}$.
Then, for all $\arpoint\in\arpoints$ and $\dvec\in\dspace$, it holds
\begin{equation}
    \notag
    \fench(\arpoint,\dvec)\ge\breg(\arpoint,\mirror(\dvec))\ge\frac{\norm{\arpoint-\mirror(\dvec)}^2}{2}.
\end{equation}
\end{lemma}
\begin{proof}
For the first inequality we have,
\begin{align*}
\fench(\arpoint,\dvec)&=\hreg(\arpoint)+\hreg^{\ast}(\dvec)-\braket{\dvec}{\arpoint}\\
&=\hreg(\arpoint)-\hreg(\mirror(\dvec))+\braket{\dvec}{\mirror(\dvec)}+\braket{\dvec}{-\arpoint}\\
%&=\hreg(\arpoint)-\hreg(\mirror(\dvec))+\braket{\dvec}{\mirror(\dvec)-\arpoint}\\
&=\hreg(\arpoint)-\hreg(\mirror(\dvec))-\braket{\dvec}{\arpoint-\mirror(\dvec)}
\end{align*}
Since $\dvec \in \partial \hreg(\mirror(\dvec))$, by \cref{lem:mirror-optimality} we get
\begin{equation}
\notag
\braket{\nabla \hreg(\mirror(\dvec))}{\mirror(\dvec)-\arpoint}\leq \braket{\dvec}{\mirror(\dvec)-\arpoint}
\end{equation}
With all the above we then have
\begin{align*}
\fench(\arpoint,\dvec)&=\hreg(\arpoint)-\hreg(\mirror(\dvec))-\braket{\dvec}{\arpoint-\mirror(\dvec)}\\
&\geq \hreg(\arpoint)-\hreg(\mirror(\dvec))-\braket{\nabla \hreg(\mirror(\dvec))}{\arpoint-\mirror(\dvec)}\\
&=\breg(\arpoint,\mirror(\dvec))
\end{align*}
and the result follows. 
The second inequality follows directly from the fact that the regularizer $\hreg$ is $1$-strongly convex relative to $\norm{\cdot}$. 
\end{proof}

\begin{remark*}
From the above proof we see that $\fench(\arpoint,\dvec)=\hreg(\arpoint)-\hreg(\mirror(\dvec))-\product{\dvec}{\arpoint-\mirror(\dvec)}$.
Since $\dvec\in\subd\hreg(\mirror(\dvec))$ by \cref{lem:mirror-optimality}, Fenchel coupling is also closely related to a generalized version of Bregman divergence
which is defined for $\arpoint\in\points$, $\point\in\dom\subd\hreg$, and $\gvec\in\subd\hreg(\point)$ by $\breg(\arpoint,\point;\gvec)=\hreg(\arpoint)-\hreg(\point)-\product{\gvec}{\arpoint-\point}$. This definition is formally introduced in \cite{JKM19}, but its use in the literature can be traced back to much earlier work such as \cite{Kiw97}.
\end{remark*}

\begin{remark*}
By using $\point=\mirror(\grad\hreg(\point))$ and $\fench(\arpoint,\grad\hreg(\point))=\breg(\arpoint,\point)$, we see immediately that Bregman reciprocity is implied by Fenchel reciprocity. 
\end{remark*}

\subsection{Optimistic dual averaging}

We first prove \cref{lem:template-descent} for \acf{OptDA}. Its update writes 
\begin{equation}
    \tag{OptDA}
    %\label{eq:OptDA}
    \begin{aligned}
    \vt[\state] &= \argmin_{\point\in\points} \sum_{\runalt=1}^{\run-1}\product{\vt[\gvec][\runalt]}{\point} + \vt[\regpar]\hreg(\point),
    \\
    \inter &= \argmin_{\point\in\points}\thinspace
    \product{\last[\gvec]}{\point} + \vt[\regpar]\breg(\point, \vt[\state]).
\end{aligned}
\end{equation}
Let us define $\vt[\dstate]=(-1/\vt[\regpar])\sum_{\runalt=1}^{\run-1}\vt[\gvec][\runalt]$ so that $\vt[\state]=\mirror(\vt[\dstate])$.
For any $\arpoint\in\points$, we can apply the three-point identity for Fenchel coupling \eqref{eq:fench-3point} to the update of $\update[\state]$ and get
\begin{equation}
    \notag
    \begin{aligned}
    \product{\vt[\gvec]}{\update-\arpoint}
    &=\product{\vt[\regpar]\vt[\dstate]-\update[\regpar]\update[\dstate]}{\update-\arpoint}\\
    &=\vt[\regpar]\product{\vt[\dstate]-\update[\dstate]}{\update-\arpoint}
    +(\update[\regpar]-\vt[\regpar])\product{0-\update[\dstate]}{\update-\arpoint}\\
    &= \vt[\regpar]
        (\fench(\arpoint,\vt[\dstate])-\fench(\arpoint,\update[\dstate])-\fench(\update,\vt[\dstate]))\\
    &~~+ (\vt[\regpar][\run+1]-\vt[\regpar])
        (\fench(\arpoint,0)-\fench(\arpoint,\update[\dstate])-\fench(\update,0)).
    \end{aligned}
\end{equation}
As $\fench(\arpoint,0)=\hreg(\arpoint)-\hreg(\mirror(0))=\hreg(\arpoint)-\min\hreg$,
writing $\armeasure(\arpoint)=\hreg(\arpoint)-\min\hreg$, the above gives
\begin{equation}
    \label{eq:dual-avg-inq}
    \product{\vt[\gvec]}{\update-\arpoint}
    \le \vt[\regpar]\fench(\arpoint,\vt[\dstate]) - \update[\regpar]\fench(\arpoint,\update[\dstate])
    -\vt[\regpar]\fench(\update,\vt[\dstate]) + (\update[\regpar]-\vt[\regpar])\armeasure(\arpoint).
\end{equation}
As for the update of $\inter$, we note that $\inter=\mirror(\grad\hreg(\current)-\last[\gvec]/\vt[\regpar])$.
Therefore, invoking \cref{lem:mirror-optimality} gives
\begin{equation}
    \notag
    \product{\grad\hreg(\inter)}{\inter-\arpoint}\le
    \big\langle\grad\hreg(\current)-\frac{\last[\gvec]}{\vt[\regpar]},\inter-\arpoint\big\rangle.
\end{equation}
For the specific choice $\arpoint\subs\update$, using the three-point identity for Bregman divergence \eqref{eq:breg-3point} we obtain
\begin{equation}
    \label{eq:mirror-descent-inq}
    \begin{aligned}[b]
    \product{\last[\gvec]}{\inter-\update}
    &\le \vt[\regpar] \product{\grad\hreg(\current)-\grad\hreg(\inter)}{\inter-\update}\\
    &= \vt[\regpar](\breg(\update,\current) - \breg(\update,\inter) - \breg(\inter,\current)).
    \end{aligned}
\end{equation}
Since $\fench(\update,\vt[\dstate])\ge\breg(\update,\current)$ by \cref{lem:fench-breg}, combining \eqref{eq:dual-avg-inq} and \eqref{eq:mirror-descent-inq} leads to
\begin{equation}
    %\label{eq:regret-OptDA-single}
    \notag
    \begin{aligned}[b]
    \product{\vt[\gvec]}{\inter-\arpoint}
    &= \product{\vt[\gvec]-\last[\gvec]}{\inter-\update}
    + \product{\last[\gvec]}{\inter-\update}
    + \product{\vt[\gvec]}{\update-\arpoint}\\
    &\le \vt[\regpar]\fench(\arpoint,\vt[\dstate]) - \update[\regpar]\fench(\arpoint,\update[\dstate])
    + (\update[\regpar]-\vt[\regpar])\armeasure(\arpoint)\\
    &~~ + \product{\vt[\gvec]-\last[\gvec]}{\inter-\update}
    - \vt[\regpar]\breg(\update,\inter) - \vt[\regpar]\breg(\inter,\current).
    \end{aligned}
\end{equation}
This proves the generated iterates of \ac{OptDA} satisfy \eqref{eq:template-descent-main} with $\vwt[\estseq]=\vw[\fench](\cdot,\vwt[\dstate])$ and $\vw[\armeasure]=\vw[\hreg]-\min\vw[\hreg]$.

\subsection{Dual stabilized optimistic mirror descent}

We next prove the generated iterates of \acf{DS-OptMD} satisfy \eqref{eq:template-descent-main} with $\vwt[\estseq]=\vw[\breg](\cdot,\vwt[\state])$ and $\vw[\armeasure]=\vw[\breg](\cdot,\vwt[\state][\play][\start])$.
The algorithm is stated recursively as
\begin{equation}
    \tag{DS-OptMD}
    \begin{aligned}
    \inter &= \argmin_{\point\in\points}\thinspace
    \product{\last[\gvec]}{\point} + \vt[\regpar]\breg(\point, \vt[\state]),
    \\
    \update[\state] &= \argmin_{\point\in\points}\thinspace
    \product{\vt[\gvec]}{\point} + \vt[\regpar]\breg(\point, \vt[\state])
    + (\vt[\regpar][\run+1]-\vt[\regpar])\breg(\point, \vt[\state][\start]).
\end{aligned}
\end{equation}
By definition of the Bregman divergence and the mirror map, the second step is equivalent to
\begin{equation}
    \notag
    \update[\state] = \mirror\left(
    \frac{\vt[\regpar]}{\update[\regpar]}\grad\hreg(\current)
    + (1-\frac{\vt[\regpar]}{\update[\regpar]})\grad\hreg(\vt[\state][\start])-\frac{\vt[\gvec]}{\update[\regpar]}
    \right).
\end{equation}
This shows that the update of $\update$ consists in fact of a mixing step in the dual space with weight $\vt[\regpar]/\update[\regpar]$ followed by a standard mirror descent step. Applying \cref{lem:mirror-optimality} gives
\begin{equation}
    \notag
    \product{\grad\hreg(\update)}{\update-\arpoint}\le
    \left\langle\frac{\vt[\regpar]}{\update[\regpar]}\grad\hreg(\current)
    + (1-\frac{\vt[\regpar]}{\update[\regpar]})\grad\hreg(\vt[\state][\start])-\frac{\vt[\gvec]}{\update[\regpar]},\update-\arpoint\right\rangle.
\end{equation}
We rearrange the terms and use the three-point identity \eqref{eq:breg-3point} to get
\begin{equation}
    \label{eq:base-descent-DS}
    \begin{aligned}[b]
    \product{\vt[\gvec]}{\update[\state]-\arpoint}
    &\le\vt[\regpar]\product{\grad\hreg(\vt[\state])-\grad\hreg(\update[\state])}{\update-\arpoint}\\
    &~~+(\update[\regpar]-\vt[\regpar])\product{\grad\hreg(\vt[\state][\start])-\grad\hreg(\update[\state])}{\update-\arpoint}
    \\
    &\le \vt[\regpar]
    (\breg(\arpoint,\vt[\state])-\breg(\arpoint,\update)-\breg(\update,\current))\\
    &~~+ (\vt[\regpar][\run+1]-\vt[\regpar])
    (\breg(\arpoint,\vt[\state][\start])-\breg(\arpoint,\update)-\breg(\update,\vt[\state][\start]))
    \end{aligned}
\end{equation}
Since $\inter$ is computed exactly as in \eqref{eq:OptDA}, inequality \eqref{eq:mirror-descent-inq} still holds.
We conclude by putting together \eqref{eq:base-descent-DS} and \eqref{eq:mirror-descent-inq}
\begin{equation}
    %label{eq:regret-OMPDS-single}
    \notag
    \begin{aligned}[b]
    \product{\vt[\gvec]}{\inter-\arpoint}
    &= \product{\vt[\gvec]-\last[\gvec]}{\inter-\update}
    + \product{\last[\gvec]}{\inter-\update}
    + \product{\vt[\gvec]}{\update-\arpoint}\\
    &\le \vt[\regpar]\breg(\arpoint,\vt[\state]) - \update[\regpar]\breg(\arpoint,\update[\state])
    + (\update[\regpar]-\vt[\regpar])\breg(\arpoint,\vt[\state][\start])\\
    &~~ + \product{\vt[\gvec]-\last[\gvec]}{\inter-\update}
    - \vt[\regpar]\breg(\update,\inter) - \vt[\regpar]\breg(\inter,\current).
    \end{aligned}
\end{equation}
This prove \cref{lem:template-descent} for \ac{DS-OptMD}. \qed

\section{Adaptive optimistic algorithms}
\label{app:template}
In the remainder of the appendix, we consider a broad family of algorithms which we refer to as ``optimistic and compatible with dynamic learning rate''.
Given a regularizer $\hreg$ and a sequence of non-decreasing positive numbers $\seqinf[\vt[\regpar]]$, % and whose associated Bregman divergence is denoted by $\breg$,
an algorithm of this family produces a sequence of iterates $\seqinf[\vt[\state][\runalt]][\runalt][\N/2]$ satisfying that

%We say that a learning algorithm is optimistic and compatible with dynamic learning rates if, given a sequence of increasing positive numbers $\seqinf[\vt[\regpar]]$, a predictor sequence $\seqinf[\inter[\appr]]$, and a regularizer $\hreg$ whose associated Bregman divergence is denoted by $\breg$, it produces a sequence of iterates $\seqinf[\vt[\state]][\runalt][\N/2]$ satisfying

\begin{enumerate}[leftmargin=*]
    \item For some non-negative continuous functions $\seqinf[\vt[\estseq]]$ and $\armeasure$ defined on $\vw[\points]$ (the player's action set), we have, for all $\arpoint\in\vw[\points]$,
    \begin{equation}
        \label{eq:template-descent}
        \begin{aligned}[b]
        \update[\regpar]\update[\estseq](\arpoint)
        &\le
        \current[\regpar]\current[\estseq](\arpoint)
        - \product{\vt[\gvec]}{\inter-\arpoint}
        + (\update[\regpar]-\vt[\regpar])\armeasure(\arpoint)\\
        &~+ \product{\vt[\gvec]-\last[\gvec]}{\inter-\update}
        - \vt[\regpar]\breg(\update,\inter) - \vt[\regpar]\breg(\inter,\current),
        \end{aligned}
    \end{equation}
    where $\breg$ is the associated Bregman divergence of $\hreg$.
    \item For every $\run\in\N$, $\inter$ is generated by
    \begin{equation}
    \notag
    \inter = \argmin_{\point\in\vw[\points]}\thinspace
    \product{\last[\gvec]}{\point} + \vt[\regpar]\breg(\point, \vt[\state]).
\end{equation}
\end{enumerate}
By replacing $\armeasure$ with $\max(\armeasure,\vt[\estseq][\start])$ if needed, we may assume $\vt[\estseq][\start]\le\armeasure$ without loss of generality.
Thanks to \cref{lem:template-descent}, we know that both \eqref{eq:DS-OptMD} and \eqref{eq:OptDA} are optimistic and compatible with dynamic learning rate.
As another example, it can be proved in a similar way that \eqref{eq:OptMD} is optimistic and compatible with dynamic learning rate if $\sup_{\arpoint,\point\in\vw[\points]}\breg(\arpoint,\point)<+\infty$.
In this case, $\vt[\estseq]=\breg(\cdot,\vt[\state])$ and $\armeasure\equiv\sup_{\arpoint,\point\in\vw[\points]}\breg(\arpoint,\point)$.

Since the player's cost function is convex with respect to its own action by \cref{asm:convexity+smoothness}, their regret can be bounded by the linearized regret,\footnote{This argument will be used implicitly throughout the proofs.} which, using \eqref{eq:template-descent}, can be in turn bounded by
\begin{equation}
    \label{eq:template-regret}
    \begin{aligned}[b]
    \sum_{\run=1}^{\nRuns}\product{\vt[\gvec]}{\inter[\state]-\arpoint}
    &\le \vt[\regpar][\nRuns+1]\armeasure(\arpoint)
    - \vt[\regpar][\nRuns+1]\vt[\estseq][\nRuns+1](\arpoint)
    + \sum_{\run=1}^{\nRuns}
    \product{\vt[\gvec]-\last[\gvec]}{\inter[\state]-\update[\state]}\\
    &~-\sum_{\run=1}^{\nRuns}\vt[\regpar] \left(
    \breg(\update[\state],\inter[\state])
    +\breg(\inter[\state],\vt[\state])\right).
    \end{aligned}
\end{equation}
To further obtain \eqref{eq:template-descent-main}, we need to invoke Young's inequality and the strong convexity of $\hreg$. More details can be found in the proof of \cref{thm:social-regret-bounded} (\cref{app:social-regret-proof}).
For those results that require the reciprocity conditions, this translates into the following requirement on $\vt[\estseq](\arpoint)$.
%we will make the following assumption relating $\vt[\estseq](\arpoint)$ to the distance between $\vt[\state]$ and $\arpoint$.

\begin{assumption}
\label{asm:distance-funciton}
For some norm $\norm{\cdot}$ and its associated distance function $\dist$, the sequence $\seqinf[\vt[\estseq]]$ satisfies
\begin{enumerate}[label=(\alph*)]
    \item For any $\run\in\N$, $\vt[\estseq](\arpoint)\ge(1/2)\norm{\vt[\state]-\arpoint}^2$.
    \label{asm:distance-funciton-a}
    \item For any compact set $\cpt\in\vw[\points]$ and $\sbradius>0$, there exists $\sradius>0$ such that if $\dist(\vt[\state],\cpt)\le\sradius$ then $\vt[\estseq](\cpt)\defeq\min_{\arpoint\in\cpt}\vt[\estseq](\arpoint)\le\sbradius$.
    \label{asm:distance-funciton-b}
\end{enumerate}
\end{assumption}

For $\vt[\estseq]=\breg(\cdot,\vt[\state])$ and $\vt[\estseq]=\fench(\cdot,\vt[\dstate])$, \cref{asm:distance-funciton}\ref{asm:distance-funciton-a} is indeed verified (\cref{lem:fench-breg}) and \cref{asm:distance-funciton}\ref{asm:distance-funciton-b} is implied by the corresponding reciprocity condition (this can be proved by using some standard arguments of the point-set topology).

In the sequel, we will restate all our results in the case where players ``adopt an adaptive optimistic learning strategy''.
This means that the player runs an optimistic algorithm that is compatible with dynamic learning rate with a regularizer $\vw[\hreg]$ and the adaptive scheme \eqref{eq:adaptive-reg}, and plays $\vwt[\action]=\vwtinter[\state]$.
For ease of presentation, we will take $\vw[\tau]=1$ throughout, and we will assume that $\vw[\hreg]$ is $1$-strongly convex relative to $\vwp[\norm{\cdot}]$,
but the proof can be easily adapted to general $\vw[\tau]$ and $\vwp[\norm{\cdot}]$.
It will also be convenient to define the norm on the joint action space as
\begin{equation}
    \label{eq:joint-norm}
    \norm{\allplayers{\vw[\point]}} = \sqrt{\sumplayer\vwp[\norm{\vw[\point]}^2]}.
\end{equation}

% \begin{equation}
%     \label{eq:template-regret}
%     \begin{aligned}[b]
%     \sum_{\run=1}^{\nRuns}\product{\vwtinter[\gvec]}{\vwtinter[\state]-\vw[\arpoint]}
%     &\le \vwt[\regpar][\play][\nRuns+1]\vw[\armeasure](\vw[\arpoint])
%     - \vwt[\regpar][\play][\nRuns+1]\vwt[\estseq][\play][\nRuns+1](\vw[\arpoint])
%     + \sum_{\run=1}^{\nRuns}
%     \product{\vwtinter[\gvec]-\vwtpast[\gvec]}{\vwtinter[\state]-\vwtupdate[\state]}\\
%     &~-\sum_{\run=1}^{\nRuns}\vwt[\regpar] \left(
%     \vw[\breg](\vwtupdate[\state],\vwtinter[\state])
%     +\vw[\breg](\vwtinter[\state],\vwt[\state])\right),
%     \end{aligned}
% \end{equation}

\section{Proofs for regret bounds}
\label{app:regret}
%In this appendix we present the proofs for the regret bounds of the algorithms.

\subsection{Robustness to adversarial opponent}

{
\addtocounter{theorem}{-1}
\renewcommand{\thetheorem}{\ref{thm:adversarial-regret}}
\begin{theorem}
Suppose that \cref{asm:convexity+smoothness} holds, and a player $\play\in\players$ adopts an adaptive optimistic learning strategy.
If $\vw[\arpoints] \subseteq \vw[\points]$ is bounded and $\gbound = \sup_{\run} \norm{\vwt[\gvec]}$,
the regret incurred by the player is bounded as
$\vwt[\reg][\play][\nRuns](\vw[\arpoints])=\bigoh(\gbound\sqrt{\nRuns}+\gbound^2)$.
\end{theorem}
}

%\AdvReg*

\begin{proof}
By Young's inequality and the strong convexity of $\vw[\hreg]$,
\begin{equation}
    \label{eq:product-young-individual}
    \begin{multlined}[b]
    \product{\vwt[\gvec]-\vwtlast[\gvec]}{\vwtinter[\state]-\vwtupdate[\state]}
    - \vwt[\regpar]\vw[\breg](\vwtupdate[\state],\vwtinter[\state])\\
    \le \frac{\vwpdual[\norm{\vwt[\gvec]-\vwtlast[\gvec]}^2]}{2\vwt[\regpar]}
    + \frac{\vwt[\regpar]}{2}\vwp[\norm{\vwtinter[\state]-\vwtupdate[\state]}^2]
    - \frac{\vwt[\regpar]}{2}\vwp[\norm{\vwtinter[\state]-\vwtupdate[\state]}^2]
    = \frac{\vwt[\increment]}{2\vwt[\regpar]}.
    \end{multlined}
\end{equation}
From \eqref{eq:template-regret} we then obtain
\begin{equation}
    \label{eq:individual-regret-inq}
    \begin{aligned}[b]
    \sum_{\run=1}^{\nRuns}\product{\vwt[\gvec]}{\vwtinter[\state]-\vw[\arpoint]}
    &\le
    \vwt[\regpar][\play][\nRuns+1]\vw[\armeasure](\vw[\arpoint])
    + \frac{1}{2}\sum_{\run=1}^{\nRuns} \frac{\vwt[\increment]}{\vwt[\regpar]}\\
    &=
    \vwt[\regpar][\play][\nRuns+1]\vw[\armeasure](\vw[\arpoint])
    + \frac{1}{2}\sum_{\run=1}^{\nRuns} \frac{\vwt[\increment]}{\vwtupdate[\regpar]}
    + \frac{1}{2}\sum_{\run=1}^{\nRuns}
    \left(\frac{1}{\vwt[\regpar]}-\frac{1}{\vwtupdate[\regpar]}\right)\vwt[\increment]\\
    &\le
    (1+\vw[\armeasure](\vw[\arpoint]))\sqrt{1+\sum_{\run=1}^{\nRuns}\vwt[\increment]}
    + 2\sum_{\run=1}^{\nRuns}
    \left(\frac{1}{\vwt[\regpar]}-\frac{1}{\vwtupdate[\regpar]}\right)\gbound^2\\
    &\le  (1+\vw[\armeasure](\vw[\arpoint]))\sqrt{1+4\gbound^2\nRuns} + 2\gbound^2
    \end{aligned}
\end{equation}
We have used \cref{lem:adaptive} in the second to last inequality.
We conclude by maximizing the above inequality over $\vw[\arpoint]\in\vw[\arpoints]$.
\end{proof}

\subsection{Constant bound on social regret}
\label{app:social-regret-proof}

{
\addtocounter{theorem}{-1}
\renewcommand{\thetheorem}{\ref{thm:social-regret-bounded}}
\begin{theorem}
Suppose that \cref{asm:convexity+smoothness} holds and all players $\play\in\players$ adopt an adaptive optimistic learning strategy.
Then, for every bounded comparator set $\arpoints\subseteq\points$, the players' social regret is bounded as $\vt[\reg][\nRuns](\arpoints)=\bigoh(1)$.
% Let \cref{asm:convexity+smoothness} holds and that all the players adopt an adaptive optimistic learning strategy. Then for any bounded subset $\arpoints\subset\points$, we have $\vt[\reg][\nRuns](\arpoints)=\bigoh(1)$.
\end{theorem}
}

\begin{proof}
Let $\bb{\arpoint}=(\vw[\arpoint])_{\play\in\players}\in\arpoints$.
Since $\arpoints$ is bounded and $\vw[\armeasure]$ is continuous,
there exists $\vw[\inibound]>0$ such that it always holds $\vw[\armeasure](\vw[\arpoint])\le\vw[\inibound]$.
We start by rewriting the regret bound \eqref{eq:template-regret} as
\begin{equation}
\label{eq:OptDA-regret-rewrite}
\begin{aligned}[b]
    \sum_{\run=1}^{\nRuns}\product{\vwt[\gvec]}{\vwtinter[\state]-\vw[\arpoint]}
    &\le \vwt[\regpar][\play][\nRuns+1]\vw[\armeasure](\vw[\arpoint])
    - \vwt[\regpar][\play][\nRuns+1]\vwt[\estseq][\play][\nRuns+1](\vw[\arpoint])
    \\
    &~~- \vwt[\regpar][\play][\start]
    \vw[\breg](\vwt[\state][\play][\interstart],\vwt[\state][\play][\start])
    - \frac{\vwt[\regpar][\play][\nRuns]}{2}
    \vw[\breg](\vwtupdate[\state][\play][\nRuns],\vwtinter[\state][\play][\nRuns])
    \\
    &~~ - \sum_{\run=2}^{\nRuns} \left(\frac{\vwt[\regpar][\play][\run-1]}{2}\vw[\breg](\vwt[\state],\vwtpast[\state]) + \vwt[\regpar]\vw[\breg](\vwtinter[\state],\vwt[\state])\right)\\
    &~~ + \sum_{\run=1}^{\nRuns}
    \left(\product{\vwt[\gvec]-\vwtlast[\gvec]}{\vwtinter[\state]-\vwtupdate[\state]}
    - \frac{\vwt[\regpar]}{2}\vw[\breg](\vwtupdate[\state],\vwtinter[\state]) \right)
\end{aligned}
\end{equation}
%
% Moreover, we have $\vt[\hreg]=\vwt[\regpar]\vw[\hreg]$ and we choose $\vtp[\norm{\cdot}]=\sqrt{\vwt[\regpar]}\vwp[\norm{\cdot}]$.
% Then for all $\run\ge2$, it holds $\vtp[\norm{\cdot}]\ge\vtp[\norm{\cdot}][\run-1]$; using the strong convexity of $\vt[\hreg]$ we get
% %
% \begin{equation}
% \label{eq:2breg-combine}
% \begin{aligned}[b]
% \vtp[\norm{\inter-\past}^2][\run-1]
% &\le
% 2\vtp[\norm{\inter-\current}^2]+2\vtp[\norm{\current-\past}^2][\run-1]\\
% &\le4\vt[\breg](\inter,\current)+4\vt[\breg][\run-1](\current,\past).
% \end{aligned}
% \end{equation}
On one hand, the strong convexity of $\vw[\hreg]$ implies
\begin{equation}
\label{eq:2breg-combine}
\begin{aligned}[b]
\vwp[\norm{\vwtinter-\vwtpast}^2]
&\le
2\vwp[\norm{\vwtinter-\vwt}^2]+2\vwp[\norm{\vwt-\vwtpast}^2]\\
&\le
4\vw[\breg](\vwtinter,\vwt)+4\vw[\breg](\vwt,\vwtpast).
\end{aligned}
\end{equation}
On the other hand, similar to \eqref{eq:product-young-individual},
\begin{equation}
    \label{eq:product-young}
    \product{\vwt[\gvec]-\vwtlast[\gvec]}{\vwtinter[\state]-\vwtupdate[\state]}
    - \frac{\vwt[\regpar]}{2}\vw[\breg](\vwtupdate[\state],\vwtinter[\state])\\
    \le \frac{\vwpdual[\norm{\vwt[\gvec]-\vwtlast[\gvec]}^2]}{\vwt[\regpar]}.
\end{equation}
%
% \begin{equation}
%     \label{eq:product-young}
%     \begin{multlined}[b]
%     \product{\inter[\gvec]-\inter[\appr[\gvec]]}{\inter-\update}
%     - \frac{1}{2}\vt[\breg][\run](\update,\inter)\\
%     \le \vtpdual[\norm{\inter[\gvec]-\inter[\appr[\gvec]]}^2]
%     + \frac{1}{4}\vtp[\norm{\inter-\update}^2]
%     - \frac{1}{4} \vtp[\norm{\inter-\update}^2] = \vtpdual[\norm{\inter[\gvec]-\inter[\appr[\gvec]]}^2].
%     \end{multlined}
% \end{equation}
%
Combining \eqref{eq:OptDA-regret-rewrite}, \eqref{eq:2breg-combine}, \eqref{eq:product-young}, we obtain
\begin{equation}
    \label{eq:OptDA-regret-key}
    \begin{aligned}[b]
    \sum_{\run=1}^{\nRuns}\product{\vwt[\gvec]}{\vwtinter[\state]-\vw[\arpoint]}
    &\le \vwt[\regpar][\play][\nRuns+1]\vw[\armeasure](\vw[\arpoint])
    - \vwt[\regpar][\play][\nRuns+1]\vwt[\estseq][\play][\nRuns+1](\vw[\arpoint])\\
    %-\fench_{\vwt[\regpar][\play][\nRuns+1]\vw[\hreg]}(\vw[\arpoint],\vwt[\dstate][\play][\nRuns+1])
    %- \vt[\breg][\start](\vwt[\state][\play][\interstart],\vwt[\state][\play][\start])
    %- \vt[\breg][\nRuns](\vwt[\state][\play][\nRuns+1],\vwtinter[\state][\play][\nRuns])\\
    &~~ + \sum_{\run=1}^{\nRuns}\frac{\vwpdual[\norm{\vwt[\gvec]-\vwtlast[\gvec]}^2]}{\vwt[\regpar]}
    - \frac{1}{8}\sum_{\run=2}^{\nRuns} \vwtlast[\regpar]\vwp[\norm{\vwtinter[\state]-\vwtpast[\state]}^2][\play]\\
    &\le \vwt[\regpar][\play][\nRuns+1]\vw[\inibound]
    + \vwpdual[\norm{\vw[\vecfield](\vt[\jaction][\start])}^2]\\
    &~~ + \sum_{\run=2}^{\nRuns}
    \left(\frac{\vwpdual[\norm{\vw[\vecfield](\current[\jaction])-\vw[\vecfield](\last[\jaction])}^2]}{\vwt[\regpar]}
    - \frac{\vwtlast[\regpar]}{8} \vwp[\norm{\vwtinter[\state]-\vwtpast[\state]}^2]\right).
    \end{aligned}
\end{equation}
In the current setting, the realized joint action is $\current[\jaction]=\inter[\jstate]$. 
%From now on, we will denote by $\inter[\jstate]=\current[\jaction]=(\vwtinter[\state])_{\play\in\players}$ the joint leading state. %and $\inibound=\max_{\play\in\players}\vw[\inibound]$ a shared upper bound of $\vw[\hreg](\vw[\arpoint])$.
With the norm on $\points$ defined in \eqref{eq:joint-norm}, we have $\sumplayer\norm{\vwtinter[\state]-\vwtpast[\state]}^2 = \norm{\inter[\jstate]-\past[\jstate]}^2$.
Note that $\vwt[\regpar]\ge1$ for all $\run$ and $\play$ by definition. Summing \eqref{eq:OptDA-regret-key} from $\play=1$ to $\nPlayers$ and maximizing over $\arpoint\in\arpoints$ then gives
\begin{equation}
    \label{eq:OptDA-regret-refined}
    \begin{aligned}[b]
    \vt[\reg][\nRuns](\arpoints)
    &\le \sumplayer
    \left(\vwt[\regpar][\play][\nRuns+1]\vw[\inibound]
    +\vwpdual[\norm{\vw[\vecfield](\vt[\jaction][\start])}^2]\right)\\
    &~~+ \sum_{\run=2}^{\nRuns}
    \left(\sumplayer
    \frac{\vwpdual[\norm{\vw[\vecfield](\inter[\jstate])-\vw[\vecfield](\past[\jstate])}^2]}{\vwt[\regpar]}
    - \frac{1}{8} \norm{\inter[\jstate]-\past[\jstate]}^2\right).
    \end{aligned}
\end{equation}

In the remainder of the proof, we show that the \acl{RHS} of \eqref{eq:OptDA-regret-refined} is bounded from above by some constant.
Since all the norms are equivalent in a finite dimensional space, from \cref{asm:convexity+smoothness} we know that for every $\play\in\players$, there exists $\vw[\lips]>0$ such that for all $\jaction, \alt{\jaction}\in\points$,
\begin{equation}
    \label{eq:lips}
    \vwpdual[\norm{\vw[\vecfield](\jaction)-\vw[\vecfield](\alt{\jaction})}]
    \le \vw[\lips]\norm{\jaction-\alt{\jaction}}.
\end{equation}
Subsequently, %From \eqref{eq:lips} we deduce
\begin{equation}
    \label{eq:lips-sum}
    \norm{\inter[\jstate]-\past[\jstate]}^2
    \ge \sumplayer
    \frac{1}{\nPlayers\vw[\lips]^2}\vwpdual[\norm{\vw[\vecfield](\inter[\jstate])-\vw[\vecfield](\past[\jstate])}^2].
\end{equation}
It is thus sufficient to show that for each $\play\in\players$, there exists $\vw[\Cst]\in\R_+$ such that for all $\nRuns\in\N$,
\begin{gather}
    \label{eq:OptDA-regret-finite-1}
    \vwt[\regpar][\play][\nRuns+1]\vw[\inibound] - \frac{1}{16\nPlayers\vw[\lips]^2}\sum_{\run=2}^{\nRuns}
    \vwpdual[\norm{\vw[\vecfield](\inter[\jstate])-\vw[\vecfield](\past[\jstate])}^2] \le \vw[\Cst],\\
    \label{eq:OptDA-regret-finite-2}
    \sum_{\run=2}^{\nRuns}\left(
    \frac{\vwpdual[\norm{\vw[\vecfield](\inter[\jstate])-\vw[\vecfield](\past[\jstate])}^2]}{\vwt[\regpar]}
    - \frac{1}{16\nPlayers\vw[\lips]^2}
    \vwpdual[\norm{\vw[\vecfield](\inter[\jstate])-\vw[\vecfield](\past[\jstate])}^2]
    \right) \le \vw[\Cst].
\end{gather}
To simplify the notation, we will write $\vw[\deccst]=1/(16\nPlayers\vw[\lips]^2)$. 
We recall that $\vwt[\regpar]=\sqrt{1+\sum_{\runalt=1}^{\run-1}\vwt[\increment]}$ where $\vwt[\increment]=\vwpdual[\norm{\vwt[\gvec]-\vwtlast[\gvec]}^2]$.
%$\vwt[\increment]=\vwpdual[\norm{\vw[\vecfield](\inter[\jstate])-\vw[\vecfield](\past[\jstate])}^2]$ for $\run\ge2$
%and $\vwt[\increment][\play][\start] = \vwpdual[\norm{\vw[\vecfield](\vt[\jaction][\start])}^2]$.
%Under these notations, we have .
Using the inequality $\sqrt{a+b}\le\sqrt{a}+\sqrt{b}$, we can bound the \acl{LHS} of \eqref{eq:OptDA-regret-finite-1} as following
\begin{equation}
    \label{eq:OptDA-finite1-quadratic}
    \vw[\inibound]\sqrt{1+\sum_{\runalt=1}^{\nRuns}\vwt[\increment]}
    - \vw[\deccst]\sum_{\run=2}^{\nRuns}\vwt[\increment]
    \le \vw[\inibound]\sqrt{1+\vwt[\increment][\play][\start]}
    + \vw[\inibound]\sqrt{\sum_{\runalt=2}^{\nRuns}\vwt[\increment]}
    - \vw[\deccst]\sum_{\run=2}^{\nRuns}\vwt[\increment]
    = \vw[\func]\left(\sqrt{\sum_{\run=2}^{\nRuns}\vwt[\increment]}\right).
\end{equation}
where $\vw[\func]\from\scalar\in\R\mapsto -\vw[\deccst]\scalar^2 + \vw[\inibound]\scalar + \vw[\inibound]\sqrt{1+\vwt[\increment][\play][\start]}$ is a quadratic function with negative leading coefficient
and is hence bounded from above. This proves \eqref{eq:OptDA-regret-finite-1} by setting $\vw[\Cst]\ge\max_{\scalar\in\R_+}\vw[\func](\scalar)$.

Note that $(\vwt[\regpar])_{\run\in\N}$ is non-decreasing. Therefore, it either converges to some finite limit or tends to plus infinity.
We can thus write $\lim_{\run\to+\infty}\vwt[\regpar]=\vw[\regpar]\in\R_+\union\{+\infty\}$.
%while in the former case $\vw[\regpar]\in\R_+$ and in the latter situation $\vw[\regpar] = +\infty$. 
%Looking more closely at these two distinct cases helps us 
To prove \eqref{eq:OptDA-regret-finite-2}, we tackle the two cases separately:

\textbf{Case 1, $\vw[\regpar]\in\R_+$:} In other words, $\sum_{\run=2}^{+\infty}\vwt[\increment]$ is finite. Since $\vwt[\regpar]\ge1$, by taking $\vw[\Cst]\ge\sum_{\run=2}^{+\infty}\vwt[\increment]$ inequality \eqref{eq:OptDA-regret-finite-2} is verified.

\textbf{Case 2, $\vw[\regpar]=+\infty$:}
Then $\lim_{\run\to+\infty}1/\vwt[\regpar]=0$.
The quantity $\alt{\run} = \min_{\run}\setdef{\run}{1/\vwt[\regpar]\le\vw[\deccst]}$ is well-defined and the inequality \eqref{eq:OptDA-regret-finite-2} is satisfied as long as $\vw[\Cst]\ge\sum_{\run=2}^{\alt{\run}-1}(1/\vwt[\regpar]-\vw[\deccst])\vwt[\increment]$.

To summarize, we have proved that \eqref{eq:OptDA-regret-finite-1} and \eqref{eq:OptDA-regret-finite-2} must hold for some $\vw[\Cst]\in\R_+$. Therefore, invoking $\eqref{eq:OptDA-regret-refined}$ and $\eqref{eq:lips-sum}$ we have effectively proved $\vt[\reg][\nRuns](\arpoints)=\bigoh(1)$.
\end{proof}

\subsection{Individual regret bound in variationally stable games}

\begin{lemma}
\label{lem:regpar-bounded}
Let \cref{asm:convexity+smoothness} holds and that all players $\play\in\players$ adopt an adaptive optimistic learning strategy.
Assume additionally that the game is variationally stable.
Then, for every $\play\in\players$, the sequence $(\vwt[\regpar])_{\run\in\N}$ converges to a finite constant $\vw[\regpar]\in\R_+$ (equivalently, $\sum_{\run=1}^{+\infty}\vwt[\increment]<+\infty$).
\end{lemma}
\begin{proof}
In this proof we borrow the notations from the proof of \cref{thm:social-regret-bounded}.
First, summing the \acl{LHS} of \eqref{eq:OptDA-regret-key} from $\play=1$ to $\nPlayers$ leads to $\sum_{\run=1}^{\nRuns}\product{\jvecfield(\inter[\jstate])}{\inter[\jstate]-\bb{\arpoint}}$.
Since the game is variationally stable, we may take $\bb{\arpoint}\subs\bb{\sol}\in\sols$ a Nash equilibrium of the game, which gurantees that $\product{\jvecfield(\jaction)}{\jaction-\bb{\sol}}\ge0$ for all $\jaction\in\points$. 
Summing \eqref{eq:OptDA-regret-key} from $\play=1$ to $\nPlayers$ and using the Lipschitz continuity of the functions, similar to \eqref{eq:OptDA-regret-refined}, we obtain
\begin{equation}
    \label{eq:minty-bound}
    \begin{aligned}[b]
    0
    &\le \sumplayer
    \left(\vwt[\regpar][\play][\nRuns+1]\vw[\armeasure](\vw[\sol])
    +\vwpdual[\norm{\vw[\vecfield](\vt[\jaction][\start])}^2]\right)
    %-\fench_{\vwt[\regpar][\play][\nRuns+1]\vw[\hreg]}(\bb{\sol},\vwt[\dstate][\play][\nRuns+1])\right)
    \\
    &~~+ \sum_{\run=2}^{\nRuns}
    \left(\sumplayer
    \frac{\vwpdual[\norm{\vw[\vecfield](\inter[\jstate])-\vw[\vecfield](\past[\jstate])}^2]}{\vwt[\regpar]}
    - \frac{1}{8} \norm{\inter[\jstate]-\past[\jstate]}^2\right).
    \end{aligned}
\end{equation}
Combining \eqref{eq:OptDA-regret-finite-1} and \eqref{eq:OptDA-regret-finite-2} with the above inequality, we deduce that for any $\play$, there exists $\vw[\wilde{\Cst}]\in\R$ such that for all $\nRuns\in\N$, %it holds (here $\vw[\inibound]=\vw[\armeasure](\vw[\sol])$)
\begin{equation}
\notag
    \vw[\armeasure](\vw[\sol])\sqrt{1+\sum_{\runalt=1}^{\nRuns}\vwt[\increment]}
    - \vw[\deccst]\sum_{\run=2}^{\nRuns}\vwt[\increment]
    \ge \vw[\wilde{\Cst}].
\end{equation}
Invoking \eqref{eq:OptDA-finite1-quadratic} then gives $\vw[\func]\left(\sqrt{\sum_{\run=2}^{\nRuns}\vwt[\increment]}\right)\ge \vw[\wilde{\Cst}]$.
%, and this holds for all $\nRuns$.
Since $\vw[\func]$ is a quadratic function with negative leading coefficient, $\lim_{\scalar\to+\infty}\vw[\func](\scalar)=-\infty$. Accordingly, $\sum_{\run=2}^{+\infty}\vwt[\increment]$ is finite, which in turn implies $\vw[\regpar]=\lim_{\run\to+\infty}\vwt[\regpar]<+\infty$.
\end{proof}

{
\addtocounter{theorem}{-1}
\renewcommand{\thetheorem}{\ref{thm:inidividual-regret-bound}}
\begin{theorem}
Suppose that \cref{asm:convexity+smoothness} holds and all players $\play\in\players$ adopt an adaptive optimistic learning strategy.
If the game is variationally stable, then, for every bounded comparator set $\vw[\arpoints]\subseteq\vw[\points]$, the individual regret of player $\play\in\players$ is bounded as $\vwt[\reg][\play][\nRuns](\vw[\arpoints])=\bigoh(1)$.
%Let \cref{asm:convexity+smoothness} holds and that all the players adopt an adaptive optimistic learning strategy.
%Assume additionally that the game is variationally stable.
%Then, for every $\play\in\players$ and every bounded subset $\vw[\arpoints]\subset\vw[\points]$, $\vwt[\reg][\play][\nRuns](\vw[\arpoints])=\bigoh(1)$.
\end{theorem}
}

\begin{proof}
From the first line of \eqref{eq:individual-regret-inq} we have
\begin{equation}
    \label{eq:individual-regret-inq'}
    \sum_{\run=1}^{\nRuns}\product{\vwt[\gvec]}{\vwtinter[\state]-\vw[\arpoint]}
    \le \vwt[\regpar][\play][\nRuns+1]\vw[\armeasure](\vw[\arpoint])
    + \frac{1}{2}\sum_{\run=1}^{\nRuns} \frac{\vwt[\increment]}{\vwt[\regpar]}.
\end{equation}
As $\vw[\armeasure]$ is continuous and $\vw[\arpoints]$ is bounded, $\vw[\inibound]=\max_{\vw[\arpoint]\in\vw[\arpoints]}\vw[\armeasure](\vw[\arpoint])$ is well-defined. 
Moreover, $1/\vwt[\regpar]\le1$ for all $\run$. Maximizing \eqref{eq:individual-regret-inq'} over $\vw[\arpoint]\in\vw[\arpoints]$ then gives
\begin{equation}
    \notag
    \vwt[\reg][\play][\nRuns](\vw[\arpoints])
    \le \vwt[\regpar][\play][\nRuns+1]\vw[\inibound]
    + \frac{1}{2}\sum_{\run=1}^{\nRuns} \vwt[\increment]
    \le \vw[\regpar]\vw[\inibound] + \frac{1}{2}\sum_{\run=1}^{+\infty} \vwt[\increment],
\end{equation}
where $\vw[\regpar]=\lim_{\run\to+\infty}\vwt[\regpar]$ and $\sum_{\run=1}^{+\infty} \vwt[\increment]$ are finite according to \cref{lem:regpar-bounded}. We have thus proved $\vwt[\reg][\play][\nRuns](\vw[\arpoints])=\bigoh(1)$.
\end{proof}

\section{Proofs for last-iterate convergence}
\label{app:convergence}
%In this appendix, the norm on the joint space $\norm{\cdot}$ will be the one defined by \eqref{eq:joint-norm} and we will write $\vw[\regpar]=\lim_{\toinf}\vwt[\regpar]\in\R_+\union\{+\infty\}$.

\subsection{Convergence to best response}

In this part, we focus on the learning of a single player when the realized actions of the other players converge asymptotically.
For ease of notation, the player index $\play$ will be dropped when there is no confusion.

\begin{lemma}
\label{lem:to-zero-BR}
Let player $\play$ adopt an adaptive (optimistic) learning strategy.
Then, if the sequence of received feedback is bounded, both the sequences
\begin{enumerate*}[\itshape a\upshape)]
\item $\seqinf[\vt[\regpar][\run+1]-\vt[\regpar]]$ and
\item $\seqinf[\vt[\increment]/\vt[\regpar]]$
\end{enumerate*}
tend to zero. 
\end{lemma}
\begin{proof}
This trivially holds if $\lim_{\run\to+\infty}\vt[\regpar]<+\infty$ (which is equivalent to $\sum_{\run=1}^{+\infty}\vt[\increment]<+\infty$).
Otherwise, we have $\vt[\regpar]\to+\infty$.
Let $\gbound$ be an upper bound on the received feedback.
Since $\vt[\increment]\le4\gbound^2$, we deduce the sequence \emph{b}) converges to $0$. For the sequence \emph{a}), we simply note that
\begin{equation}
    \notag
    \vt[\regpar][\run+1]-\vt[\regpar]
    = \frac{\vt[\regpar^2][\run+1]-\vt[\regpar^2]}{\vt[\regpar][\run+1]+\vt[\regpar]}
    = \frac{\vt[\increment]}{\vt[\regpar][\run+1]+\vt[\regpar]}\le\frac{2\gbound^2}{\vt[\regpar]}\xrightarrow{\vt[\regpar]\to+\infty}0.
\end{equation}
\end{proof}

{
\addtocounter{theorem}{-1}
\renewcommand{\thetheorem}{\ref{thm:cvg-best-response}}
\begin{theorem}
Suppose that \cref{asm:convexity+smoothness} holds, and a player $\play\in\players$ adopts an adaptive optimistic learning strategy that verifies \cref{asm:distance-funciton}.
Assume additionally that $\vw[\points]$ is compact.
Then, if all other players' actions converge to a point $\vw[\limp{\jaction}][\playexcept]\in\prod_{\playalt\neq\play}\vw[\points][\playalt]$, player $\play$'s realized actions converge to the best response to $\vw[\limp{\jaction}][\playexcept]$.
% player $\play$'s realized actions converge to the best response to $\vw[\limp{\jaction}][\playexcept]$.
%the trajectory of chosen actions of the player in question converges to a best response.
% Let \cref{asm:convexity+smoothness} hold and player $\play$ adopts an adaptive optimistic learning strategy that verifies \cref{asm:distance-funciton}.
% Suppose additionally that $\vw[\points]$ is compact.
% Then, if all the other players' actions converge to a point $\vw[\limp{\jaction}][\playexcept]\in\prod_{\playalt\neq\play}\vw[\points][\playalt]$,
% player $\play$'s realized actions converge to the best response to $\vw[\limp{\jaction}][\playexcept]$.
\end{theorem}
}

\begin{proof}
Let $\vw[\sol]\in\vw[\sols]\defeq\BR(\vw[\limp{\jaction}][\playexcept])$. From \eqref{eq:template-descent} we derive immediately that
\begin{equation}
    \label{eq:PMPDS-key-descent}
    \begin{aligned}[b]
    \update[\regpar]\update[\estseq](\vw[\sol])
    &\le \vt[\regpar]\vt[\estseq](\vw[\sol])
    + (\vt[\regpar][\run+1]-\vt[\regpar][\run])\inibound
    +\frac{\vt[\increment]}{\vt[\regpar]}\\
    &~~-\product{\vw[\vecfield](\inter[\jstate])}{\vwtinter[\state]-\vw[\sol]}
    -\frac{\vt[\regpar]}{4}\vwp[\norm{\vwtupdate[\state]-\vwtinter[\state]}^2],
    \end{aligned}
\end{equation}
where $\inibound=\max_{\vw[\sol]\in\vw[\sols]}\armeasure(\vw[\sol])$.
The scalar product term is not necessarily non-negative, but with $\inter[\jstatealt] = (\vwtinter[\state],\vw[\limp{\jaction}][\playexcept])$, $\bb{\sol}=(\oneandother[\sol][\limp{\jaction}])$, and $\radius$ the diameter of $\vw[\points]$, we can decompose
\begin{equation}
    \label{eq:product-lower-bound-decompose-br}
    \begin{aligned}
    \product{\vw[\vecfield](\inter[\jstate])}{\vwtinter[\state]-\vw[\sol]}
    &=
    \product{\vw[\vecfield](\inter[\jstate])-\vw[\vecfield](\inter[\jstatealt])}{\vwtinter[\state]-\vw[\sol]}
    + \product{\vw[\vecfield](\inter[\jstatealt])}{\vwtinter[\state]-\vw[\sol]}\\
    &\ge
    -\radius\vwpdual[\norm{\vw[\vecfield](\inter[\jstate])-\vw[\vecfield](\inter[\jstatealt])}]
    + \vw[\loss](\inter[\jstatealt]) - \vw[\loss](\bb{\sol}).
    \end{aligned}
\end{equation}
In the inequality we have used the convexity of $\vw[\loss](\cdot,\vw[\limp{\jaction}][\playexcept])$.
Since $\vw[\points]$ is compact and $\vw[\vecfield]$ is continuous, the function
\begin{equation}
    \notag
    \func:\vw[\jaction][\playexcept]\mapsto\max_{\vw[\arpoint]\in\vw[\points]}\vwpdual[\norm{\vw[\vecfield](\oneandother[\arpoint][\jaction])-\vw[\vecfield](\oneandother[\arpoint][\limp{\jaction}])}]
\end{equation}
is continuous by Berge's maximum theorem. Therefore $\func(\vwtinter[\jstate][\playexcept])$ converges to $0$ when $\run$ goes to infinity. Moreover, from \eqref{eq:product-lower-bound-decompose-br} we have
\begin{equation}
    \label{eq:product-lower-bound-decompose-br-2}
    \product{\vw[\vecfield](\inter[\jstate])}{\vwtinter[\state]-\vw[\sol]}
    \ge
    -\radius\func(\vwtinter[\jstate][\playexcept])
    + \vw[\loss](\inter[\jstatealt]) - \vw[\loss](\bb{\sol}).
\end{equation}
Let us write $\vw[\sol[\loss]] = \min_{\vw[\action]\in\vw[\points]}\vw[\loss](\vw[\action],\vw[\limp{\jaction}][\playexcept])$.
Combining \eqref{eq:PMPDS-key-descent}, \eqref{eq:product-lower-bound-decompose-br-2} and minimizing with respect to $\vw[\sol]\in\vw[\sols]$ leads to 
\begin{equation}
    \label{eq:PMPDS-key-descent-set}
    \begin{aligned}[b]
    \vt[\regpar][\run+1]\update[\estseq](\vw[\sols])
    &\le \vt[\regpar]\vt[\estseq](\vw[\sols])
    + (\vt[\regpar][\run+1]-\vt[\regpar][\run])\inibound
    +\frac{\vt[\increment]}{\vt[\regpar]}
    +\radius\func(\vwtinter[\jstate][\playexcept])\\
    &~~ - (\vw[\loss](\inter[\jstatealt]) - \vw[\sol[\loss]])
    -\frac{\vt[\regpar]}{4}\vwp[\norm{\vwtupdate[\state]-\vwtinter[\state]}^2].
    \end{aligned}
\end{equation}
We define 
$\vt[\tozero]=(\vt[\regpar][\run+1]-\vt[\regpar][\run])\inibound
+\vt[\increment]/\vt[\regpar]
+\radius\func(\vwtinter[\jstate][\playexcept])$.
As $\vw[\vecfield]$ is continuous, $\vw[\points]$ is compact, and the iterates $\seqinf[\vwt[\jaction][\playexcept]]$ converges and is hence bounded, the sequence of feedback received by player $\play$ is also bounded.
Applying \cref{lem:to-zero-BR} then gives $\lim_{\run\to+\infty}\vt[\tozero]=0$.

Let us next prove that for any $\sbradius>0$, we have $\vt[\estseq](\vw[\sols])\le\sbradius$ for all $\run$ large enough.
% By Bregman reciprocity, if a sequence $\seqinf[\vt[\action]]$ converges to a point $\vw[\sol]$ in the solution set $\vw[\sols]$, then $\breg(\vw[\sols],\vt[\action])\le\breg(\vw[\sol],\vt[\action])$ converges to $0$.
% Therefore, we can apply \cref{lem:compact-neighborhood}
Since $\vw[\sols]\subset\vw[\points]$ is a compact set, \cref{asm:distance-funciton}\ref{asm:distance-funciton-b} ensures the existence of $\sradius>0$ such that if $\dist(\vwt[\state],\vw[\sols])\le\sradius$ then $\vt[\estseq](\vw[\sols])\le\sbradius$.
%Let $\run_1\in\N$ satisfies that if $\run\ge\run_1$ then $\vt[\tozero]\le\min(\sbradius/2,\sradius/4)$ (the existence of $\run_1$ is guaranteed by $\lim_{\run\to+\infty}\vt[\tozero]=0$).
We distinguish between three different situations:

\smallskip
\textbf{Case 1, $\dist(\vwtinter[\state],\vw[\sols])\ge\sradius/2$:}
By convexity of $\vw[\loss](\cdot,\vw[\limp{\jaction}][\playexcept])$ this clearly implies the existence $\cst>0$ such that $\vw[\loss](\inter[\jstatealt]) - \vw[\sol[\loss]]\ge\cst$ whenever we are in this situation.
As $\lim_{\run\to+\infty}\vt[\tozero]=0$, there exists $\run_1\in\N$ such that for all $\run\ge\run_1$, 
$\vt[\tozero]\le\cst/2$.
For any $\run\ge\run_1$, the inequality \eqref{eq:PMPDS-key-descent-set} then gives
\begin{equation}
    \notag
    \vt[\regpar][\run+1]\update[\estseq](\vw[\sols])
    \le \vt[\regpar]\vt[\estseq](\vw[\sols])
    +\vt[\tozero]
    -\cst
    -\frac{\vt[\regpar]}{4}\vwp[\norm{\vwtupdate[\state]-\vwtinter[\state]}^2]
    \le\vt[\regpar]\vt[\estseq](\vw[\sols]) - \frac{\cst}{2}. 
\end{equation}

\textbf{Case 2, $\dist(\vwtinter[\state],\vw[\sols])\le\sradius/2$ and $\vwp[\norm{\vwtupdate[\state]-\vwtinter[\state]}]\ge\sradius/2$:}
We define $\run_2\in\N$ such that for all $\run\ge\run_2$, 
$\vt[\tozero]\le\sradius^2/32$. Then for $\run\ge\run_2$,
\begin{equation}
    \notag
    \vt[\regpar][\run+1]\update[\estseq](\vw[\sols])
    \le \vt[\regpar]\vt[\estseq](\vw[\sols])
    +\vt[\tozero]
    - (\vw[\loss](\inter[\jstatealt]) - \vw[\sol[\loss]])
    -\frac{\sradius^2}{16}
    \le\vt[\regpar]\vt[\estseq](\vw[\sols]) - \frac{\sradius^2}{32}. 
\end{equation}

\textbf{Case 3, $\dist(\vwtinter[\state],\vw[\sols])\le\sradius/2$ and $\vwp[\norm{\vwtupdate[\state]-\vwtinter[\state]}]\le\sradius/2$:}
By the triangular inequality this implies $\dist(\vwtupdate[\state],\vw[\sols])\le\sradius$ and thus $\update[\estseq](\vw[\sols])\le\sbradius$ by the choice of $\sradius$.

\smallskip
\textbf{Conclude.}
Let us consider the sequence $\seqinf[\vt[\vdist]]\in(\R_+)^{\N}$ defined by $\vt[\vdist]=\vt[\regpar]\vt[\estseq](\vw[\sols])$.
For $\run\ge\max(\run_1,\run_2)$, whenever we are in Case 1 or 2, we have $\update[\vdist]\le\vt[\vdist]-\min(\cst/2,\sradius^2/32)$.
Since $\seqinf[\vt[\vdist]]$ is non-negative, this can not happen for all $\run\ge\max(\run_1,\run_2)$; this means Case 3 must happen for some $\alt{\run}\ge\max(\run_1,\run_2)$.
Note that for both Case 1 and 2 we get $\update[\estseq](\vw[\sols])\le\vt[\estseq](\vw[\sols])$.
Therefore, with the three cases presented above we see that for all $\run\ge\alt{\run}+1$ we have $\vt[\estseq](\vw[\sols])\le\sbradius$.
We have proved that for any $\sbradius>0$, the distance measure $\vt[\estseq](\vw[\sols])$ becomes eventually smaller than $\sbradius$.
This means $\lim_{\run\to+\infty}\vt[\estseq](\vw[\sols])=0$ and accordingly $\lim_{\run\to+\infty}\dist(\vwt[\state],\vw[\sols])=0$ thanks to \cref{asm:distance-funciton}\ref{asm:distance-funciton-a}.

We next prove $\vwp[\norm{\vwtinter[\state]-\vwt[\state]}]\to0$.
%If $\vt[\regpar]\to+\infty$, by the non-expansiveness of the mirror map and the boundedness of the feedback we have $\vwp[\norm{\vwt[\state]-\vwtinter[\state]}]\le\gbound/\vt[\regpar]$ and the result follows immediately.
%Otherwise $\seqinf[\vt[\regpar]]$ converges to some finite limit $\regpar\in\R_+$.
In \eqref{eq:template-descent} we may keep the $\breg(\vwtinter[\state],\vwt[\state])$ term, then similar to how \eqref{eq:PMPDS-key-descent-set} is derived, we get
\begin{equation}
    \notag
    \vt[\regpar]\breg(\vwtinter[\state],\vwt[\state])
    \le \vt[\regpar]\vt[\estseq](\vw[\sols])
    - \vt[\regpar][\run+1]\update[\estseq](\vw[\sols])
    + \vt[\tozero].
\end{equation}
This implies 
\begin{equation}
    \notag
    \vwp[\norm{\vwtinter[\state]-\vwt[\state]}^2]
    \le
    2\left(
    \vt[\estseq](\vw[\sols])-\update[\estseq](\vw[\sols])
    +\frac{\vt[\tozero]}{\vt[\regpar]}\right).
\end{equation}
As the \acl{RHS} of the above inequality tends to zero when $\run$ goes to infinity, we conclude that  $\vwp[\norm{\vwtinter[\state]-\vwt[\state]}]\to0$.
As a consequence, $\lim_{\run\to+\infty}\dist(\vwtinter[\state],\vw[\sols])=0$.
%and $\lim_{\run\to+\infty}\vw[\gap_{\vw[\points]}](\inter[\jstate])=0$ by continuity of $\vw[\loss]$ and compactness of $\vw[\points]$.
\end{proof}

\subsection{Convergence to Nash equilibrium}

In this part, we show the convergence of the realized actions to a Nash equilibrium when all the players adopt an adaptive optimistic learning strategy in a variationally stable game.
According to \cref{lem:regpar-bounded}, the limit $\vw[\regpar]=\lim_{\toinf}\vwt[\regpar]$ is finite in this case.

\begin{lemma}
\label{lem:diff-to-zero}
Let \cref{asm:convexity+smoothness} holds and that all players $\play\in\players$ adopt an adaptive optimistic learning strategy in a variationally stable game. Then, $\norm{\inter[\jstate]-\current[\jstate]}\to0$ and $\norm{\current[\jstate]-\past[\jstate]}\to0$ as $\toinf$.
\end{lemma}
\begin{proof}
Let $\bb{\sol}$ be a Nash equilibrium.
We apply the regret bound \eqref{eq:template-regret} to $\vw[\arpoint]\subs\vw[\sol]$, and sum these bounds for $\play=1$ to $\nPlayers$, with Young's inequality \eqref{eq:product-young}, we get
\begin{equation}
    \label{eq:breg-diff-bounded}
    \frac{1}{2}\sum_{\run=1}^{\nRuns}\sumplayer
    \vwt[\regpar]\left(\vw[\breg]
    (\vwt[\state][\play][\run+1],\vwtinter[\state])
    + \vw[\breg]
    (\vwtinter[\state],\vwt[\state])\right)
    \le \sumplayer
    \left(\vw[\regpar]\vw[\hreg](\vw[\sol])+\sum_{\run=1}^{+\infty}\frac{\vwt[\increment]}{\vwt[\regpar]}\right).
\end{equation}
The \acl{RHS} of \eqref{eq:breg-diff-bounded} is finite by \cref{lem:regpar-bounded}. 
With strong convexity of $\vw[\hreg]$, this implies
\begin{equation}
    \notag
    \sum_{\run=1}^{+\infty}\left(\norm{\update[\jstate]-\inter[\jstate]}^2
    +\norm{\inter[\jstate]-\current[\jstate]}^2\right) < +\infty.
\end{equation}
As a consequence, both $\norm{\inter[\jstate]-\current[\jstate]}$ and $\norm{\current[\jstate]-\past[\jstate]}$ converge to zero when $\toinf$.
\end{proof}

\begin{lemma}
\label{lem:bregman-cvg}
Let \cref{asm:convexity+smoothness} holds and that all players $\play\in\players$ adopt an adaptive optimistic learning strategy in a variationally stable game. Then, $\sumplayer\vw[\regpar]\vwt[\estseq](\vw[\sol])$ converges for all Nash equilibrium $\bb{\sol}\in\sols$.
\end{lemma}
\begin{proof}
Let $\bb{\sol}$ be a Nash equilibrium.
From the descent inequality \eqref{eq:template-descent}, it is straightforward to show that
\begin{equation}
    \notag
    \begin{aligned}[b]
    \sumplayer\vwtupdate[\regpar]\vwtupdate[\estseq](\vw[\sol])
    &\le \sumplayer\vwt[\regpar]\vwt[\estseq](\vw[\sol])
    -\product{\jvecfield(\inter[\jstate])}{\inter[\jstate]-\bb{\sol}}\\
    &~~+\sumplayer
    \left((\vwtupdate[\regpar]-\vwt[\regpar][\run])\vw[\armeasure](\vw[\sol])
    +\frac{\vwt[\increment]}{2\vwt[\regpar]}\right).
    \end{aligned}
\end{equation}
By the choice of $\bb{\sol}$, $\product{\jvecfield(\inter[\jstate])}{\inter[\jstate]-\bb{\sol}}\ge0$.
On the other hand, thanks to \cref{lem:regpar-bounded} we know that the term on the second line is summable.
Therefore, by applying \cref{lem:summable-cvg}, we deduce the convergence of $\sumplayer\vwt[\regpar]\vwt[\estseq](\vw[\sol])$.
This in particular implies that $\vwt[\estseq](\vw[\sol])$ is bounded above for all $\play$ and $\run$; hence
$\sumplayer(\vw[\regpar]-\vwt[\regpar])\vwt[\estseq](\vw[\sol])$ converges to zero, and the convergence of $\sumplayer\vw[\regpar]\vwt[\estseq](\vw[\sol])$ follows immediately.
\end{proof}

{
\addtocounter{theorem}{-1}
\renewcommand{\thetheorem}{\ref{thm:converge-to-Nash}}
\begin{theorem}
Suppose that \cref{asm:convexity+smoothness} holds and all players $\play\in\players$ adopt an adaptive optimistic learning strategy which verifies \cref{asm:distance-funciton}.
Then the induced trajectory of play converges to a \acl{NE} provided that either of the following conditions is satisfied
\begin{enumerate}[\itshape a\upshape),itemsep=0pt]
\item
The game is strictly variationally stable.
\item
The game is variationally stable and $\vw[\hreg]$ is subdifferentiable on all of $\vw[\points]$.
\end{enumerate}
% Let \cref{asm:convexity+smoothness} holds and that all the players adopt an adaptive optimistic learning strategy which verifies \cref{asm:distance-funciton}. Then, the sequence of realized actions converges to a Nash equilibrium if one of the following is satisfied
% \begin{enumerate}[\itshape a\upshape)]
% \item the game is variationally stable and $\forall \play\in\players, \vw[\points]\subset\dom\subd\vw[\hreg]$;
% \item the game is strictly variationally stable.
% \end{enumerate}
\end{theorem}
}

\begin{proof}
We first show that in both cases, a cluster point of $\seqinf[\vt[\jstate]]$ is necessarily a Nash equilibrium.

% We first prove that the iterates are bounded. When deriving \eqref{eq:minty-bound} from \eqref{eq:OptDA-regret-key}, we may keep the $- \vwt[\regpar][\play][\nRuns+1]\vwt[\estseq][\play][\nRuns+1](\vw[\arpoint])$ term.
% Therefore, for any Nash equilibrium $\bb{\sol}$, it holds
% %
% \begin{equation}
%     \label{eq:minty-bound-on-distance}
%     \begin{aligned}[b]
%     \sumplayer\vwt[\regpar][\play][\nRuns+1]\vwt[\estseq][\play][\nRuns+1](\vw[\sol])
%     &\le \sumplayer
%     \left(\vwt[\regpar][\play][\nRuns+1]\vw[\armeasure](\vw[\sol])
%     +\vwpdual[\norm{\vw[\vecfield](\vt[\jaction][\start])}^2]\right)
%     \\
%     &~+ \sum_{\run=2}^{\nRuns}
%     \left(\sumplayer
%     \frac{\vwpdual[\norm{\vw[\vecfield](\inter[\jstate])-\vw[\vecfield](\past[\jstate])}^2]}{\vwt[\regpar]}
%     - \frac{1}{8} \norm{\inter[\jstate]-\past[\jstate]}^2\right).
%     \end{aligned}
% \end{equation}
% %
% We have shown that the \acl{RHS} of \eqref{eq:minty-bound-on-distance} is bounded above in the proof of \cref{thm:social-regret-bounded}. Invoking \cref{asm:distance-funciton}\ref{asm:distance-funciton-a} proves the boundedness of $\seqinf[\vt[\jstate]]$.

\smallskip
\emph{a})
 Let $\limp{\jaction}$ be a cluster point of $\seqinf[\vt[\jstate]]$ and $\bb{\sol}$ be a Nash equilibrium. The point $\limp{\jaction}$ is also a cluster point of $\seqinf[\inter[\jstate]]$ since $\lim_{\toinf}\norm{\inter[\jstate]-\current[\jstate]}=0$.
From the proof of \cref{thm:social-regret-bounded}, we have $\sum_{\run=1}^{\nRuns}\product{\jvecfield(\inter[\jstate])}{\inter[\jstate]-\bb{\sol}}=\bigoh(1)$.
As $\product{\jvecfield(\inter[\jstate])}{\inter[\jstate]-\bb{\sol}}\ge0$ for all $\run$, this implies $\lim_{\toinf}\product{\jvecfield(\inter[\jstate])}{\inter[\jstate]-\bb{\sol}}=0$.
Subsequently, $\product{\jvecfield(\limp{\jaction})}{\limp{\jaction}-\bb{\sol}}=0$ by the continuity of $\jvecfield$, which shows that $\limp{\jaction}$ must be a Nash equilibrium by the strict variational stability of the game.

\smallskip
\emph{b})
Let $\limp{\jaction}\in\points$ be a cluster point of $\seqinf[\vt[\jstate]]$.
We recall that $\vwtinter[\state]$ is obtained by
\begin{equation}
    \notag
    \vwtinter[\state] = \argmin_{\point\in\vw[\points]}
    \left\{
    \product{\vw[\vecfield](\past[\jstate])}{\point} + \vwt[\regpar]\vw[\breg](\point, \vwt[\state])
    \right\}.
    %{\vwt[\regpar]\vw[\hreg])}(\grad{\vt[\hreg]}(\current[\jstate])-\jvecfield(\past[\jstate])),
\end{equation}
For any $\vw[\arpoint]\in\vw[\points]$, the optimality condition \cref{lem:mirror-optimality} then gives
\begin{equation}
    \label{eq:optimality-before-limit}
    \product{\vw[\vecfield](\past[\jstate])
    + \vwt[\regpar]\grad\vw[\hreg](\vwtinter[\state])
    - \vwt[\regpar]\grad\vw[\hreg](\vwt[\state])}{\vw[\arpoint]-\vwtinter[\state]}\ge0.
\end{equation}
Let $(\vt[\jstate][\extr])_{\run\in\N}$ be a subsequence that converges to $\limp{\jaction}$.
With $\lim_{\toinf}\norm{\inter[\jstate]-\current[\jstate]}=0$ and $\lim_{\toinf}\norm{\current[\jstate]-\past[\jstate]}=0$ (\cref{lem:diff-to-zero}), we deduce
$\inter[\jstate][\extr]\to\limp{\jaction}$ and $\past[\jstate][\extr]\to\limp{\jaction}$.
Since both $\grad\vw[\hreg]$ and $\vw[\vecfield]$ are continuous ($\grad\vw[\hreg]$ is a \emph{continuous} selection of the subgradients of $\vw[\hreg]$)
%and $\vw[\vecfield]$ is continuous by \cref{asm:smooth})
and $\vw[\points]\subset\dom\subd\vw[\hreg]$, by substituting $\run\subs\extr$ in \eqref{eq:optimality-before-limit} and letting $\run$ go to infinity, we get
\begin{equation}
    \notag
    \product{\vw[\vecfield](\limp{\jaction})
    +\vw[\regpar]\grad\vw[\hreg](\vw[\limp{\action}])
    -\vw[\regpar]\grad\vw[\hreg](\vw[\limp{\action}])}
    {\vw[\arpoint]-\vw[\limp{\action}]}\ge0.
\end{equation}
In other words, for all $\vw[\arpoint]\in\vw[\points]$, it holds that
\begin{equation}
    %\label{eq:VI-obtain}
    \notag
    \product{\grad_{\vw[\action]}\vw[\loss](\limp{\jaction})}{\vw[\arpoint]-\vw[\limp{\action}]}\ge0.
\end{equation}
Since $\vw[\loss]$ is convex in $\vw[\action]$ by \cref{asm:convexity+smoothness}, the above implies
\begin{equation}
    \notag
    \vw[\loss](\vw[\arpoint],\vw[\limp{\jaction}][\playexcept]) \ge \vw[\loss](\limp{\jaction}).
\end{equation}
This is true for all $\play\in\players$ and all $\vw[\arpoint]\in\vw[\points]$, which shows that $\limp{\jaction}$ is indeed a Nash equilibrium.

\smallskip
\textbf{Conclude.}
\cref{lem:bregman-cvg} along with \cref{asm:distance-funciton}\ref{asm:distance-funciton-a} implies the boundedness of $\seqinf[\vt[\jstate]]$.
With the above we can readily show that $\dist(\vt[\jaction],\sols)\to0$ and $\limsup_{\run\to+\infty}\vwsp[\gap_{\vw[\arpoints]}](\vt[\jaction])\le0$ for all $\play$ and every compact set $\vw[\arpoints]\subset\vw[\points]$ ($\vt[\jaction]=\inter[\jstate]$ is the realized action at time $\run$).

Below, we further prove the convergence of the iterates to a point using \cref{asm:distance-funciton}\ref{asm:distance-funciton-b} and \cref{lem:bregman-cvg}.
The sequence $\seqinf[\vt[\jstate]]$, being bounded, necessarily possesses a cluster point which we denote by $\limp{\jaction}$. We have proved that $\limp{\jaction}$ must be a Nash equilibrium. Therefore, by \cref{lem:bregman-cvg} the sequence $\sumplayer\vw[\regpar]\vwt[\estseq](\vw[\limp{\action}])$ converges.
In \cref{asm:distance-funciton}\ref{asm:distance-funciton-b}, we take $\cpt\subs\{\vw[\limp{\action}]\}$ and this means that when $\vwt[\state]$ is close enough to $\vw[\limp{\action}]$, $\vwt[\estseq](\vw[\limp{\action}])$ becomes arbitrarily small.
Consequently, $\sumplayer\vw[\regpar]\vwt[\estseq](\vw[\limp{\action}])$ can only converge to zero.
By invoking \cref{asm:distance-funciton}\ref{asm:distance-funciton-a}, we then get $\lim_{\toinf}\vt[\jstate]=\limp{\jaction}$, or equivalently $\lim_{\toinf}\inter[\jstate]=\limp{\jaction}$. %(\cref{lem:diff-to-zero}).
\end{proof}

% \begin{corollary}
% Let \cref{asm:convexity,asm:minty,asm:smooth} hold and that all the players adopt an adaptive optimistic learning strategy which verifies \cref{asm:distance-funciton}. If moreover
% \begin{enumerate*}[\itshape a\upshape)]
% \item $\forall \play\in\players, \vw[\points]\subset\dom\grad\vw[\hreg]$; or
% \item $\jvecfield$ is strictly variationally stable,
% \end{enumerate*}
% then, for every $\play\in\players$ and every compact set $\vw[\arpoints]\in\vw[\points]$, we have $\limsup_{\run\to+\infty}\vwsp[\gap_{\vw[\arpoints]}](\vt[\jaction])\le0$ 
% \end{corollary}

\subsubsection{Finite two-player zero-sum games with adaptive OMWU}
We now investigate the specific case of learning in a finite two-player zero-sum game with adaptive \eqref{eq:OMWU-dynamic}.
We consider the saddle-point formulation of the problem.
Let us denote respectively by $\minvar\in\simplex_{\nPuresA}$ and $\maxvar\in\simplex_{\nPuresB}$ the mixed strategy of the first and the second player. A point $(\solA,\solB)$ is a Nash equilibrium if for all $\minvar\in\simplex_{\nPuresA}$ and $\maxvar\in\simplex_{\nPuresB}$,
\begin{equation}
    \label{eq:saddle}
    \solA^{\top}\gamemat\solB\le\minvar^{\top}\gamemat\solB,
    ~~~~
    \solA^{\top}\gamemat\maxvar\le\solA^{\top}\gamemat\solB.
\end{equation}
where $\gamemat$ is the payoff matrix and without loss of generality we assume $\norm{\gamemat}_{\infty}\le1$.
We define $\gameval=\min_{\minvar\in\simplex_{\nPuresA}}\max_{\maxvar\in\simplex_{\nPuresB}}\minvar^{\top}\gamemat\maxvar$ as the value of the game and we will write $\vc[\point]$ for the $\indg$-th coordinate of $\point$.
A pure strategy $\vw[\alpha]$ of player $\play$ is called \emph{essential} if there exists a Nash equilibrium in which player $\play$ plays $\vw[\alpha]$ with positive probability.
We have the following lemma from \cite{MPP18}.

\begin{lemma}
\label{lem:nash-all-essential}
Let $\gamemat\in\R^{\nPuresA\times\nPuresB}$ be the game matrix for a finite two-player zero-sum game with value $\gameval$. There is a Nash equilibrium $(\solA,\solB)$ such that each player plays each of their essential strategies with positive probability, and
\begin{equation}
    \notag
    \forall \pure \notin \supp(\solA),~ \vc[(\gamemat\solB)]>\gameval,
    ~~~~
    \forall \purealt \notin \supp(\solB),~ \vc[(\gamemat^{\top}\solA)][\purealt]<\gameval.
\end{equation}
\end{lemma}

In the following, we will denote by $\sol=(\solA,\solB)$ such an equilibrium.
As an immediate consequence, for all $\pure\in\supp(\solA)$, $\vc[(\gamemat\solB)]=\gameval$ and for all $\purealt\in\supp(\solB)$, $\vc[(\gamemat^{\top}\solA)][\purealt]=\gameval$. We also define
\begin{equation}
    \notag
    \gamediff
    =\min\left\{
    \min_{\pure\notin\supp(\solA)}\vc[(\gamemat\solB)]-\gameval,
    \gameval-\max_{\purealt\notin\supp(\solB)}\vc[(\gamemat^{\top}\solA)][\purealt]
    \right\}>0.
\end{equation}
Moreover, %$\gamediff\le1$
\begin{equation}
    \notag
    \gamediff
    \le\frac{\min_{\pure\notin\supp(\solA)}\vc[(\gamemat\solB)]-\gameval
    +\gameval-\max_{\purealt\notin\supp(\solB)}\vc[(\gamemat^{\top}\solA)][\purealt]}{2}
    \le
    \frac{\norm{\gamemat\solB}_{\infty} + \norm{\gamemat^{\top}\solA}_{\infty}}{2}
    \le 1.
\end{equation}

For any $\test{\minvar}\in\simplexA$, we denote by
\begin{equation}
    \notag
    \subsupp_{\test{\minvar}}=\setdef{\minvar\in\simplexA} {\supp(\minvar)\subset\supp(\test{\minvar})}.
\end{equation}
the set of the points whose support is included in that of $\test{\minvar}$.
For $\test{\maxvar}\in\simplexB$, $\subsupp_{\test{\maxvar}}$ is defined in the same way.
The next lemma, extracted from \cite{WLZL21}, is crucial to our proof.
\begin{lemma}
\label{lem:mix-still-nash}
Let $\test{\point}=(\test{\minvar},\test{\maxvar})\in\simplex_{\nPuresA}\times\simplex_{\nPuresB}$ satisfy %that $\supp(\solA)\in\supp(\test{\minvar})$, $\supp(\solB)\in\supp(\test{\maxvar})$, and 
that for all
$(\minvar,\maxvar)\in\subsupp_{\solA}\times\subsupp_{\solB}$,
\begin{equation}
    \label{eq:two-player-supp-opt}
    (\minvar-\test{\minvar})^{\top}\gamemat\test{\maxvar}
    + \test{\minvar}^{\top}\gamemat(\test{\maxvar}-\maxvar)\ge0.
\end{equation}
Then $\pointalt=(1-\gamediff/2)\sol+(\gamediff/2)\test{\point}$ is also a Nash equilibrium.
\end{lemma}
\begin{proof}
We rewrite the \acl{LHS} of \eqref{eq:two-player-supp-opt} as
\begin{equation}
    \label{eq:two-player-supp-opt-rewrite}
    (\minvar-\test{\minvar})^{\top}\gamemat\test{\maxvar}
    + \test{\minvar}^{\top}\gamemat(\test{\maxvar}-\maxvar)
    = \minvar^{\top}\gamemat\test{\maxvar} - \gameval
    + \gameval - \test{\minvar}^{\top}\gamemat\maxvar
    = \minvar^{\top}\gamemat(\test{\maxvar}-\solB)
    + (\solA-\test{\minvar})^{\top}\gamemat\maxvar.
\end{equation}
The second inequality holds because $(\minvar,\maxvar)\in\subsupp_{\solA}\times\subsupp_{\solB}$.
With the choice $(\minvar,\maxvar)\subs(\solA,\solB)$ and \eqref{eq:two-player-supp-opt} we then get
\begin{equation}
    \notag
    \solA^{\top}\gamemat(\test{\maxvar}-\solB)
    + (\solA-\test{\minvar})^{\top}\gamemat\solB \ge 0.
\end{equation}
This implies 
\begin{equation}
    \label{eq:2p-nash-equal}
    \solA^{\top}\gamemat(\test{\maxvar}-\solB)=(\solA-\test{\minvar})^{\top}\gamemat\solB=0
\end{equation}
by the definition of Nash equilibrium \eqref{eq:saddle}.

We next prove that $(\solA,\alt{\maxvar})$ is also a Nash equilibrium with $\alt{\maxvar}=(1-\gamediff/2)\solB+(\gamediff/2)\test{\maxvar}$.
From \eqref{eq:2p-nash-equal} we already have
\begin{equation}
    \notag
    \solA^{\top}\gamemat\alt{\maxvar}=\solA^{\top}\gamemat\solB=\gameval=\max_{\maxvar\in\simplexB}\solA^{\top}\gamemat\maxvar.
\end{equation}
It remains to show that $\solA^{\top}\gamemat\alt{\maxvar}=\min_{\minvar\in\simplexA}\minvar^{\top}\gamemat\alt{\maxvar}$.
By choosing $\maxvar=\solB$ in \eqref{eq:two-player-supp-opt-rewrite}, we know that for all $\minvar\in\subsupp_{\solA}$, it holds $\minvar^{\top}\gamemat(\test{\maxvar}-\solB)\ge0$.
In other words,
\begin{equation}
    \label{eq:coordinate-non-negative}
    \forall \pure\in\supp(\solA),~~
    \vc[(\gamemat(\test{\maxvar}-\solB))]\ge0
\end{equation}
Let $\minvar\in\simplexA$. We decompose
\begin{equation}
    \label{eq:game-product-decompose}
    \minvar^{\top}\gamemat\alt{\maxvar}
    = \sum_{\pure\in\supp(\solA)}\vc[\minvar]\vc[(\gamemat\alt{\maxvar})]
    + \sum_{\pure\notin\supp(\solA)}\vc[\minvar]\vc[(\gamemat\alt{\maxvar})].
\end{equation}
The first term can be bounded below using \eqref{eq:coordinate-non-negative},
\begin{equation}
    \label{eq:game-product-lower-supp}
    \sum_{\pure\in\supp(\solA)}\vc[\minvar]\vc[(\gamemat\alt{\maxvar})]
    = \sum_{\pure\in\supp(\solA)}
    \left(\frac{\gamediff}{2}\,\vc[\minvar]\vc[(\gamemat(\test{\maxvar}-\solB))]
    +\vc[\minvar]\vc[(\gamemat\solB)]\right)
    \ge \sum_{\pure\in\supp(\solA)} \vc[\minvar]\gameval.
\end{equation}
We proceed to lower bound the second term
\begin{equation}
    \label{eq:game-product-lower-nonsupp}
    \begin{aligned}[b]
    \sum_{\pure\notin\supp(\solA)}\vc[\minvar]\vc[(\gamemat\alt{\maxvar})]
    &\ge \sum_{\pure\notin\supp(\solA)}
    \left(
    \vc[\minvar]\vc[(\gamemat\solB)]
    -\frac{\gamediff}{2}\,\abs{\vc[\minvar]\vc[(\gamemat(\test{\maxvar}-\solB))]}
    \right)\\
    &\ge \sum_{\pure\notin\supp(\solA)}
    \left(
    \vc[\minvar]\vc[(\gamemat\solB)]
    -\frac{\gamediff}{2}\,\vc[\minvar]\norm{\gamemat}_{\infty}\norm{\test{\maxvar}-\solB}_1
    \right)\\
    &\ge \sum_{\pure\notin\supp(\solA)}
    \vc[\minvar](\vc[(\gamemat\solB)]-\gamediff)\\
    &\ge \sum_{\pure\notin\supp(\solA)} \vc[\minvar]\gameval.
    \end{aligned}
\end{equation}
In the last inequality we use the definition of $\gamediff$.
Combining \eqref{eq:game-product-decompose}, \eqref{eq:game-product-lower-supp}, and \eqref{eq:game-product-lower-nonsupp} we have $\minvar^{\top}\gamemat\alt{\maxvar}\ge\gameval=\solA^{\top}\gamemat\alt{\maxvar}$.
We have therefore proved that $(\solA,\alt{\maxvar})$ is a Nash equilibrium.
In the same way we can show that with $\alt{\minvar}=(1-\gamediff/2)\solA+(\gamediff/2)\test{\minvar}$,
the point $(\alt{\minvar},\solB)$ is also a Nash equilibrium.
We then conclude that $\alt{\action}=(\alt{\minvar},\alt{\maxvar})$ is indeed a Nash equilibrium.
\end{proof}

In the following we analyse the case where both $\vw[\hreg][1]$ and $\vw[\hreg][2]$ are negative entropy regularizers (\ie both players play adaptive OMWU). The case where one is negative entropy regularizer and the other satisfies that $\vw[\points]\subset\dom\subd\vw[\hreg]$ can be proved similarly.
%the condition that $\grad\hreg$ is defined on the whole simplex can be proved similarly.
The Bregman divergence for the negative entropy regularizer is the KL divergence which we will denote by $\dkl$.
We take $\grad\vw[\hreg][1]\from(\vc[\minvar])_{\pure\in\oneto{\nPuresA}}\to(-\log \vc[\minvar])_{\pure\in\oneto{\nPuresA}}$
and
$\grad\vw[\hreg][2]\from(\vc[\maxvar][\purealt])_{\purealt\in\oneto{\nPuresA}}\to(-\log \vc[\maxvar][\purealt])_{\purealt\in\oneto{\nPuresB}}$.

\AdaptiveOMWU*

\begin{proof}
Consider the solution $\sol=(\solA,\solB)$ that we have chosen using \cref{lem:nash-all-essential}.
By \cref{lem:bregman-cvg} we know that $\vw[\regpar][1]\dkl(\solA,\vt[\minvar])+\vw[\regpar][2]\dkl(\solB,\vt[\maxvar])$ are bounded above.
This implies that for all $\pure\in\supp(\solA)$ and $\purealt\in\supp(\solB)$, the coordinates $\vct[\minvar]$ and $\vct[\maxvar][\purealt]$ are bounded below.
In particular, for any cluster point $\limp{\action}=(\limpA,\limpB)$, we have $\supp(\solA)\subset\supp(\limpA)$ and $\supp(\solB)\subset\supp(\limpB)$.

We will proceed to prove the sequence of produced iterates only has one cluster point.
We first use the optimality condition \eqref{eq:optimality-before-limit} but only apply it to $\vw[\arpoint][1]\subs\minvar\in\subsupp(\solA)$. This gives
\begin{equation}
    \label{eq:optimality-before-limit-support}
    \sum_{\pure\in\supp(\solA)}
    (\vc[\vw[\vecfield][1](\past[\jstate])]+\vwt[\regpar][1](\log(\vct[\minvar])-\log(\vctinter[\minvar])))
    (\vc[\minvar]-\vctinter[\minvar])\ge0
\end{equation}
We consider a subsequence that goes to a cluster point $\limp{\point}=(\limpA,\limpB)$. Since $\vc[\limpA]>0$ for all $\pure\in\supp(\solA)$ and both $\norm{\vt[\jstate]-\past[\jstate]}$ and $\norm{\inter[\jstate]-\current[\jstate]}$ go to zero (\cref{lem:diff-to-zero}), \eqref{eq:optimality-before-limit-support} implies
\begin{equation}
    \notag
    \sum_{\pure\in\supp(\solA)}
    \vc[\vw[\vecfield][1](\limp{\jaction})]
    (\vc[\minvar]-\vc[\limpA])\ge0.
\end{equation}
Equivalently, $(\minvar-\limpA)^{\top}\gamemat\limpB\ge0$.
In the same way, for all $\maxvar\in\subsupp(\solB)$ we have $\limpA^{\top}\gamemat(\limpB-\maxvar)\ge0$.
We can thus apply \cref{lem:mix-still-nash} and we know that $(\alt{\solA},\alt{\solB}) = (1-\gamediff/2)\sol+(\gamediff/2)\limp{\point}$ is also a Nash equilibrium.
By the choice of $\sol$, we have $\supp(\alt{\solA})\subset\supp(\solA)$ and $\supp(\alt{\solB})\subset\supp(\solB)$.
Subsequently, $\supp(\solA)=\supp(\limpA)$ and $\supp(\solB)=\supp(\limpB)$.

Using \cref{lem:bregman-cvg}, we can define 
\begin{equation}
    \notag
    \begin{aligned}
    \limdist&=\lim_{\run\to+\infty}\vw[\regpar][1]\dkl(\solA,\vt[\minvar])+\vw[\regpar][2]\dkl(\solB,\vt[\maxvar])\\
    \alt{\limdist}&=\lim_{\run\to+\infty}\vw[\regpar][2]\dkl(\alt{\solA},\vt[\minvar])+\vw[\regpar][2]\dkl(\alt{\solB},\vt[\maxvar]).
    \end{aligned}
\end{equation}
Since $\supp(\solA)=\supp(\limpA)$ and $\supp(\solB)=\supp(\limpB)$, we can use the continuity of the KL divergence with respect to the second variable and deduce that $\vw[\regpar][1]\dkl(\solA,\limpA)+\vw[\regpar][2]\dkl(\solB,\limpB)=\limdist$.
Similarly, $\vw[\regpar][1]\dkl(\alt{\solA},\limpA)+\vw[\regpar][2]\dkl(\alt{\solB},\limpB)=\alt{\limdist}$.
These two equations also hold if we consider another cluster point $\alt{\limp{\point}}=(\alt{\limpA},\alt{\limpB})$.
As a consequence,
\begin{equation}
    \label{eq:KL-equal-1}
    \begin{multlined}
    \vw[\regpar][1]\sum_{\pure\in\supp(\solA)}\vc[\solA]\log\vc[\limpA]
    +\vw[\regpar][2]\sum_{\purealt\in\supp(\solB)}\vc[\solB][\purealt]\log\vc[\limpB][\purealt]\\
    =
    \vw[\regpar][1]\sum_{\pure\in\supp(\solA)}\vc[\solA]\log\vc[\alt{\limpA}]
    +\vw[\regpar][2]\sum_{\purealt\in\supp(\solB)}\vc[\solB][\purealt]\log\vc[\alt{\limpB}][\purealt],
    \end{multlined}
\end{equation}
and
\begin{equation}
    \label{eq:KL-equal-2}
    \begin{multlined}
    \vw[\regpar][1]\sum_{\pure\in\supp(\solA)}\vc[\alt{\solA}]\log\vc[\limpA]
    +\vw[\regpar][2]\sum_{\purealt\in\supp(\solB)}\vc[\alt{\solB}][\purealt]\log\vc[\limpB][\purealt]\\
    =
    \vw[\regpar][1]\sum_{\pure\in\supp(\solA)}\vc[\alt{\solA}]\log\vc[\alt{\limpA}]
    +\vw[\regpar][2]\sum_{\purealt\in\supp(\solB)}\vc[\alt{\solB}][\purealt]\log\vc[\alt{\limpB}][\purealt],
    \end{multlined}
\end{equation}
With $(\alt{\solA},\alt{\solB}) = (1-\gamediff/2)\sol+(\gamediff/2)\limp{\point}$ and $\gamediff>0$, using \eqref{eq:KL-equal-1} and \eqref{eq:KL-equal-2} we get
\begin{equation}
    \notag
    \begin{multlined}[b]
    \vw[\regpar][1]\sum_{\pure\in\supp(\solA)}\vc[\limpA]\log\vc[\limpA]
    +\vw[\regpar][2]\sum_{\purealt\in\supp(\solB)}\vc[\limpB][\purealt]\log\vc[\limpB][\purealt]\\
    =
    \vw[\regpar][1]\sum_{\pure\in\supp(\solA)}\vc[\limpA]\log\vc[\alt{\limpA}]
    +\vw[\regpar][2]\sum_{\purealt\in\supp(\solB)}\vc[\limpB][\purealt]\log\vc[\alt{\limpB}][\purealt],
    \end{multlined}
\end{equation}
Note that we also have $\supp(\solA)=\supp(\alt{\limpA})$ and $\supp(\solB)=\supp(\alt{\limpB})$. The above is thus equivalent to
\begin{equation}
    \notag
    \vw[\regpar][1]\dkl(\limpA,\alt{\limpA})
    +\vw[\regpar][2]\dkl(\limpB,\alt{\limpB}) = 0
\end{equation}
This shows $\limp{\point}=\alt{\limp{\point}}$, and therefore $\seqinf[\vt[\jstate]]$ has only one cluster point; in other words, the algorithm converges (recall that $\lim_{\toinf}\norm{\inter[\jstate]-\current[\jstate]}=0$).
To conclude, we note that if a no regret learning algorithm converges, it must converge to a Nash equilibrium.
In fact, for all $\minvar\in\simplexA$, we have $\sum_{\run=1}^{\nRuns}\product{\vw[\vecfield][1](\inter[\jstate])}{\inter[\minvar]-\minvar}=\smalloh(\nRuns)$ and thus $\liminf_{\run\to+\infty}\product{\vw[\vecfield][1](\inter[\jstate])}{\inter[\minvar]-\minvar}\le0$.
However, $\lim_{\run\to+\infty}\product{\vw[\vecfield][1](\inter[\jstate])}{\inter[\minvar]-\minvar}=\product{\vw[\vecfield][1](\limp{\point})}{\limpA-\minvar}$.
This shows $\product{\vw[\vecfield][1](\limp{\point})}{\limpA-\minvar}\le0$ for all $\minvar\in\simplexA$ and thus $\limpA$ is a best response to $\limpB$.
The same argument also applies to the second player; accordingly, $\limp{\point}$ is indeed a Nash equilibrium.
\end{proof}

\subsection{A dichotomy result for general convex games}
\label{app:dichotomy}

Below we prove a variant of \cref{thm:dichotomy-general-sum} which does not require the compactness assumption. \cref{thm:dichotomy-general-sum} is a direct corollary of this variant.

{
\addtocounter{theorem}{-1}
\renewcommand{\thetheorem}{\ref{thm:dichotomy-general-sum}$'$}
\begin{theorem}
Suppose that \cref{asm:convexity+smoothness} holds and all players $\play\in\players$  adopt an adaptive optimistic learning strategy.
%Let \cref{asm:convexity+smoothness} holds and that all the players adopt either adaptive \ac{OptDA} or adaptive \ac{DS-OptMD}.
Assume additionally that $\vw[\points]\subset\dom\subd\vw[\hreg]$.
Then one of the following holds:
% Let \cref{asm:convexity+smoothness} holds and that all the players adopt an adaptive optimistic learning strategy. Assume additionally that $\vw[\points]\subset\dom\subd\vw[\hreg]$ for every $\play\in\players$. Then, either of the following holds:
\begin{enumerate}[\upshape(\itshape a\upshape),topsep=1.2pt,itemsep=1.2pt,leftmargin=*]
    \item For every $\play\in\players$ and every compact set $\vw[\arpoints]\in\vw[\points]$, the individual regret 
    $\vwt[\reg][\play][\nRuns](\vw[\arpoints])$ is bounded above
    %and the per-round regret $\vwt[\gap](\vw[\arpoints])$ goes to zero when $\toinf$,
    \ie $\vwt[\reg][\play][\nRuns](\vw[\arpoints])=\bigoh(1)$.
    %and $\lim_{\toinf}\vwt[\gap](\vw[\arpoints])=0$.
    Moreover, every cluster point of the realized actions is a Nash equilibrium of the game.
    \label{thm:dichotomy-general-sum-a}
    \item For every compact set $\arpoints\subset\points$, the social regret with respect it tends to minus infinity when $\toinf$, \ie $\lim_{\toinf}\vt[\reg][\nRuns](\arpoints)=-\infty$.
    \label{thm:dichotomy-general-sum-b}
\end{enumerate}
\end{theorem}
}

\begin{proof}
By Lipschitz continuity of $\vw[\vecfield]$, there exists $\vw[\lips]>0$ such that
\begin{equation}
    \notag
    \vwpdual[\norm{\vw[\vecfield](\inter[\jstate])-\vw[\vecfield](\past[\jstate])}]
    \le
    \vw[\lips]\norm{\inter[\jstate]-\past[\jstate]}
    \le
    \vw[\lips](\norm{\inter[\jstate]-\current[\jstate]}
    +\norm{\current[\jstate]-\past[\jstate]}).
\end{equation}
We set 
$\vwt[\alt{\increment}]= 2\vw[\lips]^2(\norm{\inter[\jstate]-\current[\jstate]}^2
+\norm{\current[\jstate]-\past[\jstate]}^2)$ for $\run\ge2$ so that $\vwt[\increment]\le\vwt[\alt{\increment}]$.
We also define $\vwt[\alt{\increment}][\play][\start]=\vwt[\increment][\play][\start]=\vwpdual[\norm{\vw[\vecfield](\vt[\jaction][\start])}^2]$, $\vw[\deccst]=1/(16\nPlayers\vw[\lips]^2)$, and $\vw[\inibound]=\max_{\vw[\arpoint]\in\vw[\arpoints]}\vw[\armeasure](\vw[\arpoint])$.
Then, from the regret bound \eqref{eq:template-regret}, similar to how \eqref{eq:OptDA-regret-refined} is derived, we deduce
\begin{equation}
    \notag
    \begin{aligned}
    \vt[\reg][\nRuns](\arpoints)
    &\le
    \sumplayer
    \left(\vwt[\regpar][\play][\nRuns+1]\vw[\inibound]
    +\vwpdual[\norm{\vw[\vecfield](\vt[\jaction][\start])}^2]\right)\\
    &~~+ \sum_{\run=2}^{\nRuns}
    \left(\sumplayer
    \frac{\vwt[\increment]}{\vwt[\regpar]}
    - \frac{1}{4}
    \left(\norm{\inter[\jstate]-\current[\jstate]}^2
    +\norm{\current[\jstate]-\past[\jstate]}^2\right)\right)\\
    &\le
    \sumplayer
    \left(
    \vw[\inibound]
    \sqrt{1+\sum_{\run=1}^{\nRuns}\vwt[\alt{\increment}]}
    +\vwt[\increment][\play][\start]
    -\vw[\deccst]
    \sum_{\run=2}^{\nRuns}\vwt[\alt{\increment}]
    \right)
    + \sum_{\run=2}^{\nRuns}\sumplayer
    \left(
    \frac{\vwt[\increment]}{\vwt[\regpar]}
    - \vw[\deccst]\vwt[\increment]\right).
    \end{aligned}
\end{equation}

Following the reasoning of the proof of \cref{thm:social-regret-bounded}, we know there exists a constant $\Cst$ such that for all $\nRuns\in\N$,
\begin{equation}
    \notag
    \vt[\reg][\nRuns](\arpoints)\le\Cst+\vw[\func][1]\left(\sqrt{\sum_{\run=2}^{\nRuns}\vwt[\alt{\increment}][1]}\right),
\end{equation}
where $\vw[\func][1]\from\scalar\in\R\mapsto -\vw[\deccst][1]\scalar^2 + \vw[\inibound][1]\scalar$ is quadratic and has negative leading coefficient.
Therefore, $\vt[\reg][\nRuns](\arpoints)\to-\infty$ when $\sum_{\run=1}^{+\infty}\vwt[\alt{\increment}][1]=+\infty$, and this corresponds to the situation \ref{thm:dichotomy-general-sum-b}.

Otherwise, $\sum_{\run=1}^{+\infty}(\norm{\inter[\jstate]-\current[\jstate]}^2+\norm{\current[\jstate]-\past[\jstate]}^2)<+\infty$
and this implies 
\begin{enumerate}[\itshape i\upshape), itemsep=0pt]
    \item $\lim_{\toinf}\norm{\inter[\jstate]-\current[\jstate]}^2=0$;
    \item $\lim_{\toinf}\norm{\current[\jstate]-\past[\jstate]}^2=0$;
    \item for all $\play\in\players$, $\sum_{\run=1}^{+\infty}\vwt[\increment]<+\infty$ and hence $\vw[\regpar]=\lim_{\toinf}\vwt[\regpar]<+\infty$.
\end{enumerate}

To conclude, we prove the boundedness of individual regrets as in the proof of \cref{thm:inidividual-regret-bound} and that every cluster point of $\seqinf[\inter[\jstate]]$ is a Nash equilibrium as in the proof of \cref{thm:converge-to-Nash} (case \emph{a}).
%Let us fix $\play\in\players$. Since $\vw[\arpoints]$ is a compact, we can define 
%$\vw[\radius]=\max_{\vw[\arpoint]\in\vw[\arpoints]}\vwpdual[\norm{}]$
% Let $\vw[\arpoint]\in\vw[\arpoints]$. By using the optimality condition at $\vwtinter[\state]$ \eqref{eq:optimality-before-limit}, we have
% %
% \begin{equation}
%     \notag
%     \begin{aligned}
%     &\product{\vw[\vecfield](\inter[\jstate])}{\vwtinter[\state]-\vw[\arpoint]}\\
%     &~~~\le
%     \product{
%     \vw[\vecfield](\inter[\jstate])
%     - \vw[\vecfield](\past[\jstate])
%     + \vwt[\regpar]\grad\vw[\hreg](\vwt[\state])
%     - \vwt[\regpar]\grad\vw[\hreg](\vwtinter[\state])}
%     {\vwtinter[\state]-\vw[\arpoint]}\\
%     &~~~\le
%     \radius\vwpdual[\norm{
%     \vw[\vecfield](\inter[\jstate])
%     - \vw[\vecfield](\past[\jstate])}^2]
%     + \vw[\regpar]\radius
%     \vwpdual[\norm{
%     \grad\vw[\hreg](\vwt[\state])
%     - \grad\vw[\hreg](\vwtinter[\state])}^2].
%     \end{aligned}
% \end{equation}
\end{proof}

% \begin{corollary}
% \label{cor:dichotomy-general-sum}
% Let \cref{asm:convexity,asm:smooth} hold and that all the players adopt an adaptive optimistic learning strategy. Assume additionally that $\vw[\points]\subset\dom\grad\vw[\hreg]$ for every $\play\in\players$ and $\points$ is compact. Then, either of the following holds:
% \begin{enumerate}[\upshape(\itshape a\upshape),topsep=1.2pt,itemsep=1.2pt,leftmargin=*]
%     \item The sequence of realized actions converges to the set of Nash equilibria, \ie $\lim_{\toinf}\dist(\vt[\jaction],\sols)=0$.
%     Furthermore, for every $\play\in\players$, it holds $\vwt[\reg][\play][\nRuns](\vw[\points])=\bigoh(1)$ and $\limsup_{\run\to+\infty}\vwsp[\gap_{\vw[\arpoints]}](\vt[\jaction])\le0$. 
%     \item The social regret tends to minus infinity when $\toinf$, \ie $\lim_{\toinf}\vt[\reg][\nRuns](\points)=-\infty$.
% \end{enumerate}
% \end{corollary}

\section{Technical lemmas for numerical sequences}
In this appendix we provide two basic lemmas for numerical sequences, one for bounding the adversarial regret of adaptive methods \cite[Lemma 3.5]{ACG02}, and the other for the analysis of quasi-Fejér sequence \cite[Lemma 3.1]{Com01}.

\begin{lemma}
\label{lem:adaptive}
For any real numbers $\seq{\scalar}{1}{\nRuns}$ such that $\sum_{\runalt=1}^{\run}\scalar_{\runalt}>0$ for all $\run\in\oneto{\nRuns}$, it holds
\begin{equation}
    \notag
    \sum_{\run=1}^{\nRuns}
    \frac{\scalar_{\run}}{\sqrt{\sum_{\runalt=1}^{\run}\scalar_{\runalt}}}
    \le 2 \sqrt{\sum_{\run=1}^{\nRuns}\scalar_{\run}}.
\end{equation}
\end{lemma}
\begin{proof}
The function $y\in\R^+\mapsto\sqrt{y}$ being concave and has derivative $y\mapsto1/(2\sqrt{y})$, it holds for every $z\ge0$,
\begin{equation}
    \notag
    \sqrt{z}\le\sqrt{y}+\frac{1}{2\sqrt{y}}(z-y).
\end{equation}
Take $y=\sum_{\runalt=1}^{\run}\vt[\scalar][\runalt]$ and $z=\sum_{\runalt=1}^{\run-1}\vt[\scalar][\runalt]$ gives
\begin{equation}
\notag
    2\sqrt{\sum_{\runalt=1}^{\run-1}\vt[\scalar][\runalt]}+\frac{\vt[\scalar][\run]}{\sqrt{\sum_{\runalt=1}^{\run}\vt[\scalar][\runalt]}}\le2\sqrt{\sum_{\runalt=1}^{\run}\vt[\scalar][\runalt]}.
\end{equation}
We conclude by summing the inequality from $\run=2$ to $\run=\nRuns$ and using $\sqrt{\vt[\scalar][1]}\le2\sqrt{\vt[\scalar][1]}$.
\end{proof}

\begin{lemma}
\label{lem:summable-cvg}
Let $\seqinf[\vt[\srv]]\in\R_+^\N$ be a non-negative sequence and $\seqinf[\vt[\srvp]]\in\R_+^\N$ be summable such that, for all $\run\in\N$,
\begin{equation}
    \label{eq:quasi-fejer}
    \update[\srv]\le\vt[\srv]+\vt[\srvp].
\end{equation}
Then, $\seqinf[\vt[\srv]]$ converges.
\end{lemma}
\begin{proof}
Since $\seqinf[\vt[\srvp]]$ is summable, we can define $\vt[\alt{\srv}]=\vt[\srv]+\sum_{\runalt=\run}^{+\infty}\vt[\srvp][\runalt]\in\R_+$. Inequality \eqref{eq:quasi-fejer} then implies $\update[\alt{\srv}]\le\vt[\alt{\srv}]$.
Therefore, $\seqinf[\vt[\alt{\srv}]]$ converges, and accordingly $\seqinf[\vt[\srv]]$ converges. %($\lim_{\toinf}\sum_{\runalt=\run}^{+\infty}\vt[\srvp][\runalt]=0$).
\end{proof}

\end{document}